\documentclass[12pt]{report}
\usepackage{amssymb,amsthm,mythesis,uwthesis}
\usepackage{mathtools}
\usepackage{ stmaryrd }
\usepackage{cite}
\usepackage{hyperref} 
\usepackage[utf8]{inputenc}
\usepackage{enumerate}
\usepackage{amsmath}
\usepackage{amssymb}

\begin{document}

\title{Algorithmic Randomness and Kolmogorov Complexity for Qubits}
\author{Tejas Bhojraj}

\profA{Professor Joseph Miller, Professor, Mathematics, UW-Madison}
\profB{Professor Andr\'e Nies, Professor, Computer Sciences, University of Auckland}
\profC{Professor Uri Andrews, Associate Professor, Mathematics, UW-Madison}
\profD{Professor Jin-Yi Cai, Professor, Computer Sciences, UW-Madison}
\degree{Doctor of Philosophy}
\dept{Mathematics}
\thesistype{dissertation}
\beforepreface
\prefacesection{Abstract}

This work extends the theories of algorithmic randomness and Kolmogorov complexity of bitstrings to the quantum realm. In addition, it describes a method to generate an arithmetically random infinite bitstring from a certain computable, non-quantum-random infinite qubitstring.

Nies and Scholz defined quantum Martin-L{\"o}f randomness (q-MLR): the first notion of algorithmic randomness to be defined for qubitstrings. We define a notion of quantum Solovay randomness and show it to be equivalent to q-MLR using purely linear algebraic methods. Quantum Schnorr randomness is then introduced. A quantum analogue of the law of large numbers is shown to hold for quantum Schnorr random states. 

We now turn to a quantum analogue of Kolmogorov complexity. We introduce quantum-K ($QK$), a measure of the descriptive complexity of density matrices using classical prefix-free Turing machines and show that the initial segments of weak Solovay random and quantum Schnorr random states are incompressible in the sense of $QK$. Many properties enjoyed by prefix-free Kolmogorov complexity ($K$) have analogous versions for $QK$; notably a counting condition. Several connections between Solovay randomness and $K$, including the Chaitin type characterization of Solovay randomness, carry over to those between weak Solovay randomness and $QK$. Schnorr randomness has a Levin\textendash Schnorr characterization using $K_C$; a version of $K$ using a computable measure machine, $C$. We similarly define $QK_C$, a version of $QK$. Quantum Schnorr randomness is shown to have a Levin\textendash Schnorr and a Chaitin type characterization using $QK_C$. 

We then explore a notion of `measuring' a state. We formalize how `measurement' of a state induces a probability measure on the space of infinite bitstrings. A state is `measurement random' (mR) if the measure induced by it, under any computable basis, assigns probability one to the set of Martin-L{\"o}f randoms. While quantum-Martin-L{\"o}f random states are mR, the converse fails: there is a mR state, $\rho$ which is not quantum-Martin-L{\"o}f random. In fact, something stronger is true. While $\rho$ is computable and can be easily constructed, measuring it in any computable basis yields an arithmetically random sequence with probability one. 

The work concludes by studying the asymptotic von Neumann entropy of computable states. 

\prefacesection{Acknowledgements}
I thank my advisor Joe Miller for his unfailing encouragement and for his open-mindedness in allowing me to choose for my thesis work, a topic which was initially unfamiliar not just to him but also to me (quantum information). I am grateful to Andr\'e Nies for his advice and support. His paper \cite{unpublished} was responsible for introducing me to the theme dealt with in this work.

I thank Joe Miller, Steffen Lempp and Uri Andrews for fostering my interest in computability theory and logic during my undergraduate years.
I thank my teachers (too many to name individually) at UW-Madison for teaching me ideas and techniques which I am sure will have an enduring influence on my future work.

%

\afterpreface
%

\chapter{Introduction}\label{Introduction}   
%
%
%
%
%
%

Quantum physics describes a physical system by a unit vector in an appropriate vector space. Although the vector space can be infinite dimensional in general, this thesis deals purely with finite dimensional spaces. The simplest setting is that of a two dimensional vector space: a qubit is a unit vector in $\mathbb{C}^2$ and describes a two dimensional quantum system. Consider the orthonormal basis of $\mathbb{C}^2$ comprised of the unit eigenvectors of the $z-$operator denoted (in the usual bra-ket notation) by $|0\big >$ and $|1\big >$. An arbitrary qubit has the form $\alpha |0\big> + \beta |1\big>$ where $|\alpha|^2 + |\beta|^2 = 1$. The $|0\big>$ and $|1\big>$ are the quantum analogues of the classical $0$ and $1$ respectively. While a bit can only take on two possible values ($0$ or $1$), a qubit can be any unit length \emph{linear combination} of the basis vectors $|0\big>$ and $|1\big>$. So, a qubit generalizes the classical bit. This suggests that notions concerning classical bits can be extrapolated to qubits.

Section \ref{qm}, which may be skipped by the reader familiar with quantum theory, reviews some quantum theory background relevant to this work.

Information theory has been generalized to the quantum realm \cite{Nielsen:2011:QCQ:1972505}. Similarly, the theory of computation has been extended to the quantum setting; a notable example being the conception of a quantum Turing machine \cite{Mller2007QuantumKC,doi:10.1137/S0097539796300921}. It hence seems natural to extend algorithmic randomness, a discipline using concepts from computation and information,  to the quantum realm. Algorithmic randomness studies the randomness of infinite bitstrings using two main tools: (1) effective measure theory and (2) Kolmogorov complexity. While classical Kolmogorov complexity has inspired many competing definitions of quantum Kolmogorov complexity \cite{Berthiaume:2001:QKC:2942985.2943376,Mller2007QuantumKC,Vitnyi2001QuantumKC}, effective measure theory has only recently been extended to the quantum setting \cite{unpublished, bhojraj2020quantum}.

What does algorithmic randomness study? Consider infinite sequences of ones and zeroes (called  bitstrings in this paper). First consider the bitstring $101010101010\cdots$. It has an easily describable `pattern' to it; namely that the ones and zeroes alternate. Now take a bitstring obtained by tossing a fair coin repeatedly. Intuitively, it seems that the second bitstring, in contrast to the first, is unlikely to have patterns. Algorithmic randomness tries to quantify our intuition that the second bitstring is more `random', more `structureless' than the first. For this, it uses two central concepts: (1) An effectively null set and (2) Kolmogorov complexity. Roughly speaking an `effectively null set' is one which can be approximated by a computable sequence of open sets whose measures tend to zero in a nice way. Varying the precise definition of `effectively null' yields various randomness notions such as for example, Martin-L{\"o}f randomness, Solovay randomness and Schnorr randomness (See \cite{misc} and \cite{misc1} for more details on effective measure theory and its use in algorithmic randomness). The Kolmogorov theoretical approach quantifies the randomness of infinite bitstrings by measuring the incompressibility of their finite initial segments. Roughly speaking, a finite bitstring $\sigma$ is `incompressible' if $K(\sigma)$ is large ($K$ stands for the prefix-free Kolmogorov complexity. See \cite{misc} and \cite{misc1} for an exposition on $K$ and its properties). In the Kolmogorov theoretical approach, an infinite bitstring $X$ is `random' if its finite initial segments are asymptotically incompressible as $n$ goes to infinity. I.e., if $X \upharpoonright n$ is the first $n$ bits of $X$, then $K(X \upharpoonright n)$ grows `fast' as $n$ tends to infinity. Varying the precise definition of `fast' yields different randomness notions of varying strengths. 

It turns out that the randomness of an infinite bitstring measured via the effective measure theory approach is intimately related to its randomness measured in terms of their initial segment incompressibility.

While algorithmic randomness is concerned with the randomness of bitstrings, the present thesis is concerned with \emph{quantum} algorithmic randomness: the study of the randomness of \emph{qubitstrings} (infinite sequences of qubits), also called \emph{states}\cite{unpublished, bhojraj2020quantum}.

Chapter 2 extends the classical effective measure theory approach to the quantum realm. It studies the quantum analogues of Martin-L{\"o}f, Solovay and Schnorr randomness, which are defined using effectively null sets in the classical theory. Chapter 3 extends the theory of classical Kolmogorov complexity to the quantum setting. It introduces quantum-$K$ $(QK)$, a quantum version of $K$, and relates it to the three quantum randomness notions defined in Chapter 2. The remaining two chapters explore interesting applications of the main theory developed in chapters 2 and 3. In Chapter 4, we construct a computable state which is \emph{not} quantum Martin-L{\"o}f random but which yields an arithmetically random bitstring with probability one when `measured' (the notion of measuring a state is also defined in chapter 4). Arithmetic randomness is a strong form of classical randomness, strictly stronger than 
Martin-L{\"o}f randomness (See 6.8.4 in \cite{misc1} for details on arithmetic randomness). This suggests that it is easier to generate classical randomness than it is to generate quantum randomness.

While there have been several protocols for generating a random sequence of bits from a quantum source\cite{Pironio_2010, Baumeler2017KolmogorovAF, Kessler_2020, abbott2014value, DBLP:phd/hal/Abbott15}, to the best of our knowledge, none of the currently known protocols produce as strong a form of randomness as arithmetic randomness.

The final chapter explores the von-Neumann entropies of the finite initial segments of computable states.
Section \ref{qm} reviews some background from quantum theory and the following sections give an overview of each chapter.

\section{Quantum theory background}
\label{qm}
A more detailed account may be found in the textbook by Nielsen and Chuang \cite{Nielsen:2011:QCQ:1972505}. We assume the reader to be familiar with the bra-ket notation. A n-dimensional system is described by $|\psi \big >$,  a unit vector in $ \mathbb{C}^n$. A physical quantity corresponds to a Hermitian operator H on  $\mathbb{C}^n$. By Hermicity, H has a spectral decomposition: \[H = \sum_{i\leq n} e_{i} |i\big> \big< i|\]
where, $(|i\big>)_{i\leq n}$ is the complete orthonormal set of eigenvectors of H. Measuring H on  $|\psi\big >$ produces outcome $e_i$ with probabilty $|\big <i|\psi\big >|^{2}=Tr ( |\psi \big >\big <\psi|  |i \big >\big <i|) $. So, the outcome is non-deterministic except when the $|\psi\big >$ is an eigenvector of H. The only possible outcomes of measurement of $H$ are it's eigenvalues. If the outcome is $e_i$, the post-measurement system is in $|i \big >$. As usual, we denote the eigenvectors of the z-operator (a operator on $\mathbb{C}^2$) by $|1\big>$ and $|0\big>$. Any $|\psi \big> = \alpha  |1\big>+ \beta |0\big>$ with  $\alpha, \beta\in \mathbb{C}$ and  $|\alpha|^{2}+ |\beta|^{2} =1$ is said to be a \emph{qubit}. A sequence of n qubits is modeled by a unit vector in $(\mathbb{C}^{2})^{\otimes n}$ which has an orthonormal basis comprised of elements of the form \[ \bigotimes_{i<n} |\sigma(i) \big > := |\sigma \big > \text{ for a }   \sigma \in \{0,1\}^n\]  
States which are not pure tensors are said to be \emph{entangled}. If $|\psi \big > \in (\mathbb{C}^{2})^{\otimes n}$ is entangled, its subsystem in $(\mathbb{C}^{2})^{\otimes k}$ for some $k<n$ is not a single quantum state but rather is a probabilistic mixture of multiple quantum states. To describe such subsystems, we reformulate the above in the density matrix language by replacing $(\mathbb{C}^{2})^{\otimes n}$ with $L_n$, the space of linear operators on $(\mathbb{C}^{2})^{\otimes n}$ and by replacing $|\psi \big >$ with $|\psi \big >\big <\psi| $. A positive semidefinite matrix $\rho \in L_n$ is a density matrix if Tr($\rho$) = 1. By Hermicity, $\rho$ has a complete orthonormal set of eigenvectors  $(\psi_i)_{i<2^{n}}$. So, it is unitarily diagonalizable and has eigenpairs $(\alpha_i, \psi_i)$   
    \begin{align}
    \label{denm}
    \rho = \sum_{i< 2^{n}} \alpha_{i} |\psi_i\big> \big< \psi_i|    
    \end{align} This sum must be convex as, 1=Tr($\rho$)=$\sum_{i}\alpha_{i} $. A density matrix $\rho \in L_n$ is said to be a \emph{strictly mixed} state if \ref{denm} is a strictly convex sum and is said to be a \emph{pure} state if  $\rho= |\psi \big >\big <\psi|$ for some unit vector $\psi$. A density matrix which may be pure or strictly mixed is simply referred to as a \emph{mixed} state. In the density matrix language, a system $|\psi\big>$ is represented by the pure state $|\psi \big> \big< \psi|$. A system which is in $|\psi_i\big> \big< \psi_i|$ with probability $\alpha_{i}$ is described by the mixed state $\rho = \sum_{i< 2^{n}} \alpha_{i} |\psi_i\big> \big< \psi_i|$. Measuring H on  $\rho$ produces outcome $e_i$ with probabilty $|\big<i|\rho| i\big>|^{2}=Tr (\rho |i \big >\big <i|) $. The expected value of measuring H on $\rho$ is
   \begin{align*}
    [H]_{\rho} := \sum_{i} e_{i} Tr (\rho |i \big >\big <i|) 
    =Tr(H \rho)
    \end{align*}
A system which is a composite of systems given by $\sigma \in L_n $ and $\tau \in L_k$ is described by $\rho=\sigma \otimes \tau \in L_{n+k}$.

\section{An Overview of Chapter Two}
This chapter concerns the generalization of effective measure theory as used in classical algorithmic randomness to the quantum world. The basic definitions from algorithmic randomness we state below may be found in books by Nies\cite{misc1} and Downey and Hirschfeldt\cite{misc}. Roughly speaking, a Martin-L{\"o}f random bitstring is one which has no algorithmically describable regularities. Slightly more rigorously, an infinite bitstring is said to be Martin-L{\"o}f random if it is not in any `effectively null' set. In the context of Martin-L{\"o}f randomness, a measurable set $A$ is effectively null if there is a Turing machine which computes a sequence of open sets, $(U_n)_n$ such that the measure of $U_n$ is at most $2^{-n}$ and $A \subseteq U_n$ for all $n$. By varying the definition of `effectively null', we get other notions of randomness like Solovay randomness and Schnorr randomness. Note that the randomness of a bitstring defined using this approach crucially depends on the notion of computability. In a broad sense, a bitstring is random if it has no `\emph{computably describable}' patterns. Consequently, the following notion is pertinent to us; a function, $f$ on the natural numbers is said to be \emph{computable} if there is a Turing machine, $\phi$ such that on input $n$, $\phi$ halts and outputs $f(n)$.
\begin{defn}
A sequence $(a_n)_{n\in \mathbb{N}}$ is said to be \emph{computable} if there is a computable function $f$, such that $f(n)=a_n$.

\end{defn}
The notion of a computable real number will come up when we discuss quantum Schnorr randomness. 
\begin{defn}
A real number $r$ is said to be \emph{computable} if there is a computable function $f$ such that for all $n$, $|f(n)-r|<2^{-n}$.
\end{defn}
We describe how classical algorithmic randomness generalizes to qubitstrings.
We refer the reader to the book by Nielsen and Chuang\cite{Nielsen:2011:QCQ:1972505} for preliminaries on quantum theory.

While it is clear what one means by a infinite sequence of bits, it is not immediately obvious how one would formalize the notion of an infinite sequence of qubits. To describe this, many authors have independently come up with the notion of a \emph{state} \cite{unpublished,article,brudno}. We will need the one given by Nies and Scholz \cite{unpublished}. Recall that a positive semidefinite matrix with trace equal to one is called a `density matrix' and is commonly used to represent a probabilistic mixture of pure quantum states (See \cite{Nielsen:2011:QCQ:1972505}).
\begin{defn}\cite{unpublished}
A \emph{state}, $\rho=(\rho_n)_{n\in \mathbb{N}}$ is an 
infinite sequence of density matrices such that $\rho_{n} \in \mathbb{C}^{2^{n} \times 2^{n}}$ and $\forall n$,  $PT_{\mathbb{C}^{2}}(\rho_n)=\rho_{n-1}$.
\end{defn}
The idea is that $\rho$ represents an infinite sequence of qubits whose first $n$ qubits are given by $\rho_n$. Here, $PT_{\mathbb{C}^{2}}$ denotes the partial trace which `traces out' the last qubit from $\mathbb{C}^{2^n}$. The definition requires $\rho$ to be coherent in the sense that for all $n$, $\rho_n$, when `restricted' via the partial trace to its first $n-1$ qubits, has the same measurement statistics as the state on $n-1$ qubits given by $\rho_{n-1}$.

\begin{defn}\cite{unpublished}
\label{def:tr}
Let $\tau=(\tau_n)_{n\in \mathbb{N}}$ be the state given by setting $\tau_n = \otimes_{i=1}^n I$ where $I$ is the two by two identity matrix.
\end{defn}

\begin{defn}\cite{unpublished} 
A \emph{special projection} is a hermitian projection matrix with complex algebraic entries.
\end{defn}
Since the complex algebraic numbers (roots of polynomials with rational coefficients) have a computable presentation, we may identify a special projection with a natural number and hence talk about computable sequences of special projections. Let $I$ denote the two by two identity matrix.
\begin{defn}\cite{unpublished}
\label{defn:sigclass}
A quantum $\Sigma_{1}^0$
set (or q-$\Sigma_{1}^0$
set for short) G is a computable
sequence of special projections $G=(p_{i})_{i\in \mathbb{N}}$ such that $p_i$ is $2^i$ by $2^i$ and range$(p_i \otimes I) \subseteq$ range $(p_{i+1})$ for all $i\in \mathbb{N}$. 

\end{defn}

While a $2^n$ by $2^n$ special projection may be thought of as a computable projective measurement on a system of $n$ qubits, a q-$\Sigma_{1}^0$
class corresponds to a computable sequence of projective measurements on longer and longer systems of qubits. We motivate the definition of a quantum $\Sigma_{1}^0$
set by relating it to the classical $\Sigma_{1}^0$
class. Let $2^{\omega}$, called Cantor space, denote the collection of infinite bitstrings, let $2^n$ denote the set of bit strings of length $n$, $2^{<\omega}= \bigcup_{n} 2^n$, and let $2^{\leq\omega}:=  2^{<\omega} \cup 2^{ \omega}$. Cantor space can be topologized by declaring the cylinders to be the basic open sets. If $\pi \in 2^n$ for some $n$, then the cylinder generated by $\pi$, denoted $\llbracket \pi \rrbracket$, is the set of all sequences extending $\pi$:\[\llbracket \pi \rrbracket=\{X \in 2^{\omega}: X\upharpoonright n = \pi\}.\]
If $C \subseteq 2^n$, let \[\llbracket C \rrbracket := \bigcup_{\pi \in C} \llbracket \pi \rrbracket,\] be the set of all $X \in 2^{\omega}$ such that the initial segment of $X$ of length $n$ is in $C$. One of the many equivalent ways of defining a $\Sigma_{1}^0$
class is as follows. 
\begin{defn}
\label{def:10}
A $\Sigma_{1}^0$ class $S \subseteq 2^{\omega}$ is any set of the form, \[S=\bigcup_{i\in \mathbb{N}} \llbracket A_{i} \rrbracket\] where
    \begin{enumerate}
        \item $A_{i} \subseteq 2^i, \forall i\in \mathbb{N}$ 
        \item
        The indices of $A_{i}$ form a computable sequence. (Being a finite set, each $A_i$ has a natural number coding it.)
        \item $\llbracket A_{i} \rrbracket\subseteq \llbracket A_{i+1} \rrbracket, \forall i\in \mathbb{N}$.
    \end{enumerate}
\end{defn}
Letting $\llbracket A_{i} \rrbracket:= S_i$, we write $S=(S_i )_i$. A $\Sigma_{1}^0$ class, S is coded (non-uniquely) by the index of the total computable function generating the sequence $(A_{i})_{ i\in \mathbb{N}}$ occurring in (2) in the definition of $S$. Hence, the notion of a computable sequence of $\Sigma_{1}^0$ classes makes sense. One sees that the special projections $q_i$ in the definition of the q-$\Sigma_{1}^0$ play the role of the $A_i$s which generate a the $\Sigma_{1}^0$ class, $S$. The following notion is a quantum analog of the Lebesgue measure of $S$ which equals $\lim_{n}(2^{-n}|A_n|)$, where $|.|$ refers to the cardinality. (The uniform measure on $2^{\omega}$ is the measure induced by letting the measure of $ \llbracket \tau \rrbracket$ be $2^{-|\tau|}$ for each $\tau \in 2^{<\omega}$. Here, $|\tau|:=n$ if $\tau \in 2^n$.)
\begin{defn}\cite{unpublished}
\label{def:trace}
If $G=(p_{n})_{n\in \mathbb{N}}$ is a q-$\Sigma_{1}^0$ class, define \emph{$\tau(G):=\lim_{n}(2^{-n}|q_n|)$} where, $|q_n|$ is the rank of $q_n$.
\end{defn}
 Informally, and somewhat inaccurately, a q-$\Sigma_{1}^0$ class,  $G=(p_{n})_{n\in \mathbb{N}}$, may be thought of as a projective measurement whose expected value, when `measured' on a state $\rho=(\rho_{n})_{n\in \mathbb{N}}$ is $\rho(G):=\lim_{n}$ Trace $(\rho_{n}p_n)$. In reality, a q-$\Sigma_{1}^0$ class, $G=(p_{n})_n$, is a sequence of projective measurements on larger and larger finite dimensional complex Hilbert spaces. This sequence can be used to `measure' a coherent sequence of density matrices (i.e., a state) the expected value of which is the limit of the Trace$(\rho_{n}p_n)$ (the expected value of measuring the $n^{th}$ `level').

\begin{defn}
A \emph{classical Martin-L{\"o}f test} (MLT) is a computable sequence, $(S_{m})_{m \in \mathbb{N} }$ of $\Sigma_{1}^0$ classes such that the Lebesgue measure of $S_m$ is less than or equal to $2^{-m}$ for all m.
\end{defn}
Its quantum generalization is:
\begin{defn}\cite{unpublished}
A \emph{quantum Martin-L{\"o}f test} (q-MLT) is a computable sequence, $(S_{m})_{m \in \mathbb{N}}$ of q-$\Sigma_{1}^0$ classes such that $\tau(S_m)$ is less than or equal to $2^{-m}$ for all m.
\end{defn}

\begin{defn}\cite{unpublished}
$\rho$ is \emph{q-MLR} if for any q-MLT $(S_{m})_{m \in \mathbb{N}}$, $\inf_{m \in \mathbb{N}}\rho(S_m)=0$.
\end{defn}
Roughly speaking, a state is q-MLR if it cannot be `detected by projective measurements of arbitrarily small rank'.
\begin{defn}\cite{unpublished}
$\rho$ is said to \emph{fail  the q-MLT $(S_{m})_{m \in \mathbb{N}}$, at order $\delta$}, if $\inf_{m \in \mathbb{N}}\rho(S_m)>\delta$. $\rho$ is said to \emph{pass  the q-MLT $(S_{m})_{m \in \mathbb{N}}$ at order $\delta$} if it does not fail it at $\delta$. 
\end{defn}
So, $\rho$ is q-MLR if it passes all q-MLTs at all $\delta>0$.
\begin{remark}
\label{rem:notation}
A few remarks on notation: By `bitstring', we mean a finite or infinite classical sequence of ones and zeroes. It will be clear from context whether the specific bitstring under discussion is finite or infinite. $2^n$ will denote the set of bitstrings of length $n$. Let $B^n$ denote the standard computational basis for $\mathbb{C}^{2^n}$. I.e., $B^n:= \{|\sigma\big> : \sigma \in 2^n\} $. If $S\subseteq 2^n$, let $P_{S} :=\sum_{\sigma \in S} |\sigma\big>\big<\sigma|$. `Tr' stands for trace. A sequence of q-$\Sigma_{1}^0$ classes will be indexed by the superscript. The subscript will index the sequence of special projections comprising a q-$\Sigma_{1}^0$. For example, $(S^{m})_{m \in \mathbb{N}}$ is a sequence of q-$\Sigma_{1}^0$ classes and  $S^{m}=(S^{m}_{n})_{m \in \mathbb{N}}$ is a class from the sequence. So, a sequence of q-$\Sigma_{1}^0$ classes can be thought of as a double sequence of special projections: $(S^{m}_{n})_{m,n \in \mathbb{N}}$. Lebesgue measure is denoted by $\mu$.
\end{remark}

In addition to continuing the investigation of quantum Martin-L{\"o}f randomness begun by Nies and Scholz \cite{unpublished}, we introduce and study quantum Solovay and quantum Schnorr randomness in Chapter 2.

\section{An Overview of Chapter Three}
As mentioned before, effective measure theory (using `effectively null sets') and Kolmogorov complexity theory (using descriptive complexity of initial segments) are two seemingly unrelated but equivalent approaches to study the randomness of bitstrings. Chapter 2 of this thesis and other works \cite{unpublished,qpl,bhojraj2020quantum} have generalized the first approach to the quantum realm. We work towards generalizing the second approach in Chapter 3 of which we give an overview here. 

The most basic definition from the classical theory is that of $K$: The prefix-free Kolmorogov Complexity $(K)$ of a finite bit string $\sigma$ is defined as \[K(\sigma)=\min\{|x|: \mathbb{U}(x)=\sigma\},\]
where the $x$s are finite bitstrings and $\mathbb{U}$ is the universal prefix-free Turing Machine (See \cite{misc1, misc} for detailed expositions).
Martin-L{\"o}f randomness (which is equivalent to Solovay randomness) and Schnorr randomness for infinite bitstrings, defined using the concept of `effective null sets', have characterizations in terms of initial segment prefix-free Kolmogorov complexity (denoted by $K$) \cite{misc,misc1, DBLP:series/txtcs/Calude02}; the initial segments of random infinite bitstrings are incompressible in the sense of $K$. 
 Two important characterizations show that the initial segments of Martin-L{\"o}f randoms (equivalently, of Solovay randoms) are asymptotically incompressible in the sense of $K$: the Chaitin characterization (See \cite{Chaitin:1987:ITR:24912.24913} and theorem 3.2.21 in \cite{misc}),
\[X  \text{ is Martin-L{\"o}f random} \iff \text{lim}_{n}  K(X\upharpoonright n)-n= \infty,\]
and the Levin\textendash Schnorr characterization (See theorem 3.2.9 in \cite{misc}),
\[X  \text{ is Martin-L{\"o}f random} \iff \exists c \forall n [ K(X\upharpoonright n)>n-c].\]

A prefix-free machine , $C$ is said to be a computable measure machine if the Lebesgue measure of its domain is a computable real number\cite{downey2004schnorr}. With this definition in hand, we define\cite{downey2004schnorr}, analogously to $K$,
\[K_C(\sigma)=\min\{|x|: C(x)=\sigma\}.\]

Schnorr randomness has a Levin\textendash Schnorr type characterization using $K_C$;

\[X  \text{ is Schnorr random} \iff \forall C \exists c \forall n [ K_C(X\upharpoonright n)>n-c].\]

Quantum Solovay randomness and quantum Schnorr randomness for states are defined in Chapter 2. Analogously to the classical situation, one may explore the connections between quantum Solovay randomness and quantum Schnorr randomness and the initial segment descriptive complexity of states.

Motivated by this, we asked whether there is a quantum analogue of $K$ which yields a characterization of quantum Solovay and quantum Schnorr randomness. We define $QK$, a complexity measure for density matrices based on prefix-free, classical Turing machines. The abbreviation $QK$ stands for `quantum-K', reflecting our intention of developing a quantum analogue of $K$, the classical prefix-free Kolmogorov complexity.

To the best of our knowledge, all notions of quantum Kolmogorov complexity developed so far, with one exception\cite{Vitnyi2001QuantumKC}, exclusively use machines which are not prefix-free (plain classical machines or quantum Turing machines) \cite{doi:10.1137/S0097539796300921, Berthiaume:2001:QKC:2942985.2943376,Mller2007QuantumKC}. $K_Q$, a notion developed in \cite{Vitnyi2001QuantumKC} uses a quantum Turing machine, $Q$ together with the classical prefix-free Kolmogorov complexity in its definition.

After introducing quantum-K, we show that the initial segments of weak Solovay random and quantum Schnorr random states are incompressible in the sense of $QK$. Many properties enjoyed by prefix-free Kolmogorov complexity ($K$) have analogous versions for $QK$; notably a counting condition.

Several connections between Solovay randomness and $K$, including the Chaitin type characterization of Solovay randomness, carry over to those between weak Solovay randomness and $QK$. We work towards a Levin\textendash Schnorr type characterization of weak Solovay randomness in terms of $QK$.

As mentioned above, Schnorr randomness has a Levin\textendash Schnorr characterization using $K_C$. We similarly define $QK_C$, a version of $QK$. Quantum Schnorr randomness is shown to have a  Levin\textendash Schnorr and a Chaitin type characterization using $QK_C$. The latter implies a Chaitin type characterization of classical Schnorr randomness using $K_C$. 
\section{An Overview of Chapter Four}
This chapter investigates the following question: Can a non-q-MLR, computable quantum source be used to generate a MLR sequence of bits? To make this question fully precise, we need to define what we mean by `generate'. To this end, we will formalize a notion of `measuring a state'. With this notion in hand, we will construct a computable non-q-MLR state which yields a MLR bitstring almost surely when measured.

Measuring a finite dimensional quantum system or a composite system of finitely many qubits is a pivotal concept in quantum information theory \cite{bookA}. It hence seems natural to consider defining a notion of `measuring' a state. Since measurement of a state yields a \emph{classical} infinite sequence of bits, it is interesting to explore the relation between the quantum algorithmic randomness of the measured state and the  classical algorithmic randomness of the resulting sequence.  Chapter 4 is motivated by these questions. 

We first formalize how `measurement' of a state in a basis induces a probability measure on Cantor space. A state is `measurement random' (mR) if the measure induced by it, under any computable basis, assigns probability one to the set of Martin-L{\"o}f randoms. Equivalently, a state is mR if and only if measuring it in any computable basis yields a Martin-L{\"o}f random with probability one. While quantum-Martin-L{\"o}f random states are mR, the converse fails: there is a mR state, $\rho$ which is not quantum-Martin-L{\"o}f random. In fact, something stronger is true. While $\rho$ is computable and can be easily constructed, measuring it in any computable basis yields an arithmetically random sequence with probability one. I.e., classical arithmetic randomness can be generated from a computable, non-quantum random sequence of qubits. 

\section{An Overview of Chapter Five}

As quantum Martin-L{\"o}f randomness is a notion of `randomness' for states, we don't expect computable states to be quantum Martin-L{\"o}f random. However, the tracial state, which is computable, is quantum Martin-L{\"o}f random. This rather surprising fact justifies a study of the computable quantum Martin-L{\"o}f randoms. The theme of this chapter is to use the von Neumann entropy as a measure of the randomness of computable states.

Recall that the von-Neumann entropy of a density matrix is the Shannon entropy of the distribution given by its eigenvalues (As a density matrix is positive semidefinite and has trace equal to one, its eigenvalues are real non-negative and sum to one. The eigenvalues hence form a probability distribution. See, for example \cite{Nielsen:2011:QCQ:1972505}). So, the von-Neumann entropy of a density matrix, $d$ reflects how `evenly spread out' its eigenvalues are. If the eigenmass of $d$ is `concentrated' at a few (relative to the dimension of $d$) eigenvectors then the von Neumann entropy of $d$ is low. Informally speaking, if a computable density matrix, $d$ has a low entropy, then the few  eigenvectors at which the eigenmass is concentrated can be used to construct a special projection `close' to $d$. Conversely, if the von Neumann entropy of $d$ is high, then one cannot construct such a special projection. In this chapter, we formalize this intuition and extend it from density matrices, $d$ to states, $\rho = (\rho_n )_n$. This extension from individual density matrices to states involves studying the limiting behavior of the von Neumann entropy of $\rho_n$ as $n$ goes to infinity.

Our results may be summarized by the following implications:
For any computable $\rho$, \[\exists c>0 \exists^{\infty} n H(\rho_n) > n-c \Rightarrow \rho \text { is q-MLR} \Rightarrow  H(\rho):=\text{lim}_n \dfrac{H(\rho_n )}{n} = 1.\]

Further, we also show that these implications do not reverse.
\\ 
%

\chapter{Notions of quantum algorithmic randomness}\label{Quantum algorithmic randomness}   
%
%
%
%
%
%

\section{Introduction}

This section has two major themes. First, it continues the study of quantum Martin-L{\"o}f randomness initiated by Nies and Scholz \cite{unpublished}. Second, we define quantum Solovay and quantum Schnorr randomness and prove results concerning these notions. Along with  Martin-L{\"o}f randomness, Solovay randomness and Schnorr randomness are important classical randomness notions. While Solovay randomness is equivalent to MLR, Schnorr randomness is strictly weaker. In Section \ref{sec:SolSchn}  we define quantum Solovay and quantum Schnorr randomness, show that quantum Solovay randomness is equivalent to q-MLR, show the convexity of the randomness classes in the space of states (answering open questions\cite{unpublished1, unpublished}), and obtain results regarding q-MLR states. The equivalence of quantum Solovay and quantum Martin-L{\"o}f randomness turns out to be a corollary of Theorem \ref{thm:LinA}, a linear algebraic result of independent interest concerning the approximation of density matrices by subspaces. This result, to the best of our knowledge, is novel and may prove useful in areas where approximations to density matrices are used; for example, quantum information and error correction,  quantum Kolmogorov complexity \cite{Mller2007QuantumKC,Berthiaume:2001:QKC:2942985.2943376} and quantum statistical mechanics.

In Section \ref{sec:measures}, we study states which are coherent sequences of diagonal density matrices. These states can be thought of as probability measures on Cantor space. Nies and Stephan\cite{unpublished2} defined  Martin-L{\"o}f absolutely continuity and Solovay randomness for diagonal states. We show that these two notions are the restrictions of q-MLR and quantum Solovay randomness to the space of diagonal states. We prove a result (Lemma \ref{lem:30}) about approximating a subspace of small rank by another one with a different orthonormal spanning set and of appropriately small rank. This result, novel as far as we know, may be applied to the important problem of approximating an entangled subspace (a subspace spanned by entangled pure states) by one spanned by product tensors \cite{demianowicz2019entanglement,book2}. We discuss how quantum randomness notions restrict to classical states (i.e., to infinite bitstrings) and note that quantum Schnorr randomness is strictly weaker that q-MLR, as in the classical case.

Nies and Tomamichel \cite{logicblog} showed that q-MLR states satisfy quantum versions of the law of large numbers and the Shannon\textendash McMillan\textendash Breiman theorem for i.i.d.\ Bernoulli measures. In Sections \ref{sec:lln} and \ref{sec:smb} we strengthen their results by showing that in fact, all quantum Schnorr random states (a set strictly containing the q-MLR states) satisfy these properties.

Many results in this chapter have significantly different proofs, which may be found in \cite{bhojraj2020quantum}.

\section{Notions of Quantum algorithmic randomness}
\label{sec:SolSchn}

\subsection{Solovay and Schnorr randomness}
An infinite bitstring $X$ is said to \emph{pass} the Martin-L{\"o}f test $(U_{n})_n$ if $X \notin \bigcap_{n}U_n$ and is said to be \emph{Martin-L{\"o}f random (MLR)} if it passes all Martin-L{\"o}f tests. A related randomness notion is Solovay randomness. A computable sequence of $\Sigma^{0}_1$ classes, $(S_{n})_n$ is a \emph{Solovay test} if $\sum_{n} \mu(S_{n})$, the sum of the Lebesgue measures is finite. An infinite bitstring $X$ \emph{passes} $(S_{n})_n$ if $X\in S_n$ for infinitely many $n$. It is a remarkable fact that $X$ is MLR if and only if it passes all Solovay tests. Is this also true in the quantum realm? Nies and Scholz asked \cite{unpublished1} if there is a notion of a quantum Solovay test and if so, is quantum Martin-L{\"o}f randomness equivalent to passing all quantum Solovay tests. We answer this question in the affirmative by defining a quantum Solovay test and quantum Solovay randomness as follows.
Roughly speaking, we obtain a notion of  a quantum Solovay test by replacing `$\Sigma^{0}_1$ class' and `Lebesgue measure' in the definition of classical Solovay tests with `quantum-$\Sigma^{0}_1$ set' and $\tau$ (Definition \ref{def:trace}) respectively. We show below that quantum Solovay Randomness is equivalent to q-MLR.

\begin{defn}\label{defn:1} 
A uniformly computable sequence of quantum-$\Sigma^{0}_1$ sets, $(S^{k})_{k\in\omega}$ is a \emph{ quantum-Solovay test} if  $\sum_{k\in \omega} \tau(S^{k}) <\infty.$\end{defn}\begin{defn}\label{defn:2}
For $0<\delta<1$, a state $\rho$ \emph{fails the Solovay test $(S^k)_{k\in\omega}$ at level $\delta$} if there are infinitely many $k$ such that $\rho(S^k)>\delta$.
\end{defn}
\begin{defn}
A state $\rho$ \emph{passes the Solovay test $(S^k)_{k\in\omega}$} if for all $\delta>0$, $\rho$ does not fail $(S^k)_{k\in\omega}$ at level $\delta$. I.e, lim$_{k}\rho(S^{k})=0$.
\end{defn}
\begin{defn}
A state $\rho$ is \emph{quantum Solovay random} if it passes all quantum Solovay tests.\end{defn}
An \emph{interval Solovay test}\cite{misc1} is a Solovay test, $(S_{n})_n$ such that each $S_n$ is generated by a finite collection of strings. By 7.2.22 in the book by Downey and Hirschfeldt \cite{misc1}, a Schnorr test may be defined as:
\begin{defn}
A \emph{Schnorr test} is an interval Solovay test, $(S^{m})_{m}$ such that $\sum_{m}\mu(S^{m}) $ is a computable real number. 
\end{defn}
A bitstring passes a Schnorr test if it does not fail it (using the same notion of failing as in the Solovay test). We mimic this notion in the quantum setting.
\begin{defn}
\label{defn:Schnorr}
A \emph{quantum Schnorr test} is a strong Solovay test, $(S^{m})_{m}$ such that $\sum_{m}\tau(S^{m}) $ is a computable real number. A state is quantum Schnorr random if it passes all Schnorr tests.
\end{defn}
The following two definitions are due to Nies (personal communication). The first is a quantum analogue of an interval Solovay test.
\begin{defn}
A \emph{strong Solovay test} is a computable sequence of special projections $(S^{m})_{m}$ such that $\sum_{m}\tau(S^{m}) < \infty$. A state $\rho$ fails $(S^{m})_{m}$ at $\epsilon$ if for infinitely many $m$, $\rho(S^{m})>\epsilon$.
\end{defn}
\begin{defn}
A state $\rho$ is \emph{weak Solovay random} if it passes all strong quantum Solovay tests.\end{defn}

\subsection{A general result about density matrices}
We prove a purely linear algebraic theorem about approximating density matrices by subspaces and then use it to show the equivalence of quantum Solovay and quantum Martin-L{\"
o}f randomness in the next subsection.

In words, the theorem says the following. Let $\mathcal{F}$ be a set of subspaces of `small' (at most $d$) total dimension and let $Q$ be the set of density matrices `$\delta$ close' to at least $m$ many subspaces from $\mathcal{F}$. Then, there is a subspace of small (at most $6d/\delta m$) dimension `$\delta^{2}/72$ close' to every density matrix in $Q$ .

\begin{thm}
\label{thm:LinA}
Let $m,d,n \in \mathbb{N}$ and  $ \delta\in (0,1) $ be arbitrary. Let $\mathcal{F}=(T_k)_{k}$ be a set of subspaces of $\mathbb{C}^{n}$ with $\sum_{k}\text{dim}(T_{k}) \leq d,$ and let $M_k$ be the orthonormal projection onto $T_k$. 
Let\[Q= \left\{\rho  : \rho \text{ is a density matrix on }\mathbb{C}^{n} \text{with Tr} (\rho M_k) >\delta  \text { for at least } m \text{ many } k\right\}, \]
be non-empty. Then, there is a orthonormal projection matrix $M$ such that\[\text{Tr}(M)\leq \dfrac{6d}{\delta m}   \text{ and Tr}(M \rho) \geq \dfrac{\delta^{2}}{36}  \text{ for every } \rho  \in Q.\]
 
\end{thm}

\begin{proof}

Let
\[L= \left\{\psi \in \mathbb{C}^{n}: ||\psi||=1, \sum_{k} \text{Tr}(|\psi \big> \big< \psi| M_{k})> \dfrac{m\delta}{6}\right\},\]
and let $D$ be a maximal orthonormal subset of $L$ and let $M$ be the orthonormal projection matrix onto the span of $D$.

\begin{lem}
\label{lem:mltbnd}
Tr$(M) \leq \dfrac{6d}{\delta m} $.
\end{lem}

\begin{proof}
We prove this using that $D$ is a orthonormal subset of $L$, that
Tr$(M)=|D|$ and 
that $d$ bounds the sum of the dimensions.
\begin{align*}
     d &\geq \sum_{k}  \text{Tr}(M_{k})\geq \sum _{k}  \sum_{\psi\in D} \text{Tr}(|\psi\big> \big< \psi|M_{k})=\sum_{\psi\in D}\sum _{k}   \text{Tr}(|\psi\big> \big< \psi|M_{k})\\
            &>|D|\dfrac{m\delta}{6}=\text{Tr}(M)\dfrac{m\delta}{6}.\qedhere
 \end{align*}
\qedhere
\end{proof}

Take any $\rho \in Q$. We can write it as 
 \[\rho = \sum_{i\leq n} \alpha_{i}|\psi^{i}\big> \big< \psi^{i}| \]  
for $\alpha_i$ non-negative real numbers with $\sum_{i\leq n} \alpha_i =1$ and for each $i$, $|\psi^{i}\big> \in \mathbb{C}^{n}$ and $||\psi^{i}||=1$.
For any $i\leq n$ we can decompose $\psi = \psi^i$ as 
\begin{align}
    \psi = c_o\psi_o + c_p\psi_p
\end{align}
where $\psi_o \in \text{range}(M)$ and $\psi_p \in \text{range}(M)^{\perp}$
are unit vectors and $c_o, c_p \in \mathbb{C}$ satisfy $|c_0|^{2}+|c_p|^{2}= 1$.
We now show that $\dfrac{\delta^{2}}{36} \leq $  Tr$(\rho M)$ = $\sum_{i\leq n}\alpha_{i} |c^{i}_o|^{2}$. Let $k$ be arbitrary and let $M_{k}=S$. A routine computation gives,

\begin{align}
\label{eq:1}
    &\text{Tr}(S|\psi\big> \big< \psi|)\\ 
    &\leq  |c_o|^{2}\big<S \psi_o|S\psi_o\big>+ |c_p|^{2} \big<S \psi_p|S\psi_p\big>+ 2| c_{o}||c_{p}||\big<S \psi_p|S\psi_o\big>|.
\end{align}

By the Cauchy-Schwarz inequality:

\begin{align*}
    |\big<S \psi_p|S\psi_o\big>|
    &\leq ||S\psi_o ||||S\psi_p ||\\
    & \leq (\text{max} \{||S\psi_o ||,||S\psi_p ||\})^{2}\\
    &\leq ||S\psi_o ||^{2} + ||S\psi_p ||^{2}.
\end{align*}
 
Putting this in \ref{eq:1}, we have: 
\begin{align}
\label{eq:decomp}
    &\text{Tr}(S|\psi\big> \big< \psi|)\\
    &\leq |c_o|^{2}\big<S \psi_o|S\psi_o\big>+ |c_p|^{2} \big<S \psi_p|S\psi_p\big>+ 2| c_{o}||c_{p}|(||S\psi_o ||^{2} + ||S\psi_p ||^{2})\\
    & \leq  |c_o|\big<S \psi_o|S\psi_o\big>+ |c_p| \big<S \psi_p|S\psi_p\big>+ 2|c_{o}|||S\psi_o ||^{2} + 2|c_{p}|||S\psi_p ||^{2}\\
    &=  3(|c_o|\big<S \psi_o|S\psi_o\big>+ |c_p| \big<S \psi_p|S\psi_p\big>).
\end{align}

As $\rho \in Q$, pick $H$ such that $|H|=m$ and $\text{Tr}(\rho M_{k}) > \delta$ for each $k$ in $H$. Using the above,

\begin{align*}
    m\delta
    &< \sum_{k\in H}\text{Tr}(\rho M_{k})\\
    &=   \sum_{i\leq n}\alpha_{i} \sum_{k\in H} \text{Tr}(|\psi^{i}\big> \big< \psi^{i}| M_{k})\\
    &\leq \sum_{i\leq n}\alpha_{i} \sum_{k\in H} 3(|c^{i}_o|\big<M_{k} \psi^{i}_o|M_{k}\psi^{i}_o\big>+ |c^{i}_p| \big<M_{k} \psi^{i}_p|M_{k} \psi^{i}_p\big>).
    \end{align*}
 So, 
\begin{align}
    &\dfrac{m\delta}{3} \leq \sum_{i\leq n}\alpha_{i} \sum_{k\in H} (|c^{i}_o|\big<M_{k} \psi^{i}_o|M_{k} \psi^{i}_o\big>+ |c^{i}_p| \big<M_{k} \psi^{i}_p|M_{k}\psi^{i}_p\big>)\\
    &=\label{eq:3} \sum_{i\leq n}\alpha_{i} |c^{i}_o|\sum_{k\in H} \big<M_{k} \psi^{i}_o|M_{k}\psi^{i}_o\big> + \sum_{i\leq n}\alpha_{i}|c^{i}_p| \sum_{k\in H}\big<M_{k} \psi^{i}_p|M_{k} \psi^{i}_p\big>.
\end{align}
 
Recall that our goal was to bound $\sum_{i\leq n}\alpha_{i} |c^{i}_o|^{2}$ from below by  $\delta^{2}/36$. In what follows, we achieve this by observing that the maximality of $D$ implies that $\psi^{i}_p\notin L$.
 
Fix an arbitrary $i$ and recall that $\psi^{i}_p \in \text{range}(M)^{\perp}$. Hence, $\psi^{i}_p$ is perpendicular to each element of $D$. If, $\psi^{i}_p \in L$, then $\{\psi^{i}_p\} \cup D$ is a orthonormal subset of $L$ strictly containing $D$, contradicting the maximality of $D$. So, for each $i$ it must be that $\psi^{i}_p \notin L$. But  $||\psi^{i}_p ||=1$. This implies that for each $i$,\\
\begin{align}
\label{eq:mimic}
    \sum_{k} \text{Tr}(|\psi^{i}_p  \big> \big< \psi^{i}_p | M_{k})\leq  \dfrac{m\delta}{6}.
\end{align}
As $\sum_{i\leq n} \alpha_{i} =1$ and  $|c^{i}_{p}|\leq 1$, the second term in \ref{eq:3} can be bounded from above:
\begin{align}\label{eq:4}
    \sum_{i\leq n}\alpha_{i}|c^{i}_p| \sum_{k\in H} \big<M_{k} \psi^{i}_p|M_{k} \psi^{i}_p\big>\leq \dfrac{m\delta}{6}.
\end{align}
Also note that
\begin{align} \label{eq:5}
    \sum_{k\in H} \big<M_{k} \psi^{i}_o|M_{k}\psi^{i}_o\big> \leq m.
\end{align}

\ref{eq:3}, \ref{eq:4} and \ref{eq:5} imply that

\[\dfrac{\delta}{6} \leq \sum_{i\leq n}\alpha_{i} |c^{i}_o|.\]

By Jensen's inequality,
\begin{align}
    \dfrac{\delta^{2}}{36} \leq \big (\sum_{i\leq n}\alpha_{i} |c^{i}_o|\big )^{2} \leq \sum_{i\leq n}\alpha_{i} |c^{i}_o|^{2}.
\end{align}

Finally, 
\begin{align*}
    \text{Tr}(\rho_n M)&=\sum_{i\leq n}\alpha_{i}\text{Tr}(|\psi^{i}\big> \big< \psi^{i}|M)\\
            &=\sum_{i\leq n}\alpha_{i}\text{Tr}(|M\psi^{i}\big> \big< M\psi^{i}|)\\
            &=\sum_{i\leq n}\alpha_{i}\text{Tr}(|c^{i}_{o}\psi^{i}_{o}\big> \big< c^{i}_{o}\psi^{i}_{o}|)\\
            &=\sum_{i\leq n}\alpha_{i} |c^{i}_o|^{2}\geq \dfrac{\delta^{2}}{36}.\qedhere
\end{align*}

\end{proof}

\subsection{Quantum Solovay randomness is equivalent to quantum Martin-L{\"o}f randomness }
\begin{thm}
\label{thm:000}
A state is quantum Solovay random if and only if it is quantum Martin-L{\"o}f random. 
\end{thm}
\begin{proof}
It suffices to show that if a state $\rho$ is not quantum Solovay random then it is not quantum Martin-L{\"o}f random.
To this end, let $\rho=(\rho_n)_{n\in \omega}$ be a state which fails a quantum Solovay test, $(S^{k})_{k\in \omega}$ at level $\delta$. We show that $\rho$ is not quantum Martin-L{\"o}f random by building a quantum Martin-L{\"o}f test, $(G^{m})_{m\in\omega}$, with $G^{m}= (G^{m}_n)_{n\in\omega}$,  which $\rho$ fails at level $ \delta^{2}/72$. We will use an effective version of Theorem \ref{thm:LinA} . Without loss of generality, assume that $S^{k}_n =\emptyset$ for $k>n$ and let $\sum_{k} \tau(S^k )<1$. 
 We use the notation:
\[A^{m}_{t}= \left\{\psi \in \mathbb{C}^{2^t}_{alg}: ||\psi||=1, \sum_{k\leq t} \text{Tr}(|\psi \big> \big< \psi| S^{k}_{t})> \dfrac{2^{m}\delta}{6}\right\},\]
for $t,m\in\omega$. This is analogous to $L$ in \ref{thm:LinA} with the replacements, $m\mapsto 2^{m}, n\mapsto 2^{t}$ and $d\mapsto 2^n$ and where we restrict attention to algebraic vectors. We use $A$ instead of $L$ to emphasize that we only consider complex algebraic objects in $A$, a `computable' version of $L$\\

\emph{Construction of $G^{m}$}: 
We build $G^m$ inductively as follows. Given $C^{m}_{n-1} \subseteq \mathbb{C}^{2^{n-1}}_{alg}$, a maximal (under set inclusion) orthonormal subset of $A^{m}_{n-1}$, let \[D^{m}_{n}=\left\{|\psi \big>\otimes |i\big> \in \mathbb{C}^{2^n}_{alg}: i \in \{1,0\},     \psi \in C^{m}_{n-1}\right\}.\]
 Note that $D^{m}_n \subseteq A^{m}_n$ since $C^{m}_{n-1} \subseteq A^{m}_{n-1}$. 
Define $C^{m}_n $ to be $S$ where $S$ is a maximal orthonormal set such that $S \subseteq A^{m}_n $ and $ D^{m}_n \subseteq S$. Let $G^{m}_{n}$ be the projection: \[G^{m}_{n}= \sum_{\psi \in C^{m}_{n}}|\psi\big> \big< \psi|.\] \emph{End of construction}.

\begin{lem}
$(G^{m})_{m\in\omega}$ is a quantum Martin-L{\"o}f test.
\end{lem}

\begin{proof}
Fix $m$. Clearly,  $(C^{m}_{n})_{n\in\omega}$ is a uniformly computable sequence. By construction, $\text{range}(G^{m}_{n-1} \otimes I_{2}) \subseteq \text{range}(G^{m}_n)$.
So, $G^{m}=(G^{m}_{n})_{n\in\omega}$ is a  quantum-$\Sigma^{0}_1$ set for each $m$.
The sequence $(G^{m})_{m\in\omega}$ is uniformly computable in $m$ by construction. Since, $ 1 \geq \sum_{k} \tau(S^{k})$, we have that $ 2^{n} \geq \sum_{k} \text{Tr}(S^{k}_{n})$ for all $n$. Now make the replacements $m\mapsto 2^{m}, n\mapsto 2^{n}$ and $d\mapsto 2^n$ in the proof of \ref{lem:mltbnd}  to see that Tr$(G^{m}_{n})  \leq (6/\delta)  2^{n-m}$ for all  $m,n$. So  $\tau(G^{m})  \leq (6/\delta)  2^{-m}$ for all $m$.
\end{proof}
\begin{lem}
$\rho$ fails $(G^m)_m$ at level $\dfrac{\delta^{2}}{72}$.
\end{lem}
\begin{proof}
We must show that $\text{inf}_{m\in\omega} \rho(G^m) > \dfrac{\delta^{2}}{72}.$ It suffices to show that for all $m\in\omega$, there is an $n$ such that 
$\text{Tr}(\rho_n G^{m}_n)> \dfrac{\delta^{2}}{72}.$ To this end, let $m$ be arbitrary and fix a $n$ big enough so that there exist $2^{m}$ many $ks$ such that Tr$(\rho_n S^{k}_n) > \delta$. So, let $|H|=2^m$ and $\text{Tr}(\rho_n S^{k}_n) > \delta$ for each $k$ in $H$. The projection $G^{m}_{n}$ will play the role of $M$ in the proof of Theorem \ref{thm:LinA}. Write $\rho_n$ as
 \[\rho_n = \sum_{i\leq 2^n} \alpha_{i}|\psi^{i}\big> \big< \psi^{i}| \]  
for $\alpha_i$ non-negative real numbers with $\sum_{i\leq2^n} \alpha_i =1$ and for each $i$, $|\psi^{i}\big> \in \mathbb{C}^{2^n}$ and $||\psi^{i}||=1$. First, consider the case where $|\psi^{i}\big>\in \mathbb{C}^{2^n}_{alg}$ for all $i$. 
For any $i\leq 2^{n}$ we can decompose $\psi = \psi^i$ as,
\begin{align*}
    \psi = c_o\psi_o + c_p\psi_p
\end{align*}
as in the proof of Theorem \ref{thm:LinA}, which we mimic now. By equation \ref{eq:3}, 
\begin{align}
\label{eq :3000}
    &\dfrac{2^{m}\delta}{3} \leq  \sum_{i\leq 2^n}\alpha_{i} |c^{i}_o|\sum_{k\in H} \big<S^{k}_n \psi^{i}_o|S ^{k}_n\psi^{i}_o\big> + \sum_{i\leq 2^n}\alpha_{i}|c^{i}_p| \sum_{k\in H}\big<S^{k}_n \psi^{i}_p|S^{k}_n \psi^{i}_p\big>
\end{align}

Fix an arbitrary $i$ and recall that $\psi^{i}_p \in \text{range}(G^{m}_n)^{\perp} \cap \mathbb{C}^{2^n}_{alg}$. Hence, $\psi^{i}_p$ is perpendicular to each element of $C^{m}_n$. If $\psi^{i}_p \in A^{m}_n$, then $\{\psi^{i}_p\} \cup C^{m}_n$ is a orthonormal subset of $A^{m}_n$ strictly containing $C^{m}_n$, contradicting the maximality of $C^{m}_n$. So, for each $i$ it must be that $\psi^{i}_p \notin A^{m}_n$. But, $\psi^{i}_p \in \mathbb{C}^{2^n}_{alg}$ and $||\psi^{i}_p ||=1$. This implies that for each $i$,\\
\[ \sum_{k\leq n} \text{Tr}(|\psi^{i}_p  \big> \big< \psi^{i}_p | S^{k}_{n})\leq  \dfrac{2^{m}\delta}{6}.\]
We are now in the situation of equation \ref{eq:mimic}. As the argument following it does not need complex algebraic vectors and by recalling that $M$ is replaced by $G^{m}_n$, we see that Tr$(\rho_{n}G^{m}_n)\geq \delta^{2}/36>\delta^{2}/72$. Now, suppose that not all $|\psi^{i}\big>$ are algebraic. By the density of $\mathbb{C}^{2^n}_{alg}$ in $\mathbb{C}^{2^n}$ we can approximate $\rho_n$ by a sequence $(\pi_k)_{k\in \mathbb{N}}$ of density matrices each satisfying the conditions of the previous case. So, Tr$(\pi_{k}G^{m}_n )\geq \delta^{2}/36$. By continuity, Tr$(\rho_{n}G^{m}_n)\geq \delta^{2}/36>\delta^{2}/72$.  \qed 
\end{proof}
The theorem is proved.\qed 
\end{proof}
\subsection{Convexity}
We show that all classes of random states are convex. The first result in this section is a corollary of the main theorem from the previous section.
\begin{cor}

A convex combination of q-Martin-L{\"o}f random states is q-Martin-L{\"o}f random. Formally, if $(\rho^i )_{i<k<\omega}$ are q-ML random states and $\sum_{i<k<\omega} \alpha_{i} =1$, then $\rho=\sum_{i<k}\alpha_{i}\rho^{i}$ is q-ML random.
\end{cor}
\begin{proof}
Suppose for a contradiction that there is a q-Martin-L{\"o}f test $(G^{m})_{m\in \omega}$  and a $\delta>0$ such that $\forall m\in \omega$, $ \rho(G^m)>\delta$.  So, $\forall m \in \omega$, $\exists n$ such that Tr$(\rho_{n}G^{m}_{n})>\delta $ where $\rho_n =\sum_{i<k}\alpha_{i}\rho^{i}_n$. So, $\forall m \in \omega$, $\exists n$ such that \[\delta<\text{Tr}\bigg(\sum_{i<k}\alpha_{i}\rho^{i}_{n}G^{m}_{n}\bigg)= \sum_{i<k}\alpha_{i}\text{Tr}(\rho^{i}_{n}G^{m}_{n}).\]
By convexity of the sum, there is an $i$ such that Tr$(G^{m}_{n}\rho^{i}_{n})>\delta$. In summary, \[\forall m, \text{ there is an } i  \text{ and an }n  \text{ such that Tr} (\rho^{i}_{n}G^{m}_{n})>\delta.\]
Since there are only finitely many $i$s, by the pigeonhole principle, there is an $i$ such that $\exists^{\infty} m$ with Tr $ (\rho^{i}_{n}G^{m}_{n})>\delta$, for some $n$. So, $\exists^{\infty} m$ with $ \rho^{i} (G^{m})>\delta$. So, $\rho^i$ fails the q-Solovay test $(G^{m})_{m\in \omega}$ and hence is not q-Martin-L{\"o}f random by our previous result. This is a contradiction. 
\end{proof}
\begin{thm}

A convex combination of quantum Schnorr random states is quantum Schnorr random. Formally, if $(\rho^i )_{i<k<\omega}$ are quantum Schnorr random states and $\sum_{i<k<\omega} \alpha_{i} =1$, then $\rho=\sum_{i<k}\alpha_{i}\rho^{i}$ is quantum Schnorr random.
\end{thm}
\begin{proof}
Suppose for a contradiction that there is a quantum Schnorr test $(G^{m})_{m\in \omega}$  and a $\delta>0$ such that $\exists^{\infty} m\in \omega$, $ \rho(G^m)>\delta$. Letting $G^m$ be $n_{m}$ by $n_{m}$, $\exists^{\infty} m$, such that \[\delta<\text{Tr}(\rho_{n_{m}}G^{m}_{n_{m}})=\text{Tr}(\sum_{i<k}\alpha_{i}\rho^{i}_{n_{m}}G^{m}_{n_{m}})= \sum_{i<k}\alpha_{i}\text{Tr}(\rho^{i}_{n_{m}}G^{m}_{n_{m}}).\]
By convexity of the sum, there is an $i$ such that Tr$(G^{m}_{n_{m}}\rho^{i}_{n_{m}})>\delta$. In summary, \[\exists^{\infty} m, \text{ there is an } i   \text{ such that Tr} (\rho^{i}_{n_{m}}G^{m}_{n_{m}})>\delta.\]
Since there are only finitely many $i$ s, by the pigeonhole principle, there is an $i$ such that $\exists^{\infty} m$ with Tr $ (\rho^{i}_{n_{m}}G^{m}_{n_{m}})>\delta$. So, $\exists^{\infty} m$ with $ \rho^{i} (G^{m})>\delta$. So, $\rho^i$ fails the q-Schnorr test $(G^{m})_{m\in \omega}$ and hence is not q-Schnorr. This is a contradiction.
\end{proof}
Noting that the above proof needed only the Solovay type of failing criterion, we get:
\begin{thm}

A convex combination of weak Solovay random states is q-weak Solovay  random. Formally, if $(\rho^i )_{i<k<\omega}$ are weak Solovay random states and $\sum_{i<k<\omega} \alpha_{i} =1$, then $\rho=\sum_{i<k}\alpha_{i}\rho^{i}$ is weak Solovay  random.
\end{thm}
The proof is almost identical to the previous one.
\subsection{Nesting property of quantum Martin-L{\"o}f tests}
It is interesting to see which classical results carry over to the quantum realm. For example, the existence of a universal MLT, $(U_{n})_n$ such that a bitstring is MLR if and only if it passes this $(U_{n})_n$ does carry over \cite{unpublished}. The `nesting property' of the classical Martin-L{\"o}f test says that we can, without loss of generality assume the universal test $(U_{n})_n$ to be nested; i.e., to satisfy $U_{n+1} \supseteq U_n$ for all $n$. We extend this property to the quantum setting:
\begin{thm}
\label{thm:00}
There is a q-MLT, $(Q^{m})_{m\in \mathbb{N}}$ with the properties (1) If a state $\rho$ fails the universal q-Martin-L{\"o}f test $(G^{m})_{m\in \mathbb{N}}$ at $\delta>0$, then, it also fails $(Q^{m})_{m\in \mathbb{N}}$ at $\delta>0$ (2) If $Q^{m}=(Q^{m}_{n})_{n\in \mathbb{N}}$ for all m, then for all $m$ and $n$, range$(Q^{m+1}_{n}) \subseteq $ range $(Q^{m}_{n}) $. In particular, $Q^{m+1} \leq Q^m$ for all m.
\end{thm}
\begin{proof}
Informally speaking, we want to let $Q^m$ be $\sum_{i>m} G^{i}$. Precisely, we build $Q^m$ level by level. For any natural numbers $i \leq n$, let\[G^{i}_{n} = \sum_{j=1}^{2^{n-i}} |v^{i,n}_{j}\big>\big<v^{i,n}_{j}|.\] Let 
\[S^{m}_{n}:= \text{span} \bigcup_{i=m}^{n} \{v^{i,n}_{j}: 1\leq j \leq 2^{n-i}\},
\]
and let $Q^{m}_{n}$ be the special projection onto $S^{m}_{n}$. Let $Q^{m}= (Q^{m}_{n})_{n}$. Fix an $m$. We see that $Q^{m}_{n} \leq Q^{m}_{n+1}$, since $G^{i}_{n} \leq G^{i}_{n+1}$ holds for all $i,n$. So, $Q^m$ is a q-$\Sigma^{0}_1$ class. The dimension of $S^{m}_{n}$ is at most $\sum_{i=m}^{n} \text{Trace}(G^{i}_{n}) \leq \sum_{i=m}^{n} 2^{n-i} < 2^{n-m+1}.$ So, $(Q^{m})_{m=2}^{\infty}$ is a q-MLT. Let $m$ and $n$ be arbitrary and $n\geq m+1$. Then, clearly, by definition of $S^{m}_{n}$, we see that range$(Q^{m+1}_{n}) \subseteq $ range $(Q^{m}_{n}) $. So, the nesting property holds. Let $\rho= (\rho_{n})_{n}$ be a state. By the nesting, and by properties of projection operators, we have that for a fixed $m$ and all $n$,
\[\text{Tr}(\rho_{n}Q^{m+1}_{n}) \leq \text{Tr}(\rho_{n}Q^{m}_{n}) \leq \rho(Q^{m}).\]
So, $\rho(Q^{m+1})=$ sup$_{n}\text{Tr}(\rho_{n}Q^{m+1}_{n}) \leq \rho(Q^{m})$ for all $m$. (1) clearly holds.
\end{proof}

\section{Randomness for diagonal states}
\label{sec:measures}

A state $\rho=(\rho_{n})_n$ is defined to be \emph{diagonal} if $\rho_n$ is diagonal for all $n$. So, each $\rho_n$ in a diagonal state represents a  mixture of separable states. A diagonal $\rho=(\rho_{n})_n$ can be thought of as a measure on Cantor space, denoted by $\mu_{\rho}$: if $\sigma \in 2^n$, we define $\mu_{\rho}(\llbracket\sigma\rrbracket):= \big<\sigma|\rho_{n}|\sigma\big>$. We will write $\mu_{\rho}(\sigma)$ instead of $\mu_{\rho}(\llbracket\sigma\rrbracket)$. $\mu_{\rho}$ is easily seen to be a measure by noting that the partial trace over the last qubit of $\rho_{n+1}$ equals $\rho_n$ for all $n$. Recalling the notation in Remark \ref{rem:notation} and as $S$ is prefix free, we have, \[\mu_\rho (\llbracket S \rrbracket)= \sum_{\sigma \in S} \mu_\rho (\sigma) =\sum_{\sigma \in S} \big<\sigma|\rho_{n}|\sigma\big>=\text{Tr}(\rho_n P_S).\]
This will be used frequently. Nies and Stephan have recently defined a notion of randomness for measures on Cantor space called Martin-L{\"o}f absolute  continuity\cite{unpublished2}.
\begin{defn}
A measure $\pi$ on Cantor space is called \emph{ Martin-L{\"o}f absolutely continuous} if $\inf_{m}\pi(G_{m})=0$ for each classical MLT $(G_{m})_{m\in \mathbb{N}}$.
\end{defn}
This notion turns out to be equivalent to quantum Martin-L{\"o}f randomness in the sense that for a diagonal $\rho$, $\rho$ is q-MLR if and only if $\mu_\rho$ is  Martin-L{\"o}f absolutely continuous. It is easy to see that if a diagonal $\rho$ is q-MLR, then $\mu_\rho$ is Martin-L{\"o}f absolutely continuous. We now show the other direction.
\begin{thm}
\label{thm:30}
Let $\rho$ be diagonal. If it fails a q-MLT $(G^{m})_{m\in \mathbb{N}}$ at order $\delta$, then there is a classical MLT, $(C^{m})_{m\in \mathbb{N}}$ such that  $\inf_{m}\mu_{\rho}(C^{m})>\delta/2$.
\end{thm}
\begin{proof}
We isolate here a simple but useful property. 
\begin{lem}
\label{lem:30}
Let $n$ be a natural number, $E=(e_{i})_{i=1}^{2^{n}}$ be any orthonormal basis for $\mathbb{C}^{2^{n}}$ and $F$ be any hermitian, orthonormal projection matrix acting on $\mathbb{C}^{2^{n}}$. For any $\delta>0$, let  \[S^{\delta}_{E,F}:= \left\{e_{i} \in E: \big<e_{i}|F|e_{i}\big> > \delta \right\}.\]
Then, $|S^{\delta}_{E,F}| < \delta^{-1} \text{Tr}(F)$.
\end{lem}
\begin{proof}
Note that since $F$ is a hermitian orthonormal projection, $\big<e_{i}|F|e_{i}\big>=\big<Fe_{i}|F e_{i}\big> = |F e_{i}|^{2}\geq 0$. So,
\[\delta|S^{\delta}_{E,F}| < \sum_{e_{i} \in S^{\delta}_{E,F}}\big<e_{i}|F|e_{i}\big> \leq \sum_{i \leq 2^{n}}\big<e_{i}|F|e_{i}\big> = \text{Tr}(F).\]
\qedhere
\end{proof}
We now prove Theorem \ref{thm:30}. The intuition is as follows: given a special projection, we take the set of bitstrings (thought of as qubitstrings) `close' to it. If the special projection `captures' $\delta$ much mass of $\rho$, then the projection onto the span of these qubitstrings must capture atleast $\delta/2$ much mass of $\rho$.
$\sigma$ will always denote a finite length classical bit string and $|\sigma\big>$, the corresponding element of the standard computational basis. We may assume that $\delta$ is rational. Fix $m$. We describe the construction of $C^{m}= (C^{m}_{n})_{n\in \mathbb{N}}$ (See \ref{def:10}). Let  \[T^{m}_{n}:= \left\{\sigma \in 2^{n}: \big<\sigma|G^{m}_{n}|\sigma\big> > \dfrac{\delta}{4}\right\}.\]
These are those standard basis vectors `close' to $G^{m}_{n}$. Let
\[C^{m}_{n}= \bigcup_{\sigma \in T^{m}_{n}} \llbracket \sigma \rrbracket.\] 
\begin{lem}
$C^{m} $ is a $\Sigma^{0}_1$ class for any $m$.
\end{lem}
\begin{proof}
It is easy to see that for all $\sigma \in T^{m}_{n}$ and $i\in \{0,1\}$, \[\big<\sigma i|G^{m}_{n+1}|\sigma i \big > \geq \big<\sigma|G^{m}_{n}|\sigma\big> > \dfrac{\delta}{4}.\]
So, $\{\sigma i: \sigma \in T^{m}_n , i\in \{0,1\}\} \subseteq T^{m}_{n+1}$. Also note that $T^{m}_{n}$ is uniformly computable in $n$ since $G^{m}_{n}$ is.
\end{proof}
\begin{lem}
$(C^{m})_{m\in \mathbb{N}}$ is a MLT.
\end{lem}

\begin{proof}
Fix $m$. Letting $E=B^n$ and $F=G^{m}_n$ in Lemma \ref{lem:30} and by definition of q-MLT,
\[ |T^{m}_{n}| < \dfrac{4}{\delta}
2^n \tau(G^{m})\leq \dfrac{4}{\delta} 2^{n-m}.\]
So, $\mu(C^{m}) < 2^{-m}\dfrac{4}{\delta}$. $C^{m}$ is computable in $m$ since $G^{m}$ is.
\end{proof}
 
Now we show that $\inf_{m}\mu_{\rho}(C^{m})>\delta/2$.
Fix a $m$ and a $n$ (depending on $m$) such that Tr$(\rho_{n} G^{m}_{n})>\delta$. Let $\rho_{n} = \sum_{\sigma \in 2^{n} } \alpha_{\sigma} |\sigma\big>\big<\sigma|.$
Then, \begin{align*}
&\delta<\text{Tr}(\rho_{n} G^{m}_{n})
=\sum_{\sigma \in 2^{n} } \alpha_{\sigma} \big<\sigma|G^{m}_{n}|\sigma\big> = \sum_{\sigma \in T^{m}_{n}} \alpha_{\sigma} \big<\sigma|G^{m}_{n}|\sigma\big>  + \sum_{\sigma \in 2^{n}\backslash T^{m}_{n} } \alpha_{\sigma} \big<\sigma|G^{m}_{n}|\sigma\big>\\
&\leq \sum_{\sigma \in T^{m}_{n}}  \alpha_{\sigma}  + \sum_{\sigma \in 2^{n} \backslash T^{m}_{n}} \alpha_{\sigma} \dfrac{\delta}{4}\leq \sum_{\sigma \in T^{m}_{n}}  \alpha_{\sigma}  +  \dfrac{\delta}{4}=\text{Tr}(\rho_n  P_{C^{m}_{n}}) + \dfrac{\delta}{4}= \mu_{\rho}(C^{m}_{n}) + \dfrac{\delta}{4}.
\end{align*}
The last equality follows as $T_{n}^m$ is prefix free. So, $\mu_{\rho}(C^{m})\geq \mu_{\rho}(C^{m}_{n})\geq 3\delta/4$.
%
\end{proof}

Nies and Scholz showed that a measure, $\mu$ is Martin-L{\"o}f absolutely continuous if and only if for any Solovay test $(S_k)_k$, $\lim_{k}\mu (S_k )=0$\cite{unpublished2}. Adapting the proof of Theorem \ref{thm:000} yields another proof of this.
\begin{thm}
\label{thm:2}
Let $\rho$ be diagonal. If for some Solovay test $(S_k)_k$ and $\delta>0$ we have $\exists^{\infty} k , \mu_\rho (S_k )>\delta$, then there is a Martin-L{\"o}f test $(J_m )_m$ such that $\inf_{m}\mu_\rho (J^m )>\delta/2$. 
\end{thm}
The theorem will follow from the two lemmas below. Write $S^k = (S^{k}_n )_n$ as in Definition \ref{def:10}.
Without loss of generality, let $S^{k}_n =\emptyset$ for $k>n$. Let
\[C^{m}_{t}= \left\{ \sigma  \in 2^t : \sum_{k\leq t} |\big< \sigma| S^{k}_{t}|\sigma \big>|> 2^{m-1}\delta\right\},\] 
and let $G^{m}_{t}:= P_{C^{m}_{t}}$ (See Remark \ref{rem:notation}).
Let $G^{m}=(G^{m}_{n})_{n}$. It is easy to see that $G^m$ is a q-$\Sigma^{0}_1$ set for each $m$. Let $J^{m}_n := \llbracket C^{m}_n \rrbracket$ and $J^m = (J^{m}_n)_n$. One can check that that $(J^m )_m$ is a MLT if and only if $(G^{m})_{m}$ is quantum Martin-L{\"o}f test. So, $(J^m )_m$ is a MLT since:

\begin{lem}
$(G^{m})_{m}$ is quantum Martin-L{\"o}f test.
\end{lem}

\begin{proof}
Identical to the proof of \ref{lem:mltbnd}.
\end{proof}
\begin{lem}
We have that $\inf_{m}\mu_\rho (J^m )>\delta/2$.
\end{lem}
\begin{proof}
Let $m$ be arbitrary. By assumption, there are infinitely many $k$s such that  $\mu_\rho (S^{k})>\delta$. For each of these, there is an $n$ such that $\mu_\rho (S^{k}_n )>\delta$. So, fix a $n$ so that there are $2^{m}$ many $ks$ such that $\mu_\rho( S^{k}_n) > \delta$.
Since $\rho_n$ is diagonal, let 
 \[\rho_n = \sum_{\sigma \in 2^n} \alpha_{\sigma}|\sigma \big> \big< \sigma |.\]  

By the choice of $n$, pick $M\subseteq\{1,2\cdots ,n\}$ such that $|M|=2^m$ and $ \mu_\rho (S^{k}_n) > \delta$ for each $k$ in $M$. Note that $\mu_\rho (S^{k}_n)= \text{Tr}(\rho_n P_{S^{k}_n})$, since $S^{k}_n$ is prefix free. We write $\text{Tr}(\rho_n P_{S^{k}_n})=\text{Tr}(\rho_n S^{k}_n)$ to avoid clutter. So,

\begin{align*}
    2^{m}\delta
    &<\sum_{k\in M}\mu_\rho (S^{k}_n)= \sum_{k\in M}\text{Tr}(\rho_n S^{k}_n)
    =   \sum_{\sigma \in 2^n}\alpha_{\sigma} \sum_{k\in M} \text{Tr}(|\sigma \big> \big< \sigma | S^{k}_n)\\
    &=\sum_{ \sigma \in C^{m}_{n}}\alpha_{\sigma} \sum_{k\in M} \big< \sigma |S ^{k}_n\sigma \big> + \sum_{ \sigma \notin C^{m}_{n}}\alpha_{\sigma} \sum_{k\in M} \big< \sigma |S ^{k}_n\sigma \big>\\
    &\leq \sum_{ \sigma \in C^{m}_{n}}\alpha_{\sigma} \sum_{k\in M} \big< \sigma |S ^{k}_n\sigma \big> + 2^{m-1}\delta\\
    &\leq 2^{m} \sum_{ \sigma \in C^{m}_{n}}\alpha_{\sigma} + 2^{m-1}\delta.
\end{align*}
The second last inequality follows from the definition of $G^{m}_{n}$ and convexity; the last from the choice of $M$.
Finally, we get that,
\[\delta/2 < \sum_{\sigma \in C^{m}_{n}}\alpha_{\sigma} = \mu_\rho(\llbracket C^{m}_{n}\rrbracket)\leq \mu_\rho (J^m ).\] \qedhere
\end{proof}

Next, we discuss a subset of the diagonal states; the Dirac delta measures on Cantor space.
\subsection{Quantum randomness on Cantor Space}
A classical bitstring can be thought of as a diagonal state: If $X$ is a real in Cantor space, the state $\rho_{X}=(\rho_{n})_n$ given by $\rho_{n}=|X\upharpoonright n \big>\big<X\upharpoonright n |$ is the quantum analog of $X$. Do the quantum randomness notions agree with classical notions when restricted to Cantor space? By Theorem \ref{thm:30}, we see that $\rho_{X}$ is q-MLR if and only if $X$ is MLR. Further, $\rho_{X}$ is q-MLR if and only if $\rho_{X}$ is weak Solovay random. Also, $X$ is MLR if and only if it passes all interval Solovay tests (the classical analog of strong Solovay tests). So, we see that q-MLR and weak Solovay randomness agree with the classical versions on Cantor space. What about quantum Schnorr randomness?
\begin{lem}
\label{thm:schcla}
$\rho_{X}$ is quantum Schnorr random if and only if $X$ is Schnorr random.
\end{lem}
\begin{proof}
Let $(Q^r)_r$ be a quantum Schnorr test which $\rho_X$ fails at some rational $\delta$. Let $Q^r$ be $n_r$ by $n_r$. Using notation of Lemma \ref{lem:30}, let $T^{r}:= S^{\delta}_{E,Q^{r}}$ where $E$ is the set of length $n_r$ standard basis vectors. We think of $T^r$ as a set of classical bitstrings. By Lemma \ref{lem:30}, $\tau(T^{r})\leq \delta^{-1}\tau(Q^{r})$. So,  $\sum_{r}2^{-n_{r}}|T^{r}|=\sum_{r}\tau(T^{r})  \leq \delta^{-1}\sum_{r}\tau(Q^{r})$ is computable because $\sum_{r}\tau(Q^{r})$ is. So, $(T^r)_r$ is a finite total Solovay test. Let $m$ be one of the infinitely many $r$ such that $\delta<$Tr$(\rho_{X} (Q^r))$. Then, by definition, $X\upharpoonright n_r $ is in  $T^r$. So, $X$ fails $(T_{r})_r$ and hence is not Schnorr random (by 7.2.21 and 7.2.22 in the book by Downey and Hirschfeldt\cite{misc1}). The other direction is trivial. \qedhere

\end{proof}
\subsection{Relating the randomness notions}
We have seen that \[\text{Solovay R} = \text{q-MLR} \subseteq \text{weak Solovay R} \subset \text{quantum Schnorr R}.\]
The equality follows by Theorem \ref{thm:000}. The second inclusion is strict as there is a bitstring which is Schnorr random but which fails some interval Solovay test \cite{misc1} and since by Theorem \ref{thm:schcla}, this bitstring must be quantum Schnorr random. It is open whether the first inclusion is strict.
\section{A law of large numbers for quantum Schnorr randoms}
\label{sec:lln}
The law of large numbers (LLN), specialized to Cantor space says that the limiting proportion of ones is equal to 0.5 for almost every bitstring. Random bitstrings satisfy the LLN. In fact, satisfying the LLN is the weakest form of randomness \cite{misc1}. This is quite intuitive; one would not call a bitstring `random' if it has more ones than zeroes in the limit. Analogously, we expect even our weakest notion of quantum randomness (quantum Schnorr randomness) to satisfy a quantum analogue of the LLN. This suggests that the quantum randomness notions are `natural' and mirror the classical situation. In this section, $\sigma$ will always denote a classical bitstring thought of as a qubit string. 

\begin{defn}
\cite{logicblog} $\rho$ \emph{satisfies the LLN} if lim$_{n}n^{-1}\sum_{i<n}$Tr$(\rho_{n}P^{n}_{i})=0.5$, where for all $i\geq 0, n>0$, \[P^{n}_{i}:= \sum_{\sigma: |\sigma|=n, \sigma(i)=1}|\sigma\big>\big<\sigma| .\]

\end{defn}
The intuition is that $P^{n}_i$ is the projection observable which measures whether a given density matrix on $n$ qubits `has a one in the $i^{th}$ spot'. Tr$(\rho_{n}P^{n}_{i})$ is the probability that $\rho_n$ `has a one in the $i^{th}$ spot'. If the average over $i$ of these probabilities tends to $0.5$ as $n$ goes to infinity, then the state satisfies the LLN.
\begin{thm}
Quantum Schnorr random states satisfy the LLN.
\end{thm}
\begin{proof}
We prove it by contradiction. Suppose $\rho$ is quantum Schnorr random but does not satisfy the LLN. So, there is a $\delta$ such that either $\exists^{\infty}n$, with $n^{-1}\sum_{i<n}$Tr$(\rho_{n}P^{n}_{i})> \delta + 0.5$ or  $\exists^{\infty}n$, with $n^{-1}\sum_{i<n}$Tr$(\rho_{n}P^{n}_{i})< -\delta + 0.5$. Suppose first that the former holds. A rough outline of this proof is as follows. For each $n$ we take $S_n$ to be the subspace spanned by the classical strings with the fraction of $1$s more than $0.5+ \delta/2$. $(S_n)_n$ is easily seen to be a quantum Schnorr test and it only remains to show that $\rho$ fails it. This is not obvious as $\rho$ is not necessarily classical, while $(S_n)_n$ is composed of classical vectors. To show this, we consider one of the infinitely many $n$s such that $n^{-1}\sum_{i<n}$Tr$(\rho_{n}P^{n}_{i})> \delta + 0.5$. For such an $n$, we break up $n^{-1}\sum_{i<n}$Tr$(\rho_{n}P^{n}_{i})$ into two parts: the first corresponding to the projection of $\rho_n$ onto $S^{\perp}_n$ and the other corresponding to the projection onto $S_n$ (see for example in equation \ref{eq:zeroth}). The definition of $S_n$ enables us to upper bound the first part (see \ref{eq:1st}). The second part is forced to be big since $n^{-1}\sum_{i<n}$Tr$(\rho_{n}P^{n}_{i})> \delta + 0.5$. So, $\rho$ fails $(S_n)_n$. The details are:
Define for all $n$,
\[C_{n}= \left\{\sigma : |\sigma|=n, n^{-1}\sum_{i<n} |\big< \sigma| P^{n}_{i}|\sigma \big>|> \delta/2 + 0.5\right\}.\] 
In other words,
\[C_{n}= \left\{\sigma : |\sigma|=n, n^{-1}|\{i<n: \sigma(i)=1 |> \delta/2 + 0.5\right\}.\] 
Let $S_n$ be the special projection, \[S_{n}:= \sum_{\sigma \in C_{n}} |\sigma\big>\big<\sigma|.\]
$(S_n)_n$ is a computable sequence since we may let $\delta$ be rational. By the Chernoff bound, $\tau(S_{n})=2^{-n}|C_{n}| \leq 2\text{exp}(-0.5 n\delta^{2})$ for all $n$. So, $\sum_{n}\tau(S_{n})$ is computable showing that $(S_n)_n$ is a quantum Schnorr test.

For all $n$, let 
\[\rho_n = \sum_{k< 2^n} \alpha_{k}|\psi^{k}_{n}\big> \big< \psi^{k}_{n}| \]  
for $\alpha_k$ non-negative real numbers with $\sum_{k<2^n} \alpha_k =1$ and for each $k$, $|\psi^{k}_{n}\big> \in \mathbb{C}^{2^n}$ and $||\psi^{k}_{n}||=1$.
Fix an $n$ is such that $n^{-1}\sum_{i<n}$Tr$(\rho_{n}P^{n}_{i})>0.5+\delta.$
We will drop the $n$ subscript of $\psi^{k}_{n}$ as the $n$ is fixed.
For any $k< 2^{n}$ we can decompose $\psi^{k}$ as,
\begin{align}
    \psi^{k} = c^{k}_o\psi^{k}_o + c^{k}_p\psi^{k}_p
\end{align}
where $\psi^{k}_o \in \text{range}(S_{n})$ and $\psi^{k}_p \in \text{range}(S_n)^{\perp}$
are unit vectors and $c^{k}_o, c^{k}_p \in \mathbb{C}$ satisfy $|c^{k}_0|^{2}+|c^{k}_p|^{2}= 1$.
We now show that $\dfrac{\delta^{2}}{36}$ is a lower bound for Tr$(\rho_n S_{n})$ = $\sum_{k< 2^n}\alpha_{k} |c^{k}_o|^{2}$. Note that
\begin{align}
    & n^{-1}\sum_{i<n}\text{Tr}(\rho_{n}P^{n}_{i})\\
    &=n^{-1}\sum_{i<n}\sum_{k<2^{n}}\alpha_{k} \big< \psi^{k}|P^{n}_{i}|\psi^{k}\big>\\
    &=
    \label{eq:last}
    \sum_{k<2^{n}}\alpha_{k}\big(n^{-1}\sum_{i<n} \big< \psi^{k}|P^{n}_{i}|\psi^{k}\big>\big).
\end{align}
For each fixed $k$ and $i$, by the same argument as in equation (\ref{eq:1}) and using that $|c^{k}_p|,|P^{n}_{i} \psi^{k}_p| \leq 1$ we have that
\begin{align}
\label{eq:boundd}
\big<\psi^{k}|P^{n}_{i}|\psi^{k}\big> \leq 
 |c^{k}_o|^{2}|P^{n}_{i} \psi^{k}_o|^{2}+ |c^{k}_p|^{2}  |P^{n}_{i} \psi^{k}_p|^{2} + 2|c^{k}_o|  |P^{n}_{i} \psi^{k}_o| 
.
\end{align}
Using this, we bound the term in parentheses in equation (\ref{eq:last}) for each $k$. As $k$ is fixed, replace  $\psi^k$ and $c^k$ in equation (\ref{eq:boundd}) by $\psi$ and $c$ respectively for convenience.

\begin{align}
   & n^{-1}\sum_{i<n} \big< \psi|P^{n}_{i}|\psi\big>\\
   &\leq |c_o|^{2}n^{-1}\sum_{i<n} |P^{n}_{i} \psi_o|^{2}+|c_p|^{2}n^{-1}\sum_{i<n} |P^{n}_{i} \psi_p|^{2}+2|c_o| n^{-1}\sum_{i<n} |P^{n}_{i} \psi_o|\\
   &\leq |c_o|+n^{-1}\sum_{i<n} |P^{n}_{i} \psi_p|^{2}+2|c_o|\\
   \label{eq:lasst}
   &=n^{-1}\sum_{i<n} |P^{n}_{i} \psi_p|^{2}+3|c_o|.
\end{align}
We used convexity and $|c_p|^{2}\leq 1,|c_o|\leq 1,|P^{n}_{i} \psi_o|\leq 1$ when obtaining the last inequality. Let $\psi:=\psi_p$ and for a fixed $i<n$, let $P:=P^{n}_{i}$ and consider the summand, $|P\psi|^2$ in the sum in equation (\ref{eq:lasst}) (we suppressed the indices merely for convenience). Since $\psi \in  \text{range}(S_n)^{\perp}= $ span$(C_{n}^{c})$, let $a_{\sigma}$ be complex numbers such that 
\[\sum_{\sigma \notin C_{n}} a^{2}_{\sigma} =1,\]
and
\[\psi=\sum_{\sigma \notin C_{n}} a_{\sigma}\sigma.\]
Using that $P^{*}=P$ and $P=P^2$,
\begin{align}
    & |P\psi|^2=\big<P \psi | P \psi \big> =\big< \psi | P \psi \big>\\
    &= \big< \sum_{\sigma \notin C_{n}} a_{\sigma}\sigma  |  \sum_{\tau \notin C_{n}} a_{\tau} P \tau \big>\\
    &=
    \label{eq:tricky}
    \sum_{\sigma \notin C_{n}}\sum_{\tau \notin C_{n}}  a^{*}_{\sigma}a_{\tau} \big<\sigma  |   P \tau \big>.
\end{align}
Note that $  P \tau = \tau$  or $ P \tau=0$ and that $\big<\sigma  |   \tau \big>= \delta_{\sigma=\tau}$. So, $\big<\sigma  |   P \tau \big>$ is zero whenever $\sigma \neq \tau$ (Here we used that the orthonormal vectors spanning $C^{c}_{n}$ are eigenvectors of $P$).
So, (\ref{eq:tricky}) becomes, \[|P\psi|^2\leq \sum_{\sigma \notin C_{n}}  a^{2}_{\sigma} \big<\sigma  |   P \sigma \big>.
\]
Using this and reinserting the indices, the first term in (\ref{eq:lasst})  is bounded above by
\[n^{-1}\sum_{i<n}\sum_{\sigma \notin C_{n}}  (a^{k}_{\sigma})^{2} \big<\sigma  |   P^{n}_{i} \sigma \big>.\]
Finally, putting this back in (\ref{eq:last}), 
\begin{align}
    &\sum_{k<2^{n}}\alpha_{k}\big(n^{-1}\sum_{i<n} \big< \psi^{k}|P^{n}_{i}|\psi^{k}\big>\big)\\
    &\leq \sum_{k<2^{n}}\alpha_{k}\big(n^{-1}\sum_{i<n}\sum_{\sigma \notin C_{n}}  (a^{k}_{\sigma})^{2} \big<\sigma  |   P^{n}_{i} \sigma \big>+ 3|c^{k}_{o}|)\\
    &\label{eq:zeroth}
    \leq \sum_{k<2^{n}}\alpha_{k}\big(\sum_{\sigma \notin C_{n}}(a^{k}_{\sigma})^{2}n^{-1}\sum_{i<n}   \big<\sigma  |   P^{n}_{i} \sigma \big>)+\sum_{k<2^{n}}\alpha_{k} 3|c^{k}_{o}|\\
    \label{eq:1st}
    &\leq \big(\sum_{k<2^{n}}\alpha_{k}\sum_{\sigma \notin C_{n}}(a^{k}_{\sigma})^{2}(\delta/2+0.5)\big)+ 3\sum_{k<2^{n}}\alpha_{k}|c^{k}_{o}|\\
    &\leq(\delta/2+0.5) \big(\sum_{k<2^{n}}\alpha_{k}\sum_{\sigma \notin C_{n}}(a^{k}_{\sigma})^{2}\big)+ 3\sum_{k<2^{n}}\alpha_{k}|c^{k}_{o}|\\
    &\leq(\delta/2+0.5) + 3\sum_{k<2^{n}}\alpha_{k}|c^{k}_{o}|.
\end{align}

In getting (\ref{eq:1st}) we used the definition of $C_n$. In the last step we used that $\sum_{\sigma \notin C_{n}}(a^{k}_{\sigma})^{2}=1 $ for all $k$ and convexity. In summary, we have shown that for infinitely many $n$,
\[0.5 + \delta <  n^{-1}\sum_{i<n}\text{Tr}(\rho_{n}P^{n}_{i})=\sum_{k<2^{n}}\alpha_{k}\big(n^{-1}\sum_{i<n} \big< \psi^{k}|P^{n}_{i}|\psi^{k}\big>\big) \leq (\delta/2+0.5) + 3\sum_{k<2^{n}}\alpha_{k}|c^{k}_{o}|.
\]
So, by Jensen's inequality \[\delta^{2}/36 < \big(\sum_{k<2^{n}}\alpha_{k}|c^{k}_{o}|\big)^{2} \leq \sum_{k<2^{n}}\alpha_{k}|c^{k}_{o}|^{2}= \text{Tr}(\rho_{n}S_{n}),\]
for infinitely many $n$. So, $\rho$ fails a quantum Schnorr test at $\delta^{2}/36$, a contradiciton. Now if $\exists^{\infty}n$, with $n^{-1}\sum_{i<n}$Tr$(\rho_{n}P^{n}_{i})< -\delta + 0.5$ then define 
\[Q^{n}_{i}=(P^{n}_{i})^{\perp}:= \sum_{\sigma: |\sigma|=n, \sigma(i)=0}|\sigma\big>\big<\sigma|.\]

Note that Tr$(\rho_{n}Q^{n}_{i})+$Tr$(\rho_{n}P^{n}_{i})=1$ for all $i,n$. So, for infinitely many $n$,
$1=n^{-1}(\sum_{i<n}$Tr$(\rho_{n}P^{n}_{i})+ \sum_{i<n}$Tr$(\rho_{n}Q^{n}_{i})) < -\delta + 0.5+$Tr$(\rho_{n}Q^{n}_{i}))$. I.e, $n^{-1}\sum_{i<n}$Tr$(\rho_{n}Q^{n}_{i})>\delta + 0.5$ for infinitely many $n$. Now, we can repeat the proof as in case 1 with Q replacing $P$ and $0$s replacing the $1$s.
\end{proof}
\section{A Shannon\textendash McMillan\textendash Breiman Theorem for quantum Schnorr randoms}
\label{sec:smb}
The Shannon\textendash McMillan\textendash Breiman (SMB) theorem for bitstrings roughly says that for an ergodic measure, $\mu$, on Cantor space the empirical entropy for $\mu$ almost every trajectory (infinite bitstring) equals the entropy of $\mu$. There have also been effective versions of the SMB. For example, it has been shown that the exception set for the SMB theorem in the classical setting can be covered by a Martin-L{\"o}f test \cite{hoyrup:LIPIcs:2012:3391}. In the quantum setting, where we do not have a notion of `almost every', we may replace `$\mu$ almost every trajectory' by `every $\mu$ Schnorr random state' as we do here. A special case of the SMB theorem for infinite sequences of qubits was first studied by Nies and Tomamichel \cite{logicblog}.
To formalize a $\mu$ Schnorr random state in the quantum setting, we need a definition
\begin{defn}
A computable sequence of special projections is a $\mu$\emph{ quantum Schnorr test} if $\sum_{k\in \omega} \mu(S^{k})$ is computable.
\end{defn}
A state $\rho$ is $\mu$  quantum Schnorr random if it passes all $\mu$  quantum Schnorr tests. A similar definition for quantum MLR states was made by Nies and Tomamichel  \cite{logicblog}. Intuitively, a $\mu$  quantum Schnorr random state is a `trajectory' in the state space \cite{unpublished} which is random in the sense of $\mu$.

\begin{thm}
Let $\mu=(\mu_{n})_{n}$ be a state of the form $\mu_{n}=\otimes_{1}^{n} M$ for an $M$ of the form
\[
\begin{bmatrix}

p & 0 \\
0 & 1-p\\

\end{bmatrix}
,\]
for some computable $p \in (0,1)$. If $\rho$ is $\mu$ quantum Schnorr random, then\\ lim$_{n}n^{-1}$Tr$(-\rho_{n}$log$(\mu_{n}))$= $h(\mu)$, the von-Neumann entropy of $M$.
\end{thm}
Intuitively, the theorem says that along any $\mu$ Schnorr random state, $\rho$, the empirical entropy, $n^{-1}$Tr$(-\rho_{n}$log$(\mu_{n}))$ limits to the entropy of $\mu$, which equals that of $M$ as $\mu$ is a product tensor. 
\begin{proof}
Let $M$ be as given and first, assume that $p\leq 0.5$. We prove it by contradiction. Define $L_{n}= -$ log$\mu_{n}$ for all $n$ and $h:=h(\mu)$. Suppose $\rho$ is quantum Schnorr random but there is a $\delta$ such that either $\exists^{\infty}n$, with $n^{-1}$Tr$(\rho_{n}L_{n})> \delta + h$ or  $\exists^{\infty}n$, with $n^{-1}$Tr$(\rho_{n}L_{n})< -\delta + h$. Suppose first that the former holds. The proof is similar to that of the law of large numbers, but different techniques are needed as $L_n$ is not a projection. $\sigma$ will always be used to denote classical bitstrings. For $\sigma$ of length $n$, $\big<\sigma|\mu_{n}|\sigma\big>=\mu(\sigma,\sigma)=p^{k}(1-p)^{n-k}$ where $k=$ numbers of zeros in $\sigma$. So, $\mu$ can be thought of a i.i.d.\ measure on Cantor space assigning $\mu(0)=p, \mu(1)=1-p$.

Define for all $n$,
\[C_{n}= \left\{\sigma : |\sigma|=n, -n^{-1}\text{log}\mu(\sigma)> \delta/2 + h\right\}.\] 

Let $S_n$ be the special projection, \[S_{n}:= \sum_{\sigma \in C_{n}} |\sigma\big>\big<\sigma|.\]
$(S_n)_n$ is a computable sequence since we may let $\delta$ be rational. By the Chernoff bound, $\mu(S_{n})=2^{-n}|C_{n}| \leq 2\text{exp}(-0.5 n\delta^{2})$ for all $n$. So, $\sum_{n}\mu(S_{n})$ is computable showing that $(S_n)_n$ is a $\mu-$quantum Schnorr test.

For all $n$, let 
\[\rho_n = \sum_{k< 2^n} \alpha_{k}|\psi^{k}_{n}\big> \big< \psi^{k}_{n}| \]  
for $\alpha_k$ non-negative real numbers with $\sum_{k<2^n} \alpha_k =1$ and for each $k$, $|\psi^{k}_{n}\big> \in \mathbb{C}^{2^n}$ and $||\psi^{k}_{n}||=1$.
Fix an $n$ is such that $n^{-1}$Tr$(\rho_{n}L_{n})> \delta + h$.
We will drop the $n$ subscript of $\psi^{k}_{n}$ as the $n$ is fixed.
For any $k< 2^{n}$ we can decompose $\psi^{k}$ as,
\begin{align}
    \psi^{k} = c^{k}_o\psi^{k}_o + c^{k}_p\psi^{k}_p
\end{align}
where $\psi^{k}_o \in \text{range}(S_{n})$ and $\psi^{k}_p \in \text{range}(S_n)^{\perp}$
are unit vectors and $c^{k}_o, c^{k}_p \in \mathbb{C}$ satisfy $|c^{k}_0|^{2}+|c^{k}_p|^{2}= 1$.
We find a lower bound, for Tr$(\rho_n S_{n})$ = $\sum_{k< 2^n}\alpha_{k} |c^{k}_o|^{2}$ independent of $n$.
\begin{align}
    & n^{-1} \text{Tr}(\rho_{n}L_{n})\\
    &=n^{-1}\sum_{k<2^{n}}\alpha_{k} \big< \psi^{k}|L_{n}|\psi^{k}\big>\\
    &=
    \label{eq:last2}
    \sum_{k<2^{n}}\alpha_{k}\big(n^{-1}  \big< \psi^{k}|L_{n}|\psi^{k}\big>\big).
\end{align}
Fix a $k$ and suppress it in the indices (i.e for example, let $\psi=\psi^k$). By Cauchy-Schwarz and by the self-adjointness and positivity of $L_n$,
\begin{align*}
\big<\psi^{k}|L_{n}|\psi^{k}\big> 
 &\leq  |c_o|^{2}\big< \psi_o|L_{n}\psi_o\big>+ |c_p|^{2} \big< \psi_p|L_{n}\psi_p\big>+ 2| c_{o}||c_{p}||\big<\sqrt{L_{n}} \psi_p|\sqrt{L_{n}}\psi_o\big>|\\
 &\leq |c_o|^{2}\big< \psi_o|L_{n}\psi_o\big>+ |c_p|^{2} \big< \psi_p|L_{n}\psi_p\big>+ 2| c_{o}||\sqrt{L_{n}} \psi_p||\sqrt{L_{n}}\psi_o|\\
 &\leq |c_o|^{2}\big< \psi_o|L_{n}\psi_o\big>+ |c_p|^{2} \big< \psi_p|L_{n}\psi_p\big>+ 2| c_{o}|(||\sqrt{L_{n}}||_{2})^{2}
,
\end{align*}
where $|| \cdot ||_{2}$ denotes the $L_2$ operator norm. $M_{n} $, the maximum element of the set \[\left\{ -j\text{log} p- (n-j)\text{log}(1-p): 0 \leq j \leq n\right\}\] is the largest eigenvalue of $L_n$ and so, $\sqrt{M_{n}}$ is the largest eigenvalue of $\sqrt{L_{n}}$. Noting that the $L_2$ norm of a real diagonal matrix is equal to its largest eigenvalue and that the Rayleigh quotient of a Hermitian matrix is bounded above by the largest eigenvalue, we see that

\begin{align*}
    \big<\psi^{k}|L_{n}|\psi^{k}\big>&\leq |c_o|^{2}\big< \psi_o|L_{n}\psi_o\big>+ |c_p|^{2} \big< \psi_p|L_{n}\psi_p\big>+ 2| c_{o}|\sqrt{M_{n}}^{2} \\
    & \leq |c_o|^{2}M_{n}+ |c_p|^{2} \big< \psi_p|L_{n}\psi_p\big>+ 2| c_{o}|M_{n}\\
    &\leq 3|c_o|M_{n}+ |c_p|^{2} \big< \psi_p|L_{n}\psi_p\big>.
\end{align*}
By this and noting that $n^{-1}M_n \leq \theta = $ the maximum of $-$log($p$) and $-$log($1-p$), we get an upper bound for the term in parentheses in equation (\ref{eq:last2}) for each $k$:

\begin{align}
\label{eq:lasst2}
   & n^{-1}\big< \psi^{k}|L_{n}|\psi^{k}\big> \leq  3|c^{k}_o|\theta + n^{-1} \big< \psi^{k}_p|L_{n}\psi^{k}_p\big>.
\end{align}
Since $\psi^{k}_p \in  \text{range}(S_n)^{\perp}= $ span$(C_{n}^{c})$, there are $a_{\sigma}$s such that 
\[\sum_{\sigma \notin C_{n}} a^{2}_{\sigma} =1,\]
and
\[\psi^{k}_{p}=\sum_{\sigma \notin C_{n}} a^{k}_{\sigma}\sigma.\]
Letting $\psi = \psi^{k}_{p}$ and dropping the $k$ indices for convenience.
\begin{align}
   & n^{-1}\big< \psi | L_{n} \psi \big>\\
    &= n^{-1}\big< \sum_{\sigma \notin C_{n}} a_{\sigma}\sigma  |  \sum_{\tau \notin C_{n}} a_{\tau} L_{n} \tau \big>\\
    &=
    \label{eq:tricky2}
    n^{-1}\sum_{\sigma \notin C_{n}}\sum_{\tau \notin C_{n}}  a^{*}_{\sigma}a_{\tau} \big<\sigma  |   L_{n} \tau \big>.
\end{align}
As $L_n$ is diagonal and $C_n$ is composed of classical bitstrings, equation (\ref{eq:tricky2}) becomes, \[\leq \sum_{\sigma \notin C_{n}}  |a_{\sigma}|^{2} n^{-1}\big<\sigma  | L_{n} \sigma \big>=-\sum_{\sigma \notin C_{n}}  |a_{\sigma}|^{2} n^{-1}\text{log}\mu(\sigma)
\]

\[\leq \sum_{\sigma \notin C_{n}}  |a_{\sigma}|^{2} (\delta/2 + h) \leq \delta/2 + h.\]
We used here that the $\sigma$th entry along the diagonal of $L_n$ is $-\text{log}\mu(\sigma)=-k\text{log} p- (n-k)\text{log}(1-p) $ where $k=$ number of zeros in $\sigma$ and the definition of $C_n$. This and (\ref{eq:lasst2}) gives
that,
\begin{align}
  & n^{-1}\big< \psi^{k}|L_{n}|\psi^{k}\big> \leq  3|c^{k}_o|\theta + \delta/2 + h.
\end{align}
Finally, putting this back in (\ref{eq:last2}), 
\begin{align*}
    &\sum_{k<2^{n}}\alpha_{k}\big(n^{-1} \big< \psi^{k}|L_{n}|\psi^{k}\big>\big)\leq \sum_{k<2^{n}}\alpha_{k}\big(3|c^{k}_o|\theta + \delta/2 + h)\leq(\delta/2+h) + 3 \theta \sum_{k<2^{n}}\alpha_{k}|c^{k}_{o}|.
\end{align*}
In summary, we have shown that for infinitely many $n$,
\[h + \delta <  n^{-1}\sum_{i<n}\text{Tr}(\rho_{n}P^{n}_{i}) \leq (\delta/2+h) + 3 \theta \sum_{k<2^{n}}\alpha_{k}|c^{k}_{o}|.
\]
So, by Jensen's inequality \[\delta^{2}/36\theta^{2} < \big(\sum_{k<2^{n}}\alpha_{k}|c^{k}_{o}|\big)^{2} \leq \sum_{k<2^{n}}\alpha_{k}|c^{k}_{o}|^{2}= \text{Tr}(\rho_{n}S_{n}),\]
for infinitely many $n$. So, $\rho$ fails a $\mu-$quantum Schnorr test; a contradiciton. We need to now consider the other case: 
$\exists^{\infty}n$, with $n^{-1}$Tr$(\rho_{n}L_{n})< -\delta + h$. Define $M'$ to be the reflection of $M$. I.e., $M'$ is
\[
\begin{bmatrix}

1-p & 0 \\
0 & p\\

\end{bmatrix}
,\]
and $\mu^{'}$ is the i.i.d.\ measure on Cantor space given by $M'$. Let $L^{'}_{n}:=-$log $\otimes_{1}^{n}M'$. Note that for $\sigma$ of length $n$, $\big<\sigma|\otimes_{1}^{n}M'|\sigma\big>=p^{n-k}(1-p)^{k}$ where $k=$ numbers of zeros in $\sigma$. Letting $Q_{n}:=L_{n}+L^{'}_n$, for any length $n$ string $\sigma$ having $k$ many zeros,   \[\big<\sigma|Q_{n}|\sigma\big>=-(n-k)\text{log}(1-p)-k\text{log}(p)- (n-k)\text{log}( p)-k\text{log}(1-p)=-n(\text{log}(p)+\text{log}(1-p)).\] So, $Q_{n}=-n(\text{log}(p)+\text{log}(1-p))I_{2^{n}}$.

\begin{align}
    \label{eq:entropy}
    n^{-1}\text{Tr}(\rho_{n}L_{n})+n^{-1}\text{Tr}(\rho_{n}L^{'}_{n})=n^{-1}\text{Tr}(\rho_{n}Q_{n})=- \text{log}(p)-\text{log}(1-p).
\end{align}

We see that, 
\begin{align*}
    & 2h \leq - \text{log}(p)-\text{log}(1-p)\\
    &\iff 2p\text{log}(p)+2(1-p)\text{log}(1-p) \geq \text{log}(p)+\text{log}(1-p)\\
    &\iff \text{log}(1-p) \geq\text{log}(p) \iff p \leq 0.5.\\
\end{align*}
(We used here that $p \leq 0.5$ and hence that $ (1-2p)\geq 0 $.)
So, for one of the infinitely many $n$ such that, $n^{-1}$Tr$(\rho_{n}L_{n})< -\delta + h$, equation (\ref{eq:entropy}) gives that
\begin{align*}
&(-\delta+h) + n^{-1}\text{Tr}(\rho_{n}L^{'}_{n})\\
&>n^{-1}\text{Tr}(\rho_{n}L_{n})+n^{-1}\text{Tr}(\rho_{n}L^{'}_{n})= -\text{log}(1-p)-\text{log}(p) \geq 2h.
\end{align*}
So, there are infinitely many $n$s with \[n^{-1}\text{Tr}(\rho_{n}L^{'}_{n})>\delta+h.\]
Since $M'$ has the same entropy as $M$, we can repeat the proof of the former case using $L^{'}_{n}, \mu^{'}$ instead of $L_{n}, \mu$ respectively. This completes the proof for $p\leq 0.5$. If $p>0.5$, then repeat the proof for $p\leq 0.5$ with $1-p\leq 0.5$ replacing $p$ [The first case has the same proof as it doesn't depend on the value of $p$. When proving the second case, $M'$ is diag$(p,1-p)$ and the proof goes through since $1-p\leq 0.5$.]
\end{proof}

\chapter{Prefix-free quantum Kolmogorov Complexity}\label{Prefix-free quantum Kolmogorov Complexity}
\section{Introduction}
The results in this chapter have been published already in the literature \cite{BHOJRAJ2021}.
With the intent of developing a quantum version of $K$, we introduce $QK$, a notion of descriptive complexity for density matrices using classical prefix-free Turing machines. Many connections between $K$ and Solovay and Schnorr randomness in the classical theory turn out to have analogous connections connections between $QK$ and weak Solovay and quantum Schnorr randomness. 

To the best of our knowledge, the current work is the only one to study the incompressibility of initial segments (in the sense of prefix-free classical Turing machines) of weak Solovay and quantum Schnorr random states. Nies and Scholz have explored connections between quantum Martin-L{\"o}f randomness and a version of $QC$ using unitary (quantum) machines\cite{unpublished}. 
We give an overview of the main points in this chapter. 

In Section \ref{sec:2} we introduce quantum-K ($QK$) for density matrices and some of its properties. Theorem \ref{thm:67} (generalized in Lemma \ref{lem:K}) shows that $QK$ agrees with $K$ on the classical qubitstrings. Theorem \ref{thm:8} is a tight upper bound for $QK$ similar to that for $K$. Theorem \ref{thm:count} is a counting condition similar to that for $QC$ \cite{Berthiaume:2001:QKC:2942985.2943376}, $C$ and $K$ \cite{misc, misc1}.

Section \ref{sec:4} reviews quantum algorithmic randomness: a recently developed \cite{bhojraj2020quantum, unpublished} theory of randomness for states (infinite qubitstrings) using quantum versions of the classical `effectively null set'.

Section \ref{sec:3}, the main focus of this paper, connects $QK$ with two quantum algorithmic randomness notions: weak Solovay randomness and quantum Schnorr randomness, defined in Section \ref{sec:4}. Two important characterizations show that the initial segments of Martin-L{\"o}f randoms (equivalently, of Solovay randoms) are asymptotically incompressible in the sense of $K$: the Chaitin characterization (See \cite{Chaitin:1987:ITR:24912.24913} and theorem 3.2.21 in \cite{misc}),
\[X  \text{ is Martin-L{\"o}f random} \iff \text{lim}_{n}  K(X\upharpoonright n)-n= \infty,\]
and the Levin\textendash Schnorr characterization (See theorem 3.2.9 in \cite{misc}),
\[X  \text{ is Martin-L{\"o}f random} \iff \exists c \forall n [ K(X\upharpoonright n)>n-c].\]
(Characterizations having the former form  will be called `Chaitin type' and those having the latter form will be called `Levin\textendash Schnorr type').
We investigate the extent to which these classical characterizations carry over to weak Solovay randoms and $QK$.

Theorem \ref{thm:7} is a Chaitin type of characterization of 
weak Solovay randomness ($\rho$ is weak Solovay random $\iff$ lim$_{n}QK^{\epsilon}(\rho_{n})-n = \infty$).
This shows that the Levin\textendash Schnorr condition ($\forall n[QK^{\epsilon}(\rho_{n})>^{+}n]$) is implied by weak Solovay randomness. 

Theorem \ref{thm:39} shows both Chaitin and Levin\textendash Schnorr type characterizations of weak Solovay randomness when restricting attention to a specific class of states. It is worth noting that Theorem \ref{thm:39} uses the proof of our main result (Theorem 2.11) in \cite{bhojraj2020quantum}. 

For general states, subsection \ref{subsect:weak} shows that the Levin\textendash Schnorr condition implies something slightly weaker than weak Solovay randomness. 

While $K$ plays well with Solovay randomness, $K_C$, a version of $K$ using a computable measure machine, $C$ (a prefix-free Turing machine whose domain has computable Lebesgue measure) gives a Levin\textendash Schnorr  characterization of Schnorr randomness (See theorem 7.1.15 in \cite{misc1}). Motivated by this, we introduce $QK_C$ a version of $QK$ using computable measure machines in subsection \ref{schn}. 

It turns out that $QK_C$ yields not just a Levin\textendash Schnorr type (Theorem \ref{thm:schnor}), but also a Chaitin type (Theorem \ref{thm:Chaitin}) characterization of quantum Schnorr randomness.

Theorem \ref{thm:Chaitin} together with Theorem \ref{thm:677} and lemma 3.9 in \cite{bhojraj2020quantum}, implies that Schnorr randoms have a Chaitin type characterization in terms of $K_C$ (Theorem \ref{class}). So, results in the quantum realm imply a new result in the classical setting.

In summary, we introduce $QK$ and show that the initial segments of weak Solovay random and quantum Schnorr random states are incompressible in the sense of $QK$.

\section{The Definition and Properties of QK}
\label{sec:2}
 We assume familiarity with the notions of density matrix (See for example, \cite{Nielsen:2011:QCQ:1972505}), prefix-free Kolmogorov complexity ($K$) and $\mathbb{U}$, the universal prefix-free (or self-delimiting) Turing machine (See \cite{misc,misc1, DBLP:series/txtcs/Calude02}).
 
 The output of $\mathbb{U}$ can be interpreted as unordered tuples of complex algebraic vectors (equivalently, finite subsets of natural numbers). 
 
The notation $\mathbb{U}(\sigma)\downarrow = F$ means that $ \mathbb{U}(\sigma)$ outputs the index of $F$ with respect to some fixed canonical indexing of finite subsets of the naturals. We will never use an ordering on the elements of $F$ in any of our arguments: $F$ will be used to define an orthogonal projection : $\sum_{v \in F} |v\big>\big<v|$ which clearly does not depend on an ordering on $F$. 

As explained in \cite{unpublished}, the quantum analogue of a bitstring of length $n$ is a density matrix on $\mathbb{C}^{2^n}$.
 
For a density matrix, $\tau$, let $|\tau|$ denote the $s$ such that $\tau$ is a transformation on $\mathbb{C}^{2^s}$. For any $n$, $\mathbb{C}^{2^n}_{alg}$ is the space of elements of $\mathbb{C}^{2^n}$ with complex algebraic entries.

Logarithms will always be base 2. The notation $\leq^+ , \geq^+, =^+$ will be used for `upto additive constant' relations.

For $\epsilon >0 $, $QK^{\epsilon}(\tau)$ is defined to be
\begin{defn}
\label{defn:4}
$QK^{\epsilon}(\tau) := $ inf  $\{|\sigma|+$log$|F|: \mathbb{U}(\sigma)\downarrow = F$, a orthonormal set in $\mathbb{C}^{2^{|\tau|}}_{alg}$ and $\sum_{v \in F} \big<v |\tau |v \big> > \epsilon \}$

\end{defn}
 The term $\sum_{v \in F} \big<v |\tau |v \big>$ is the squared length of the `projection of $\tau$ onto span($F$)' which also equals the probability of getting an outcome of `1' when measuring $\tau$ with the observable given by the Hermitian projection onto span$(F)$\cite{Nielsen:2011:QCQ:1972505}. Although it is useful to intuitively think of $\sum_{v \in F} \big<v |\tau |v \big>$ as the `projection of $\tau$ onto span($F$)', we use quotes as $\tau$ is a convex combination of possibly \emph{multiple} unit vectors, while the notion `projection onto a subspace' refers usually to a single vector. 
 
 Note that for a given $\tau$, $QK^{\epsilon}(\tau)$ is determined by the classical prefix-free complexities and dimensions of those subspaces, span$(F)$, such that the projection of $\tau$ onto span$(F)$ has squared length atleast $\epsilon$. I.e., $QK^{\epsilon}(\tau)$ depends only on the $K$-complexities and ranks of those projective measurements of $\tau$ such that the probability of getting an outcome of `$1$' is atleast $\epsilon$. Roughly speaking, $QK^{\epsilon}(\tau)$ depends on the dimensions and prefix-free complexities of subspaces which are $\epsilon$ `close' to $\tau$. 
 
 This is in contrast to $QC^{\epsilon}(\tau)$ which depends on the quantum complexities of density matrices, not classical prefix-free complexities of subspaces, which are $\epsilon$ close to $\tau$ (Recall that $QC$ is based on quantum Turing machines) \cite{Berthiaume:2001:QKC:2942985.2943376}. Also, while the rank of the approximating projection is taken into consideration in $QK$, the rank of the approximating density matrix is not taken into account in $QC$.
 
 So, $QC^{\epsilon}(\tau)$ quantifies the quantum complexity of approximating $\tau$ by density matrices upto $\epsilon$ while $QK^{\epsilon}(\tau)$ measures the sum of the prefix-free complexity and the logarithm of the dimension of subspaces $\epsilon$ close to $\tau$. 

A test demonstrating the quantum non-randomness of a state, $\rho$ uses computable sequences of projections of `small rank' which are $\epsilon$-close to initial segments (density matrices) of $\rho$ (See Section \ref{sec:4}). It hence seems plausible that a complexity measure for a density matrix, $\tau$ must reflect the complexities and ranks of projections $\epsilon$-close to $\tau$ in order to play well with quantum randomness notions for states.

We mention that our $QK$ is entirely different from the $ \overline{QK_M} $ and $\overline{QK^{\delta}_M}$ notions defined in Definition 3.1.1 in \cite{Mller2007QuantumKC} using quantum Turing machines.

$QK^{\epsilon}$ would not be a `natural' complexity notion for density matrices if the following theorem did not hold:
\begin{thm}
\label{thm:67}
 Fix a rational $\epsilon$. $K(\sigma)= QK^{\epsilon}(|\sigma\big>\big<\sigma|)$ holds for all classical bitstrings $\sigma$, upto an additive constant depending only on $\epsilon$.
\end{thm}
We isolate here a simple but useful property which will be used for proving Theorem \ref{thm:67}. 
\begin{lem}
\label{lem:30}
Let $n$ be a natural number, $E=(e_{i})_{i=1}^{2^{n}}$ be any orthonormal basis for $\mathbb{C}^{2^{n}}$ and $F$ be any Hermitian projection matrix acting on $\mathbb{C}^{2^{n}}$. For any $\delta>0$, let  \[S^{\delta}_{E,F}:= \{e_{i} \in E: \big<e_{i}|F|e_{i}\big> > \delta \}.\]
Then, $|S^{\delta}_{E,F}| < \delta^{-1} \text{Tr}(F)$.
\end{lem}
\begin{proof}
Note that since $F$ is a Hermitian projection, $\big<e_{i}|F|e_{i}\big>=\big<Fe_{i}|F e_{i}\big> = |F e_{i}|^{2}\geq 0$. So,
\[\delta|S^{\delta}_{E,F}| < \sum_{e_{i} \in S^{\delta}_{E,F}}\big<e_{i}|F|e_{i}\big> \leq \sum_{i \leq 2^{n}}\big<e_{i}|F|e_{i}\big> = \text{Tr}(F).\]
\end{proof}
\begin{proof}

We now prove Theorem \ref{thm:67}, the idea of which is as follows: Given a classical bitstring and a subspace `close' to it, we find a subspace spanned only by classical bitstrings `close' to this subspace. Then we compress each of the spanning classical strings and show that the string we began with must be one of these.
Fix a rational $\epsilon$. Consider the machine $P$ doing the following:
\begin{enumerate}
    \item On input $\pi$, $P$ searches for $\pi=\sigma \tau$ such that 
    $\mathbb{U}(\sigma)\downarrow = F$, an orthonormal set, $F \subseteq \mathbb{C}^{2^{n}}_{alg}$ for some $n$ and $|\tau|= \lceil \text{log}(\epsilon^{-1}|F|)\rceil$.
    \item Letting $O:= \sum_{v\in F} |v\big>\big<v|$ and $E$, the standard basis of $\mathbb{C}^{2^{n}}$, find the set $S^{\epsilon}_{E,O}$ from Lemma \ref{lem:30}.
    \item Take a canonical surjective map $g$ from the set of bitstrings of length $\lceil \text{log}(\epsilon^{-1}|F|)\rceil$ onto $S^{\epsilon}_{E,O}$. ($g$ exists since  $|S^{\epsilon}_{E,O}|<\epsilon^{-1}|F|$ by \ref{lem:30}).  Output $g(\tau)$.
\end{enumerate}
We first show that $P$ is prefix-free. Suppose $\pi$ and $\pi'$ are in the domain of $P$ and $\pi \preceq \pi'$. Then, $\pi=\sigma \tau$ and $\pi'=\sigma'\tau'$ and $\sigma$ and $\sigma'$ are in the domain of $\mathbb{U}$. $\pi \preceq \pi'$ implies that $\sigma \preceq \sigma'$ or $\sigma' \preceq \sigma$. But as $\mathbb{U}$ is prefix-free, $\sigma = \sigma'$ must hold. Since the computations $\mathbb{U}(\pi)$, and $\mathbb{U}(\pi')$ and not stuck forever at (1), it must be that $|\tau|= \lceil \text{log}(\epsilon^{-1}|F|)\rceil$ and $|\tau'|= \lceil \text{log}(\epsilon^{-1}|F'|)\rceil$ where $\mathbb{U}(\sigma)\downarrow=F=F'=\mathbb{U}(\sigma')\downarrow$. So, $\tau$ and $\tau'$ have the same length implying that $\pi=\pi'$.
\\
Now, let $\sigma \in 2^n$ be any classical bitstring. Let $\lambda$ and $F \subseteq \mathbb{C}^{2^{n}}_{alg}$ a orthonormal set such that $|\lambda|+ \text{log}(|F|)=QK^{\epsilon}(|\sigma\big>\big<\sigma|)$,  $\sum_{v \in F} \big<v |\sigma\big>\big<\sigma|v \big> > \epsilon$ and $\mathbb{U}(\lambda)=F$. Let $O:= \sum_{v\in F} |v\big>\big<v|$. Note that since
$\epsilon<\sum_{v \in F} \big<v |\sigma\big>\big<\sigma|v \big> = \big<\sigma|O|\sigma\big>$, $\sigma \in S^{\epsilon}_{E,O}$ where $E$ is the standard basis. Let $\tau$ be a length $ \lceil \text{log}(\epsilon^{-1}|F|)\rceil$ string such that $g(\tau)=\sigma$. Then, we see that $P(\lambda \tau) = \sigma $. \[\text{K}(\sigma)\leq^+   |\lambda|+|\tau|\leq^+  |\lambda|+\text{log}(\epsilon^{-1})+ \text{log}(|F|) \leq^+ QK^{\epsilon}(|\sigma\big>\big<\sigma|)\]
This establishes one direction. Note that the additive constant depends on $\epsilon$. The constant (zero) in the other direction turns out to be independent of $\epsilon$: Given some classical bitstring $\sigma$, let $\mathbb{U}(\pi)=\sigma$ and $|\pi|=\text{K}(\sigma)$. Then, letting $F=\{\sigma\}$ in \ref{defn:4}, $QK^{\epsilon}(\sigma)\leq QK^{1}(\sigma)\leq |\pi|=K(\sigma)$, for any $\epsilon>0$. 

\end{proof}
 \begin{defn}
\label{defn:78}
A `system' $B= ((b^{n}_{0},b^{n}_{1}))_{n\in \mathbb{N}}$ is a sequence of orthonormal bases for $\mathbb{C}^{2}$ such that each $b^{n}_{i}$ is complex algebraic and the sequence $ ((b^{n}_{0},b^{n}_{1}))_{n\in \mathbb{N}}$ is computable.
\end{defn} 
\begin{remark}\label{rem:K}
Let $B= ((b^{n}_{0},b^{n}_{1}))_{n\in \mathbb{N}}$ be a  system, as in \ref{defn:78}. Let $A_B$ be the set of all pure states, $\sigma$ such that $\sigma$ is a product tensor of elements from $B$. For example, $b^{1}_{0}\otimes b^{2}_{1}\otimes b^{3}_{1}\otimes b^{4}_{0} \in A_B$. Then, the previous theorem generalizes to the following: Fix a rational $\epsilon$, a $B$ and a $A_B$ as above. $K(\sigma)=^+ QK^{\epsilon}(|\sigma\big>\big<\sigma|)$ holds for all $\sigma \in A_B$, upto an additive constant depending only on $B$ and $\epsilon$. Here, $K(\sigma)$ is defined in the obvious way. For example, $K(b^{1}_{0}\otimes b^{2}_{1}\otimes b^{3}_{1}\otimes b^{4}_{0})=K(0110)$. This is proved by replacing $S^{\epsilon}_{E,O}$ with $S^{\epsilon}_{B,O}$ in the proof of Theorem \ref{thm:67}.
\end{remark}
The following lemma can be proved similarly to Theorem \ref{thm:67}.
\begin{lem}
\label{lem:K}
Fix a rational $\epsilon$ and let $(B_{n})_n$ be a computable sequence such that $B_n$ is a orthonormal basis for $\mathbb{C}^{2^{n}}$ composed of algebraic complex vectors. Then, for all $\sigma \in \bigcup_{n} B_{n}$, $K(\sigma)= QK^{\epsilon}(|\sigma\big>\big<\sigma|)$, upto an additive constant depending only on $\epsilon$ and $(B_{n})_n$.
\end{lem}
Note that $K(\sigma)$ is well-defined as $\sigma$ is complex algebraic.
The following Theorem \ref{thm:8} agrees nicely with the upper bound for $K$ in the classical setting: for all strings $x$, $K(x) \leq |x|+K(|x|)+1 $ (See theorem 2.2.9 in \cite{misc}.).
\begin{thm}
\label{thm:8}
There is a constant $d>0$ such that for any $\epsilon$ and any $\tau$, QK$^{\epsilon}(\tau) \leq |\tau| + $ K$(|\tau|) + d$.
\end{thm}
\begin{proof}
Let $k>1$. Let $P$ be the prefix-free Turing machine which on input $\pi$, such that $\mathbb{U}(\pi)=n$ outputs $E=(e_{i})_{i=1}^{2^{n}}$, the standard computational basis of $\mathbb{C}^{2^{n}}$. 
\end{proof}
It may seem that this upper bound, given by the apparently inefficient device of using $2^{|\tau|}$ many orthonormal vectors to approximate $\tau$, can be improved. However, the bound is tight by Theorem \ref{thm:67} together with the classical counting theorem (see \cite{misc1}, theorem 3.7.6.).
\\
As we shall see later, the unique tracial state $\tau = (\tau_{n})_{n\in \mathbb{N}}$ where for all $n$, $\tau_n$ is the $2^n$ by $2^n$ diagonal matrix with $2^{-n}$ along the diagonal is quantum Martin-L{\"o}f random. Theorem \ref{thm:31} shows that its initial segments achieve the upper bound given by Theorem \ref{thm:8}.
\begin{thm}
\label{thm:31}
Let $k$ be any natural number. There is a constant $t$ such that for all $n$, $QK^{2^{-k}}(\tau_{n})\geq n+K(n)-t$.
\end{thm}
\begin{proof} 
Fix a $k$ and suppose towards a contradiction that for all $t \in \mathbb{N}$, there is a $n_{t} $ such that $QK^{2^{-k}}(\tau_{n_{t}})< n_{t}+K(n_{t})-t$. So, for all $t$, there are $F_{t} \subseteq \mathbb{C}^{2^{n_{t}}}$ and $\sigma_{t}$ such that $\mathbb{U}(\sigma_{t})=F_{t}$ and

\[2^{-k}<\sum_{v\in F_{t}} \big<v|\tau_{n_{t}}|v\big>= 2^{-n_{t}}|F_{t}|,\]
and
\[|\sigma_{t}|+\text{log}(|F_{t}|) < n_{t}+K(n_{t})-t.\]
\\
Taking log on both sides of the first inequality and inserting in the second gives that for all $t$, $n_t$ and $\sigma_{t}$,
\begin{align}
    \label{eq:43}
    && t-k+|\sigma_{t}| <  K(n_{t}). 
\end{align}
Now, define a prefix-free machine $M$ as follows. On input $\pi$, $M$ checks if $\mathbb{U}(\pi)$ halts and outputs a orthonormal set $F\subseteq \mathbb{C}^{2^{n}}$ for some $n$. If so, then $M(\pi)=n$. Let r be the coding constant of $M$. Note that for all $t$, $M(\sigma_{t})=n_t$. So, $K(n_t) \leq |\sigma_{t}|+r.$
Together with \eqref{eq:43}, we have that for all $t$,
$t-k+|\sigma_{t}|< |\sigma_{t}|+r$. So, $t-k<r$ for all $t$, a contradiction.

\end{proof}

In contrast to Lemma \ref{lem:K}, we have,
\begin{lem}
Fix an $\epsilon>0$ and an $n\in \mathbb{N}$. It is not true that for all $\sigma$, complex algebraic pure states in $\mathbb{C}^{2^{n}}$, $QK^{\epsilon}(|\sigma\big>\big<\sigma|)=^{+} K(\sigma)$.
\end{lem}
\begin{proof}
Clearly, for all $\sigma$, complex algebraic pure states, $QK^{\epsilon}(|\sigma\big>\big<\sigma|)\leq^{+} K(\sigma)$ holds.
Suppose that for some $\epsilon$ and $n\in \mathbb{N}$, for all $\sigma \in \mathbb{C}^{2^{n}}_{alg}$, pure states, $QK^{\epsilon}(|\sigma\big>\big<\sigma|)\geq^{+} K(\sigma)$ holds.
By Theorem \ref{thm:8}, for all $\sigma \in \mathbb{C}^{2^{n}}_{alg}$, pure,  
$K(n)+n\geq^{+}QK^{\epsilon}(|\sigma\big>\big<\sigma|)\geq^{+} K(\sigma)$. This is a contradiction as there are only finitely many programs of length atmost $n+K(n)$ but there are infinitely many complex algebraic pure states, $\sigma$ of length $n$.
\end{proof}Analogously to $QC$ \cite{Berthiaume:2001:QKC:2942985.2943376,brudno} a `counting condition' also holds for $QK$: the cardinality of a orthonormal set of vectors with bounded complexity has an upper bound depending on the complexity bound. The counting condition for $QK$ is established in a different fashion than that for $QC$ (which uses entropy inequalities like  Holevo's-chi \cite{Berthiaume:2001:QKC:2942985.2943376} and Fanne's inequality \cite{brudno}). This reflects once again that $QC$ invloves approximating a density matrix by another density matrix while $QK$ involves `projecting' a density matrix onto a  subspace.
\begin{thm}
\label{thm:count}
Let $V=(v_{i})_{i=1}^{N} \subset \mathbb{C}^{2^{s}}$ be a collection of orthonormal vectors with $QK^{\epsilon}(|v_{i}\big>\big<v_{i}|)\leq B$ for all $i$. Then, $N\leq \epsilon^{-1}2^{B}$. 
\end{thm}
\begin{proof}
For each $v_i$, we have $\sigma_i$ and $F_i$,\[F_{i}= \sum_{t\in A_{i}} |t\big>\big<t|,\]
with $A_{i}\subset \mathbb{C}^{2^{s}}$ orthonormal, such that $\big<v_{i}|F_{i}|v_{i}\big>>\epsilon$, $\mathbb{U}(\sigma_{i})=F_i$ and $|\sigma_{i}|+$log$|A_{i}|\leq B$. Let $D\subseteq \{1,2\cdots N\}$ be maximal such that $F_{i}\neq F_j$ for $i,j$ in $D$. ($D \neq \{1,2\cdots N\}$ may hold as there may be $i,j$ with $F_{i}=F_j$). Let $F$ be the orthogonal projector onto the subspace spanned by $A:=\bigcup_{i \in D}A_{i}$. Then, $A$ has dimension atmost $\sum_{i \in D}|A_{i}|$. By $|A_{i}| \leq 2^{-|\sigma_{i}|}2^{B}$ for all $i$ and noting that $\sigma_{i}\neq \sigma_j$ for $i,j$ in $D$,
\[\text{Tr}(F) \leq \sum_{i \in D}|A_{i}| \leq \sum_{i \in D} 2^{-|\sigma_{i}|}2^{B}\leq 2^{B}\sum_{\sigma \in \text{dom}(\mathbb{U})}2^{-|\sigma|} \leq  2^B.\]

The reason behind summing over $i \in D$, rather than over $i \leq N$ was to get the second to last inequality. By the maximality of $D$, $A =\bigcup_{i \leq N}A_{i}$ and so, $A_i$ is a subspace of $A$ for all $i\leq N$. So, $\big<v_{i}|F|v_{i}\big>\geq \big<v_{i}|F_{i}|v_{i}\big> >\epsilon $ for all $i \leq N$. By orthonormality of $V$,

\[\epsilon N < \sum_{i} \big<v_{i}|F|v_{i}\big> \leq \text{Tr}(F).\]
\end{proof}

\section{Quantum algorithmic randomness}
\label{sec:4}
We briefly review quantum algorithmic randomness. All definitions in this section are from \cite{unpublished} or \cite{bhojraj2020quantum} unless indicated otherwise.

While it is clear what one means by an infinite sequence of bits, it is not immediately obvious how one would formalize the notion of an infinite sequence of qubits. To describe this, many authors have independently come up with the notion of a \emph{state} \cite{unpublished,article,brudno}. We will need the one given by Nies and Scholz \cite{unpublished}. 
\begin{defn}
A \emph{state}, $\rho=(\rho_n)_{n\in \mathbb{N}}$ is an 
infinite sequence of density matrices such that $\rho_{n} \in \mathbb{C}^{2^{n} \times 2^{n}}$ and $\forall n$,  $PT_{\mathbb{C}^{2}}(\rho_n)=\rho_{n-1}$.
\end{defn}
The idea is that $\rho$ represents an infinite sequence of qubits whose first $n$ qubits are given by $\rho_n$. Here, $PT_{\mathbb{C}^{2}}$ denotes the partial trace which `traces out' the last qubit from $\mathbb{C}^{2^n}$. The definition requires $\rho$ to be coherent in the sense that for all $n$, $\rho_n$, when `restricted' via the partial trace to its first $n-1$ qubits, has the same measurement statistics as the state on $n-1$ qubits given by $\rho_{n-1}$.
Note that a state $\rho$ is not considered to belong to an (infinite dimensional) infinite tensor product of finite dimensional Hilbert spaces \cite{vonNeumann1939}. Rather, it is an infinite sequence of density matrices, each belonging to a finite dimensional Hilbert space. We briefly comment on the partial trace operator (see \cite{Nielsen:2011:QCQ:1972505} for details) which appears in the definition of a state.
\begin{remark}
While a pure state is a single quantum system, a strictly (non-pure) mixed state is a convex combination of several pure states. By not being a \emph{single} quantum system but rather a probability distribution on a set of pure states, a strictly mixed state contains classical information (the distribution) in addition to purely quantum objects (the pure states on which the distribution is defined).

The partial trace of a pure state in a product (composite) Hilbert space describes the proper subsystem given by restricting that state to a component of the tensor product. As an entangled pure state is not a product tensor of pure states, the partial trace of an entangled pure state is a strictly mixed state. This reflects the fact that a proper subsystem of an entangled system is not a single quantum state but rather necessarily needs to be described using a probability distribution over many quantum states.

So, as against unitary transformations which map pure states to pure states, the partial trace operator can map a pure state to a classical probabilistic mixture of multiple pure states.
\end{remark}
The following state will be the quantum analogue of Lebesgue measure.
\begin{defn}
\label{def:tr}
Let $\tau=(\tau_n)_{n\in \mathbb{N}}$ be the state given by setting $\tau_n = \otimes_{i=1}^n I$ where $I$ is the two by two identity matrix.
\end{defn}
\begin{defn}
A \emph{special projection} is a hermitian projection matrix with complex algebraic entries.
\end{defn}
Since the complex algebraic numbers (roots of polynomials with rational coefficients) have a computable presentation, we may identify a special projection with a natural number and hence talk about computable sequences of special projections.
\begin{defn}
\label{defn:sigclass}
A quantum $\Sigma_{1}^0$
set (or q-$\Sigma_{1}^0$
set for short) G is a computable
sequence of special projections $G=(p_{i})_{i\in \mathbb{N}}$ such that $p_i$ is $2^i$ by $2^i$ and range$(p_i \otimes I) \subseteq$ range $(p_{i+1})$ for all $i\in \mathbb{N}$. 

\end{defn}
\begin{defn}
If $\rho$ is a state and $G=(p_{n})_{n\in \mathbb{N}}$ a q-$\Sigma_{1}^0$
set as above, then $\rho(G):=\lim_{n}$ Tr$(\rho_{n}p_n)$.
\end{defn}

\begin{defn}
A \emph{quantum Martin-L{\"o}f test} (q-MLT) is a computable sequence, $(S_{m})_{m \in \mathbb{N}}$ of q-$\Sigma_{1}^0$ classes such that $\tau(S_m)$ is less than or equal to $2^{-m}$ for all m, where $\tau$ is as in Definition \ref{def:tr}.
\end{defn}

\begin{defn}
$\rho$ is \emph{q-MLR} if for any q-MLT $(S_{m})_{m \in \mathbb{N}}$, $\inf_{m \in \mathbb{N}}\rho(S_m)=0$.
\end{defn}
Roughly speaking, a state is q-MLR if it cannot be `detected by projective measurements of arbitrarily small rank'.
\begin{defn}
$\rho$ is said to \emph{fail  the q-MLT $(S_{m})_{m \in \mathbb{N}}$, at order $\delta$}, if $\inf_{m \in \mathbb{N}}\rho(S_m)>\delta$. $\rho$ is said to \emph{pass  the q-MLT $(S_{m})_{m \in \mathbb{N}}$ at order $\delta$} if it does not fail it at $\delta$. 
\end{defn}
So, $\rho$ is q-MLR if it passes all q-MLTs at all $\delta>0$. Quantum Martin-L{\"o}f randomness is modelled on the classical notion: 
An infinite bitstring $X$ is said to \emph{pass} the Martin-L{\"o}f test $(U_{n})_n$ if $X \notin \bigcap_{n}U_n$ and is said to be \emph{Martin-L{\"o}f random (MLR)} if it passes all Martin-L{\"o}f tests (See 3.2.1 in \cite{misc}). 

A related notion is Solovay randomness. A computable sequence of $\Sigma^{0}_1$ classes, $(S_{n})_n$ is a \emph{Solovay test} if $\sum_{n} \mu(S_{n})$, the sum of the Lebesgue measures of the $S_n$s is finite. An infinite bitstring $X$ \emph{passes} $(S_{n})_n$ if $X\in S_n$ for infinitely many $n$ (See 3.2.18 in \cite{misc}).

We obtain a notion of a quantum Solovay test by replacing `$\Sigma^{0}_1$ class' and `Lebesgue measure' in the definition of classical Solovay tests with `quantum-$\Sigma^{0}_1$ set' and $\tau$ respectively. The following definitions are from \cite{bhojraj2020quantum} unless indicated otherwise:
\begin{defn}\label{defn:1} 
A uniformly computable sequence of quantum-$\Sigma^{0}_1$ sets, $(S^{k})_{k\in\omega}$ is a \emph{ quantum-Solovay test} if  $\sum_{k\in \omega} \tau(S^{k}) <\infty.$\end{defn}\begin{defn}\label{defn:2}
For $0<\delta<1$, a state $\rho$ \emph{fails the Solovay test $(S^k)_{k\in\omega}$ at level $\delta$} if there are infinitely many $k$ such that $\rho(S^k)>\delta$.
\end{defn}
\begin{defn}
A state $\rho$ \emph{passes the Solovay test $(S^k)_{k\in\omega}$} if for all $\delta>0$, $\rho$ does not fail $(S^k)_{k\in\omega}$ at level $\delta$. I.e., lim$_{k}\rho(S^{k})=0$.
\end{defn}
\begin{defn}
A state $\rho$ is \emph{quantum Solovay random} if it passes all quantum Solovay tests.\end{defn}
It is remarkable that $X$ is MLR if and only if it passes all Solovay tests (See 3.2.19 in \cite{misc}). This is also true in the quantum realm\cite{bhojraj2020quantum}.

An \emph{interval Solovay test} is a Solovay test, $(S_{n})_n$ such that each $S_n$ is generated by a finite collection of strings (See 3.2.22 in \cite{misc}). Its quantum version is:

\begin{defn}\cite{unpublished}
A \emph{strong Solovay test} is a computable sequence of special projections $(S^{m})_{m}$ such that $\sum_{m}\tau(S^{m}) < \infty$. A state $\rho$ fails $(S^{m})_{m}$ at $\epsilon$ if for infinitely many $m$, $\rho(S^{m})>\epsilon$.
\end{defn}
\begin{defn}\cite{unpublished}
A state $\rho$ is \emph{weak Solovay random} if it passes all strong quantum Solovay tests.\end{defn}
It is open whether weak Solovay randomness is equivalent to q-MLR. We need the notion of a computable real number to talk about Schnorr randomness: For the purposes of this paper, a function, $f$ from the natural numbers to the rationals is said to be \emph{computable} if there is a Turing machine, $\phi$ such that on input $n$, $\phi$ halts and outputs $f(n)$ (See Theorem 5.1.2 in \cite{misc1}). Note here that we interpret the output of a Turing machine as a rational number.
\begin{defn}
A sequence $(a_n)_{n\in \mathbb{N}}$ is said to be \emph{computable} if there is a computable function $f$, such that $f(n)=a_n$.

\end{defn}

\begin{defn}
A real number $r$ is said to be \emph{computable} if there is a computable function $f$ such that for all $n$, $|f(n)-r|<2^{-n}$.
\end{defn}

By 7.2.21 and 7.2.22 in \cite{misc1}, a Schnorr test may be defined as:
\begin{defn}
A \emph{Schnorr test} is an interval Solovay test, $(S^{m})_{m}$ such that $\sum_{m}\mu(S^{m}) $ is a computable real number. 
\end{defn}
An infinite bitstring passes a Schnorr test if it does not fail it (using the same notion of failing as in the Solovay test). We mimic this notion in the quantum setting.
\begin{defn}
\label{defn:Schnorr}
A \emph{quantum Schnorr test} is a strong Solovay test, $(S^{m})_{m}$ such that $\sum_{m}\tau(S^{m}) $ is a computable real number. A state is quantum Schnorr random if it passes all Schnorr tests.
\end{defn}

\section{Relating QK to randomness }
\label{sec:3} 
\subsection{A Chaitin type result}

Theorem \ref{thm:7} is a Chaitin type characterization of the weak Solovay random states in terms of $QK$. (Chaitin's result in the classical setting says that an infinite bitstring $X$ is Solovay random if and only lim$_{n}K(X\upharpoonright n) - n = \infty$). 
\begin{thm}
\label{thm:7}
A state $\rho = (\rho_{n})_{n}$ is weak Solovay random if and only if \[\forall \epsilon>0 \forall c >0 \forall^{\infty} n  QK^{\epsilon}(\rho_{n}) \geq n+c.\]
\end{thm}

\begin{proof}
($\Longleftarrow)$: Suppose for a contradiction that $\rho$ fails a strong Solovay test $(S_m)_{m}$ at $\epsilon>0$. The idea will be to use the subspaces given by the $ S_m $s, to approximate $\rho$. More, precisely, the $F$ appearing in the definition of $QK^{\epsilon}$ will be the orthonormal vectors given by the projection $S_m$ for an appropriate $m$. The details are as follows.  Let $M$ be the prefix-free machine doing the following. On input $\sigma$, if $\mathbb{U}(\sigma)=m$, then output $(v_i)_{i}$ where \[S_{m}= \sum_{i} |v_{i}\big>\big<v_{i}|.\] Let $c_{M}$ be it's coding constant. Take an $m$ such that Tr$(\rho_{n_{m}}S_{m})>\epsilon$ (Notation: $n_{m}$ is the natural number $n$ such that $S_m$ is a projection on $n$ qubits.). By the choice of $m$, \[\text{QK}^{\epsilon}(\rho_{n_{m}}) \leq \text{K}(m)+ c_{M}+ \text{log}(2^{n_{m}}\tau(S_{m}))= n_{m}+ \text{K}(m) - f(m) + c_{M},\]
where $f$ is the function: $f(m)= -\text{log}(\tau(S_m))$. As $f$ is computable and as $\sum_{m} 2^{-f(m)} < \infty $ by the definition of a strong Solovay test, Lemma 3.12.2 in \cite{misc1} implies that for all $m$, $\text{K}(m) - f(m) \leq q$ for some constant $q$. Noting that we may assume the sequence $n_{m}$ to be strictly increasing in $m$ and  letting $c:= q+c_{M}+1$, we see that $\exists^{\infty} n$  such that $ \text{QK}^{\epsilon}(\rho_{n})< n+c.$ 
\\
($\Longrightarrow$):
Suppose toward a contradiction that there is a $\epsilon>0$ and a constant $c>0$ such that there are infinitely many $n$ with $ \text{QK}^{\epsilon}(\rho_{n})< n+c.$  Define a strong Solovay test $S$ as follows. Let $T$ be the set of all $\sigma$ such that $\mathbb{U}(\sigma)$ halts and outputs an orthonormal set $F_{\sigma}\subseteq  \mathbb{C}^{2^{n_{\sigma}}}$  such that $|\sigma|+$ log $|F_{\sigma}|< n_{\sigma} + c$. For all $\sigma \in T$, let \[P_{\sigma}:= \sum_{v\in F_{\sigma}} |v \big>\big<v |\] and let $S:= (P_{\sigma})_{\sigma \in T}$. For all $\sigma \in T$, $2^{|\sigma|}|F_{\sigma}|< 2^{n_{\sigma} + c}$. So, $ \tau(P_{\sigma})= 2^{-n_{\sigma}}|F_{\sigma}|< 2^{c-|\sigma| }$.  So,
\[\sum_{\sigma \in T}\tau(P_{\sigma}) < 2^{c}\sum_{\sigma \in T}2^{-|\sigma|} < 2^{c}\sum_{\sigma: \mathbb{U}(\sigma)\downarrow} 2^{-|\sigma|} < \infty, \]
since $\mathbb{U}$ is prefix-free. This shows that $S$ is a strong Solovay test. For any $n$ such that $ \text{QK}^{\epsilon}(\rho_{n})< n+c$, there is a $\sigma \in T$ such that Tr$(P_{\sigma} \rho_{n})>\epsilon$. So, $\rho$ fails $S$ at $\epsilon$.
\end{proof}
The following corollary shows the equivalence of weak Solovay and q-ML randomness for a specific type of states.
Let $B= ((b^{n}_{0},b^{n}_{1}))_{n\in \mathbb{N}}$ be a system ( Definition \ref{defn:78}). Let $A^{\infty}_B$ be the set of all states which are limits of elements from $A_B$ as in
\ref{rem:K}. For example, $b^{1}_{0}\otimes b^{2}_{1}\otimes b^{3}_{0}\otimes b^{4}_{1}\cdots =: \rho \in A^{\infty}_B$.

\begin{cor}
\label{cor}
For any $B$, weak Solovay randomness is equivalent to q-MLR on $A^{\infty}_B$.
\end{cor}
\begin{proof}
Fix a system $B= ((b^{n}_{0},b^{n}_{1}))_{n\in \mathbb{N}}$ and let $\rho\in A^{\infty}_B$ be weak Solovay random. Let $\rho'$ be the bitstring induced by $\rho$. I.e., for example if $\rho=b^{1}_{0}\otimes b^{2}_{1}\otimes b^{3}_{0}\otimes b^{4}_{1} \cdots$, then $\rho':= 0101\cdots$.  By Theorem \ref{thm:7}, for $\epsilon=0.5$ \[  \forall c >0 \forall^{\infty} n  QK^{0.5}(\rho_{n}) \geq n+c.\] By Remark \ref{rem:K}, $K(\rho'\upharpoonright n)= MK^{0.5}(\rho_{n})$ upto a constant depending only on $B$. So, 
\[  \forall c >0 \forall^{\infty} n  K(\rho'\upharpoonright n) \geq n+c.\]
By Chaitin's result, \cite{Chaitin:1987:ITR:24912.24913} $\rho'$ is MLR. Now, by an easy modification of 3.13 from \cite{unpublished}, $\rho$ is q-MLR. We already know that q-MLR implies weak Solovay randomness for any state from before.
\end{proof}

\subsection{Chaitin and Levin\textendash Schnorr type results}
It turns out that weak Solovay randomness is equivalent to q-MLR and has both Chaitin ((3) in Theorem \ref{thm:39} ) and Levin\textendash Schnorr ((4) in Theorem \ref{thm:39}) type characterizations in terms of $QK$ when the states are restricted to a certain class, $\mathcal{L}$ defined below. To define this class we need to consider the halting set over the halting set : $\emptyset''=(\emptyset')'$ (See\cite{misc}). Let $\mathcal{L}$ denote the union of the two classes of states.
\begin{enumerate}
    \item States in $A^{\infty}_B$ for some $B$, as in Corollary \ref{cor}
    \item States which do not Turing compute $\emptyset''$.
    \end{enumerate}
 Nies and Barmpalias (in personal communication) have shown that q-MLR is equivalent to weak quantum Solovay randomness for states which do not compute $\emptyset''$. The same equivalence also holds on $A^{\infty}_B$ by Corollary \ref{cor}. This similarity motivates our study of $\mathcal{L}$. 
\begin{thm}
\label{thm:39}
If $\rho=(\rho_{n})_{n} \in  \mathcal{L}$, then the following are equivalent
\begin{enumerate}
\item $\rho$ is q-MLR.
    \item $\rho$ is weak Solovay random.
    \item $ \forall \epsilon>0 \forall c >0 \forall^{\infty} n,  QK^{\epsilon}(\rho_{n}) \geq n+c.$
    \item $\forall \epsilon>0\exists c \forall n, QK^{\epsilon}(\rho_{n})>n-c .$
\end{enumerate}
\end{thm}
\begin{proof}
(1)$\iff$(2) follows from the previous remarks.\\
(4)$\Longrightarrow$(1):\begin{proof} First, let $\rho \in A^{\infty}_B$ for some $B$ and let (4) hold. By the same argument as in Corollary \ref{cor}, we get that
\[  \exists c >0 \forall  n  ,K(\rho'\upharpoonright n) \geq n-c .\]
The classical Levin\textendash Schnorr result \cite{misc1} implies that $\rho'$ is MLR. Using once more 3.13 in \cite{unpublished} as in \ref{cor}, we see that $\rho$ is q-MLR. Now suppose $\rho$ does not Turing compute $\emptyset''$. We will show that (4) implies (2). Suppose for a contradiction that $(S^{m})_{m}$ is a strong Solovay test which $\rho$ fails at $\epsilon'>0$. By Theorem 2.11 in \cite{bhojraj2020quantum}, we can effectively compute a q-MLT $(G^{m})_{m}$ which $\rho$ fails at some rational $\epsilon>0$. Let $g(m):=$ the least $s$ such that Tr$(\rho_sG^{2m}_{s})>\epsilon$. As $\rho$ computes $g$, by Martin's high domination theorem (see \cite{misc1} for a proof), there is a total computable function $f$ such that $\exists^{\infty} g(n)<f(n)$. We may assume that $f(t)>3t$ for all $t$ by taking the max of 2 computable functions. Fix this $f$ (non-uniformly) and consider the following machine, $M$:\\ On input $0^{m}1$, $M$ outputs $F^{m}$ where $F^{m}$ is such that
\[G^{2m}_{f(m)}= \sum_{v \in F^{m}}|v\big>\big<v|.\]
Clearly $M$ is prefix free. Let $l-1$ be it's coding constant. Let $t$ be so that $f(t)>g(t)$. Let $F^t$ be defined similarly to $F^m$ above. Then, by definition of $g$, we have that \[\epsilon< \sum_{v \in F^{t}}\big<v|\rho_{f(t)}|v\big>.\] $M(0^{t}1)=F^{t}$ and so, there is a $\pi$ such that $|\pi|\leq t+l $ and $\mathbb{U}(\pi)=F^{t}$. Also note that $|F^{t}|\leq 2^{f(t)-2t}$ by the definition of a q-MLT. So,
$
\text{QK}^{\epsilon}(\rho_{f(t)}) \leq t+l +f(t)-2t=f(t)-t+l.$

Recall that $t$ was an arbitrary element of the infinite set $\{s: f(s)>g(s)\}$. So, for infinitely many $t$s, there is an $n=f(t)$ such that $\text{QK}^{\epsilon}(\rho_{n}) \leq n-t+l,$  contradicting $(4)$.
\end{proof}
$(3)\Longrightarrow(4)$ is obvious and $(2)\Longrightarrow(3)$ was done in Theorem \ref{thm:7}.
\end{proof}
We apply the preceding theorem to get the following quantum analog of a classical result, Proposition 3.2.14 in \cite{misc}.
\begin{thm}
\label{thm:compset}
Let $C$ be an infinite computable set, $\rho \in \mathcal{L}$ and $\epsilon>0$. If there is a $d$ such that for all $m\in C$, $QK^{\epsilon}(\rho_m)>m-d$, then $\rho$ is weak Solovay random.
\end{thm}
\begin{proof}
Let $M$ be the machine doing the following: On input $\sigma$, check if $\mathbb{U}(\sigma)=F$, an orthonormal set $F\subseteq \mathbb{C}^{2^{n}}$. If such a $F$ and $n$ exist, compute $s$ such that $n+s$ is the least element of $C$ greater than $n$ and output the set:
\[T:= \{v \otimes \pi: v \in F, \pi \in 2^{s}\}.\]
Note that $|T|=2^{s}|F|$. It is easy to see that $M$ is prefix-free. Let $l$ be it's coding constant. Suppose for a contradiction that $\rho$ is not weak Solovay random. \ref{thm:39} implies that $\forall c$, $\exists n_{c}$ such that $QK^{\epsilon}(\rho_{n_{c}}) \leq n_{c}-c.$ Let $c$ be arbitrary and take such an $n:=n_c$. There is a $\sigma$ and $F$ such that $\mathbb{U}(\sigma)=F \in \mathbb{C}^{2^{n}}$, $|\sigma|+$log$(|F|)\leq n-c$ and \[\sum_{v\in F} \big<v|\rho_{n}|v\big> >\epsilon.\] Let $t=n+s$ be the least element of $C$ greater than $n$. On input $\sigma$, $M$ outputs $T$ as above.
Note that \[Q:=\sum_{w\in T}|w\big>\big<w|= \sum_{v\in F, \pi\in 2^{s}}|v\big>\big<v|\otimes |\pi\big>\big<\pi|=\big(\sum_{v\in F }|v\big>\big<v|\big)\otimes \big( \sum_{\pi\in 2^{s}}|\pi\big>\big<\pi|\big)=W\otimes I,
\]
where $W:=\sum_{v\in F}|v\big>\big<v|$ and $I$ be the identity on $\mathbb{C}^{2^{s}}$.
Then, by the coherence property of states, \[\sum_{w\in T} \big<w|\rho_{t}|w\big>=\text{Tr}(\rho_{t}Q)=\text{Tr}(\rho_{t}[W\otimes I])=\text{Tr}(\rho_{n}W)>\epsilon.\]
Consequently, \[QK^{\epsilon}(\rho_{t})\leq |\sigma|+  \text{log}(|T|)+ l = |\sigma|+  \text{log}(|F|)+s+ l\leq n-c+s+l=t-c+l.\]
Since $d$ and $l$ were constants and $c$ was arbitrary, this contradicts the assumption.
\end{proof}

\subsection{A weak Levin\textendash Schnorr type result}
\label{subsect:weak}
Theorem \ref{thm:7} implies that if $\rho$ is weak-Solovay random then, $ \forall \epsilon>0\exists c \forall n, QK^{\epsilon}(\rho_{n})>n-c$. I.e., being strong-Solovay random implies the Levin\textendash Schnorr  condition. Does this reverse? We give two partial results in this direction: the Levin\textendash Schnorr condition implies that $\rho$ passes all strong-Solovay tests of a certain type.
\begin{defn}
For a rational $s\in(0,1)$, a $s$-strong Solovay test is a strong Solovay test $(S^{r})_{r}$ such that
$\sum_{r}\tau(S^{r})^{s} < \infty $ and $\sum_{r}\tau(S^{r}) $ is a computable real number.
\end{defn}

\begin{thm}
\label{thm:mkssr}
If $ \forall \epsilon>0\exists c \forall n, QK^{\epsilon}(\rho_{n})>n-c$, then $\rho$ passes all $s$-strong Solovay tests for all rational $s\in (0,1)$.
\end{thm}

\begin{proof}
Suppose for a contradiction that $(S^{m})_{m}$ is a $s$-strong Solovay test which $\rho$ fails at $\epsilon>0$ and $\sum_{i} \tau(S^{i})=Q $, computable. For all $m$, let $S^m$ be $2^{n_{m}}$ by $2^{n_{m}}$ and we may let the $n_m$s be distinct. Let $f(m):=-$ log$(\tau(S^m))$ and $g(m):=\lceil s f(m) \rceil$.
Partition $\omega$ into the fibers induced by $g$. ($P$ is a fiber of $g$ if $P=g^{-1}(\{x\})=\{y: g(y)=x\},$ for some $x$.). Note that $g(r)\geq  -s \text{log}\tau(S^{r})$,
and hence
$2^{-g(r)}\leq \tau(S^{r})^{s} .$ In particular, this implies that $P$ is finite. So, there are  countably infinitely many fibers,  $\{P_{1},P_{2}\dots\}$ and $\omega = \bigcup_{m} P_{m},$
where for all $m$, there is an $x_m$ such that $P_{m}=g^{-1}(\{x_{m}\})$ and $m \mapsto x_{m}$ is injective.

The fiber $P$ of $x=g(z)$ can be computed from $x$ as follows. Note that $g(c)=x$ iff, $f(c) \in [s^{-1}(x-1), s^{-1}x]$. As $Q$ is computable, compute an interval $J$ such that, $|J|<2^{-s^{-1}x}$,  $Q\in J$ and $\sum_{r\leq q} 2^{-f(r)}\in J $ for some $q$. So, $f(c)>s^{-1}x$ if $c>q$. $P$ can be computed by evaluating $g$ on $[0,q]$.

The idea is to describe $S^r$ by computing the fiber $P$ containing $r$ and then specifying the location of $r$ in the lexicographical ordering on $P$. As $(S^r)_r$ is a $s$-strong Solovay test, this description of $S^r$ is short enough to derive a contradiction. Consider the  machine, $M$ doing the following:  On input $\lambda$, check if there is a decomposition $\lambda=\pi\sigma$  such that,
\begin{itemize}
    \item There is an $m$ such that $\mathbb{U}(\pi)=0^{g(m)}$.
    \item $|\sigma|=t=\lceil\text{log}(|P|)\rceil$ where  $P$ is the fiber of $g(m)$. (Recall that $P$ can be computed from  $g(m)$.) 
    
\end{itemize}

 If these hold, then order $P$ lexicographically using the ordering on $2^t$ and let $r$ be the $\sigma^{th}$ element in this ordering. Output $F^{r}$ where $F^{r}$ is such that $S^{r}= \sum_{v \in F^{r}}|v\big>\big<v|.$
 
Note that $M$ is prefix free: Suppose $\lambda,\lambda' \in $ dom$(M)$ as witnessed by $\lambda=\pi\sigma$ and $\lambda=\pi'\sigma'$. So, $M$ finds $m$ and $m'$ such that $\mathbb{U}(\pi)=0^{g(m)}$ and $\mathbb{U}(\pi')=0^{g(m')}$. Let $\pi\sigma \preceq \pi'\sigma'$. Then, it must be that $\pi\preceq\pi'$ or $\pi'\preceq\pi$. Since $\mathbb{U}$ is prefix free, it follows that $\pi=\pi'$. So, $g(m)=g(m')$. Hence, $m$ and $m'$ are in the same fiber, $P$. Letting $t= \lceil \text{log}(|P|)\rceil$, $\lambda,\lambda' \in $ dom$(M)$ implies that $|\sigma|=|\sigma'|=t$.

For each $m\in \omega$, let $r_m$ be any element from $P_m$. Then,
\begin{align}
\label{eq:mkssr3}
    \sum_{m}|P_{m}|2^{-g(r_{m})} = \sum_{m}\sum_{r\in P_{m}} 2^{-g(r)} \leq \sum_{m}\sum_{r\in P_{m}} \tau(S^{r})^{s}= \sum_{i} \tau(S^{i})^{s}< \infty.
\end{align}

Let $h$ be the function defined by, $h(m):=$ log$(|P_{m}|)-g(r_{m})$ where $r_m$ is any representative from $P_m$. By \eqref{eq:mkssr3},
$|P_{m}|2^{-g(r_{m})} \rightarrow 0$ as $m\rightarrow \infty$. So, $h(m) \rightarrow -\infty$ as $m\rightarrow \infty$.
Each fiber is finite and $\rho$ fails the test at $\epsilon$. So, there is an infinite set $I= \{P_{j_{1}}, P_{j_{2}},\dots\}$ such that for all $i$, there is a $t_{i}\in P_{j_{i}}$ with  
Tr$(\rho_{n_{t_{i}}}S^{t_{i}})>\epsilon$.  $j_{i} \rightarrow  \infty$ as $i\rightarrow \infty$ and so, $h(j_{i})= $ log$(|P_{j_{i}}|)-g(r_{t_{i}}) \rightarrow -\infty$ as $i\rightarrow \infty$. This asymptotic behavior will be used below to derive a contradiction.

Fix an arbitrary $i$ and a $t=t_{i}\in P_{j_{i}}$ as above.
So, $
\epsilon< \sum_{v \in F^{t}}\big<v|\rho_{n_{t}}|v\big>.$

Let $t$ be the $\sigma^{th}$ element of $P_{j_{i}}$ in the lexicographic ordering used by $M$ and let $\mathbb{U}(\pi)=0^{g(t)}$ and $K(0^{g(t)})=|\pi|$. Then,
$M(\pi\sigma)=F^{t}$ and so, there is a bitstring $\iota$ such that $|\iota|\leq^{+} K(0^{g(t)})+ \lceil \text{log}(|P_{j_{i}}|)\rceil $ and $\mathbb{U}(\iota)=F^{t}$. Note that log$|F^{t}|=$ log(Tr($S^t))=$ log$(\tau(S^{t})) + n_{t}= -f(t) + n_{t}$. Let $d:=(1-s)s^{-1} >0$. So, $\{(\lceil nd \rceil, 0^{n}): n \in \omega\},$
is a bounded request set (see \cite{misc1} for a definition) and hence $K(0^{n})\leq^{+}\lceil nd \rceil$. Using all this, we get that:
\begin{align*}
\text{QK}^{\epsilon}(\rho_{n_{t}})
&\leq^{+} K(0^{g(t)})+\lceil \text{log}(|P_{j_{i}}|)\rceil  -f(t) +  n_{t}\\
&\leq^{+} \lceil d g(t)\rceil+\lceil \text{log}(|P_{j_{i}}|)\rceil -f(t) +  n_{t}\\
&\leq^{+} dsf(t)+\lceil \text{log}(|P_{j_{i}}|)\rceil -f(t) +  n_{t}\\
&=f(t)(ds-1)+\lceil \text{log}(|P_{j_{i}}|)\rceil  +  n_{t}\\
&=-sf(t) +\lceil \text{log}(|P_{j_{i}}|)\rceil  +  n_{t}\\
&\leq^{+}-g(t)+\lceil \text{log}(|P_{j_{i}}|)\rceil  +  n_{t}\\
&=h(j_{i})+n_{t}.
\end{align*}
The last equality follows as $t=t_i$ is in  $P_{j_{i}}$. This means that there is an infinite sequence $(n_{t_{i}})_{i}$ such that
\[\text{QK}^{\epsilon}(\rho_{n_{t_{i}}}) <^{+} n_{t_{i}}+h(j_{i}) .\]
Finally, recall that $h(j_{i})\rightarrow -\infty$ as $i\rightarrow \infty$ and we have a contradiction.
\end{proof}
Theorem \ref{thm:mkssr} can be strengthened by weakening the defining criteria for a $s$-strong Solovay test.
\begin{defn}
\label{def:phi}
Let $s\in(0,1)$ be a rational. Let $\phi$ be any computable, non-decreasing, non-negative function on the reals such that $\phi(sr)\geq^{\times} s\phi(r)$ (I.e., there is a $C>0$ \emph{independent of s}, such that for all $r$,  $C\phi(sr)\geq  s\phi(r)$) and $\sum_{n}2^{-n}\phi(n)<\infty$ (So, $\phi$ does not tend to infinity too fast). A $(\phi,s)$-strong Solovay test is a strong Solovay test $(S^{r})_{r}$ such that
\begin{align}
\label{eq:mkssr6}
    \sum_{r}\dfrac{\tau(S^{r})^{s}}{\phi(-\text{log}(\tau(S^{r})))} < \infty ,
\end{align}

and
\[\sum_{r}\tau(S^{r}) = Q,\] where $Q$ is a computable real number.
\end{defn}
The term in the denominator in \eqref{eq:mkssr6} tends to infinity with $r$ and hence it is easier for a strong Solovay test to be a $(\phi,s)$-strong Solovay test than to be a $s$-strong Solovay test. So, passing all $(\phi,s)$-strong Solovay tests is a more restrictive notion of randomness than passing all $s$-strong Solovay tests. So, the following theorem is an improvement of, and implies Theorem \ref{thm:mkssr}.
\begin{thm}
If $ \forall \epsilon>0\exists c \forall n, QK^{\epsilon}(\rho_{n})>n-c$, then $\rho$ passes all $(\phi,s)$-strong Solovay tests for all rational $s\in (0,1)$ and all $\phi$ as in Definition \ref{def:phi}.
\end{thm}

\begin{proof}
Suppose for a contradiction that  $(S^{m})_{m}$ is a $(\phi,s)$-strong Solovay test which $\rho$ fails at $\epsilon>0$ and $\sum_{i} \tau(S^{i})=Q $, computable. For all $m$, let $S^m$ be $2^{n_{m}}$ by $2^{n_{m}}$ and we may let the $n_m$s be distinct.
For ease of presentation, we do the proof in 2 cases. First, let $s\leq 0.5$.
Let $f(m):=-$ log$(\tau(S^m))$ and let $g(m):=\lceil s f(m) \rceil$.
Partition $\omega$ into the fibers induced by $g$.
Fix some fiber $P$ of some $x=g(r)$. I.e., $r$ is a representative from $P$. Then, $g(r)=\lceil s(- \text{log}\tau(S^{r})) \rceil$. So, $g(r)\geq  -s \text{log}\tau(S^{r}),$
and hence
$2^{-g(r)}\leq \tau(S^{r})^{s} .$
In particular, this implies that each fiber is finite. So, there are  countably infinitely many fibers,  $\{P_{1},P_{2}\dots\}$. So, $\omega = \bigcup_{m} P_{m},$
where for all $m$, there is an $x_m$ such that $P_{m}=g^{-1}(\{x_{m}\})$ and $m \mapsto x_{m}$ is injective. 
For each $m\in \omega$, let $r_m$ be any representative from $P_m$.
 For all $r$, $sf(r)\leq g(r)$ and $\phi$ is non-decreasing. So,
\begin{align}
\label{eq:mkssr13}\sum_{m} |P_{m}|\dfrac{2^{-g(r_{m})}}{\phi(g(r_{m}))} = \sum_{m}\sum_{r\in P_{m}} \dfrac{2^{-g(r)}}{\phi(g(r))}\leq  \sum_{m}\sum_{r\in P_{m}} \dfrac{\tau(S^{r})^{s}}{\phi(sf(r))} \leq^{\times} \sum_{r} \dfrac{\tau(S^{r})^{s}}{s\phi(f(r))}< \infty.
\end{align}

The fiber $P$ of $x=g(z)$ can be computed from $x$ for the same reason as in the previous proof. Its idea of `compressing' $S^r$ is also used here. 
\\
Consider the  machine, $M$ doing the following:  On input $\lambda$, search for a decomposition $\lambda=\pi\sigma$, such that
\begin{itemize}
    \item $\mathbb{U}(\pi)=0^{g(m)}$ for some $m$.
    \item
    $|\sigma| = t$ where
    $P$ is the fiber containing $m$, (which can be computed from $g(m)$) and $t= \lceil\text{log}(|P|)\rceil$.
\end{itemize}
If found, order $P$ lexicographically using the ordering on $2^t$ and let $r$ be the $\sigma^{th}$ element in this ordering. Output $F^{r}$ where $F^{r}$ is such that
$S^{r}= \sum_{v \in F^{r}}|v\big>\big<v|.$
Note that $M$ is prefix free for the same reason as in the previous proof. Let $l$ be $M$'s coding constant.
Let $h$ be the function defined by, $h(m):=$ log$(|P_{m}|)-g(r_{m})-$log $\phi g(r_{m})$, where $r_m$ is any representative from $P_m$. By \eqref{eq:mkssr13},
\[\dfrac{|P_{m}|2^{-g(r_{m})}}{\phi g(r_{m})} \rightarrow 0\] as $m\rightarrow \infty$. So, $h(m) \rightarrow -\infty$ as $m\rightarrow \infty$.
Each fiber is finite and $\rho$ fails the test at $\epsilon$. So, there is an infinite set $I= \{P_{j_{1}}, P_{j_{2}},\dots\}$ such that for all $i$, there is a $t_{i}\in P_{j_{i}}$ with  
Tr$(\rho_{n_{t_{i}}}S^{t_{i}})>\epsilon$.  $j_{i} \rightarrow  \infty$ as $i\rightarrow \infty$ and so, $h(j_{i}) \rightarrow -\infty$ as $i\rightarrow \infty$. This asymptotic behavior will be used below to derive a contradiction.

Fix an arbitrary $i$ and a $t=t_{i}\in P_{j_{i}}$ as above.
So,
\begin{align} 
\label{eqn:mkssr1}
\epsilon< \sum_{v \in F^{t}}\big<v|\rho_{n_{t}}|v\big>.
\end{align}

Let $t$ be the $\sigma^{th}$ element of $P_{j_{i}}$ in the lexicographic ordering used by $M$. Let $\pi$ be such that $K(0^{g(t)})=|\pi|$ and $\mathbb{U}(\pi)=0^{g(t)}$. Then,
$M(\pi\sigma)=F^{t}$ and so, there is a bitstring $\kappa$ such that $|\kappa|\leq K(0^{g(t)}) + \lceil \text{log}(|P_{j_{i}}|)\rceil + l$ and $\mathbb{U}(\kappa)=F^{t}$. Note that log$|F^{t}|=$ log(Tr($S^t))=$ log$(\tau(S^{t})) + n_{t}= -f(t) + n_{t}$. So,
\[
\text{QK}^{\epsilon}(\rho_{n_{t}}) \leq K(0^{g(t)})+\lceil \text{log}(|P_{j_{i}}|)\rceil +l -f(t) +  n_{t}.
\]
  
Note that $\{(n-\lceil\text{log}(\phi(n))\rceil,0^{n}): n\in \omega\}$
is a bounded request set by definition of a $(\phi,s)$ test and so
$K(0^{g(t)})\leq^{+} g(t)-$log$(\phi g(t))$. So,
\[
\text{QK}^{\epsilon}(\rho_{n_{t}}) \leq^{+} g(t)-\text{log}(\phi g(t))+ \lceil \text{log}(|P_{j_{i}}|)\rceil  -f(t) +  n_{t}.
\]
Since
$g(t)-1 < sf(t) $, we have
\begin{align}
    \label{eq:mkssr12}
    -g(t)+1>-sf(t).
\end{align}

Since $s\leq 0.5$ we have that $1-s\geq s$.\\ So,
$g(t)-f(t) \leq sf(t)-f(t)+1 = -(1-s)f(t)+1\leq -sf(t)+1.$ 
Using \eqref{eq:mkssr12}, 
\[g(t)-f(t)<-g(t)+2.\]
So,
\[
\text{QK}^{\epsilon}(\rho_{n_{t}}) <^{+} -g(t)+\lceil \text{log}(|P_{j_{i}}|)\rceil -\text{log}(\phi  g(t)) +  n_{t} = h(j_{i}) +  n_{t}.
\]
The equality follows as $t=t_i$ is in  $P_{j_{i}}$. This means that there is an infinite sequence $(n_{t_{i}})_{i}$ such that
$\text{QK}^{\epsilon}(\rho_{n_{t_{i}}}) <^{+} n_{t_{i}}+h(j_{i}).$
Finally, recall that $h(j_{i})\rightarrow -\infty$ as $i\rightarrow \infty$ and we have a contradiction.

Now let $s>0.5$ and let $f,g$ be as in the previous case. Let $b(m):= \lceil (1-s)f(m)\rceil$ and let $C:= \lceil s/(1-s)\rceil + 1$.

Consider the machine $M$ doing the following: on
input $\pi1^{y}0\sigma$, check if the following conditions hold.
\begin{itemize}
    \item 
    There is $m$ such that $\mathbb{U}(\pi)=0^{b(m)}$.
    \item If $x=b(m)$,
     $J=(s(1-s)^{-1}(x-1),s(1-s)^{-1}x + 1] \cap \omega$ and $w$ is the $y^{th}$ element of $J$, then there is a $z$ such
     that $g(z)=w$.
    
     \item
     If $P$ is the fiber of $g$ containing $z$ (P is computable from $w=g(z)$ just as in the previous case) and $t= \lceil\text{log}(|P|)
     \rceil$, then $|\sigma| = t$
     \end{itemize}
     
     If all the above are met, then order $P$ lexicographically using the ordering on $2^t$ and let $r$ be the $\sigma^{th}$ element in this ordering. Output $F^r$.

Roughly, the idea is as follows: Just as in the previous case, we want to compress $F^r$ where $r$ is the $\sigma^{th}$ number in the fiber of $w=g(m)$. The first step to achieve this is to describe $g(m)$. While in the previous case we used an $\iota$ such that $\mathbb{U}(\iota)=0^{g(m)}$ and $|\iota|=K(0^{g(m)})$), we use here the shorter string $\pi$ where  $\mathbb{U}(\pi)=0^{b(m)}$ and $|\pi|=K(0^{b(m)})$  together with $1^y$ for describing $g(m)$. From $\pi$, we get $x=b(m)$ which in turn gives $J$ which contains $g(m)$. So, $\pi$ along with $y$, the location of $g(m)$ in $J$, describes $g(m)$. After $g(m)$ is found, $F^r$ can be described just as in the previous case. The details are: As $x=b(m)$, $(1-s)f(m)\in (x-1,x]$ and hence $sf(m)$ lies in $(s(1-s)^{-1}(x-1),s(1-s)^{-1}x]$.
So, $\lceil sf(m) \rceil =  g(m)\in J=  (s(1-s)^{-1}(x-1),s(1-s)^{-1}x + 1]  \cap \omega$. Since $|J|\leq C$, $g(m)\in J$ can be determined by specifying $y \leq C$, it's location in $J$. So, $M$ can recover $g(m)$. From this point on, the remaining procedure is the same as in the previous case.
\\
We see that $M$ is prefix-free: Let $\pi1^{y}0\sigma$ and $\pi'1^{y'}0\sigma'$ be in the domain of $M$ and let $\pi1^{y}0\sigma\preceq \pi'1^{y'}0\sigma'$. By the same argument as in case1, $\pi=\pi'$ and  $M$ finds some $m,m'$ with $x=b(m)=b(m')$. It follows that $y=y'$. So, if $w$ is the $y^{th}$ (and  $y'^{th}$)  element of $J$ (as above), then $M$ finds $z,z'$ such that $g(z)=w=g(z')$. So, $z$ and $z'$ are in the same fiber $P$ and it hence follows as in the previous case that $|\sigma|=|\sigma'|$. Define $I$ and $h$ exactly as in the previous case. Fix some $i$ and let $t=t_i$ be an element of $P_{j_{i}}$ such that \eqref{eqn:mkssr1} holds. Let $t$ be the $\sigma^{th}$ element of $P_{j_{i}}$. Let $x=b(t)=\lceil (1-s)f(t)  \rceil$. So, $(1-s)f(t) \in (x-1,x]$ and  $sf(t) \in (s(1-s)^{-1}(x-1),s(1-s)^{-1}x]$. Hence, $g(t) \in (s(1-s)^{-1}(x-1),s(1-s)^{-1}x+1] \cap \omega=J$ and let $g(t)$ be the $y^{th}$ element of $J$.
Let $\pi$ be such that   $\mathbb{U}(\pi)=0^{b(t)}$ and $|\pi|=K(0^{b(t)})$. Then, on input $\pi1^{y}0\sigma$, M finds some $z$ (it could be that $z=t$, but not necessarily) such that $b(t)=b(z)=x$ and then finds that the $y^{th}$ element of $J$ is $g(z')$ for some $z'$ (again, although  $g(z')=g(t)$, it could be that $t=z'$ but not necessarily). Since $z'$ and $t$ are both in $P_{j_{i}}$, $M$ outputs $F^{t}$ after reading $\sigma$. So, there is a $\pi$
 such that $\mathbb{U}(\pi)=F^{t}$ and
 $|\pi|\leq^{+} K(0^{b(t)})+C+\lceil\text{log}(|P_{j_{i}}|)\rceil $.\\
 So, $ \text{QK}^{\epsilon}(\rho_{n_{t}})$ 
\begin{align*}
&<^{+}K(0^{b(t)})+ \lceil\text
 {log}(|P_{j_{i}}|)\rceil + \text{log}(|F^{t}|)\\
 &\leq^{+}b(t)-\lceil\text{log}(\phi(g(t)))\rceil+\lceil\text{log}(|P_{j_{i}}|)\rceil + n_{t}-f(t)\\
 &\leq^{+}-g(t)-\lceil\text{log}(\phi(g(t)))\rceil+\lceil\text{log}(|P_{j_{i}}|)\rceil + n_{t} 
\end{align*}

The last inequality is since, by \eqref{eq:mkssr12} (which holds for any $s$), $b(t)-f(t) \leq (1-s)f(t)-f(t)+1=-sf(t)+1<-g(t)+2$. Since $t=t_{i} \in P_{j_{i}}$, we see that  $ \text{QK}^{\epsilon}(\rho_{n_{t}}) <^{+}h(j_{i})-n_{t_{i}}.$ for all $i$. This gives a contradiction for the same reason as in the previous case.

\end{proof}

\subsection{QK and computable measure machines}
\label{schn}
Schnorr randomness is an important randomness notion in the classical realm\cite{misc,misc1}. While $K$ plays well with Solovay randomness, $K_C$, a version of $K$ using a computable measure machine, $C$ (a prefix-free Turing machine whose domain has computable Lebesgue measure) gives a Levin-Schnorr characterization of Schnorr randomness  (See theorem 7.1.15 in \cite{misc1}). 

So, with the intention of connecting it to quantum Schnorr randomness, we define $QK_C$ a version of $QK$ using a computable measure machine, $C$.

Theorem \ref{thm:677} shows that $QK_C$ agrees with $K_C$ on the classical bitstrings. Analogously to the classical case, Theorem \ref{thm:schnor} is a Levin\textendash Schnorr type of characterizations of quantum Schnorr randomness using $QK_C$. Theorem \ref{thm:Chaitin}, a Chaitin type characterization of quantum Schnorr randomness using $QK_C$ implies Theorem \ref{class}, a Chaitin type characterization of \emph{classical} Schnorr randomness in terms of $K_C$.

For $C$ a computable measure machine and $\sigma$ a string, $K_C$ is defined analogously to $K$; $K_C(\sigma):=$inf$\{|\tau|: C(\tau)=\sigma\}.$ The quantum version is:
for $C$, a computable measure machine and a $\epsilon>0$, define 
$QK_{C}^{\epsilon}(\tau)$ to be:
\begin{defn}
\label{defn:4}
$QK_{C}^{\epsilon}(\tau) := $ inf  $\{|\sigma|+$log$|F|: C(\sigma)\downarrow = F$, a orthonormal set in $\mathbb{C}^{2^{|\tau|}}_{alg}$ and $\sum_{v \in F} \big<v |\tau |v \big> > \epsilon \}$

\end{defn}
The infimum of the empty set is taken to be $\infty$.
Notation: In this section, $\mu$ denotes Lebesgue measure and $C_t$ denotes $C$ run upto the $t$ steps. We may assume that dom$(C_t)\subseteq 2^t$. By a \emph{sequence}, we mean a countable collection whose elements may possibly be repeated. If $S$ is a sequence, the sum $\sum_{s\in S}$ will be over all elements of $S$, with repetition.

Similarly to Theorem \ref{thm:67}, we show that $QK_C$ `agrees with' $K_C$ on the classical qubitstrings. In Theorem \ref{thm:677} and its proof, $P$ and $C$ will stand for computable measure machines.
\begin{thm}
\label{thm:677}
 For all rational $\epsilon>0$ and all $C$, there exists a $P$ such that $K_P(\sigma) \leq  QK^{\epsilon}_C(|\sigma\big>\big<\sigma|)+1$ for all classical bitstrings $\sigma$.
\end{thm}

\begin{proof}
The proof is almost identical to that of Theorem \ref{thm:67}. Fix a rational $\epsilon>0$ and a $C$. Consider the machine $P$ from the proof of Theorem \ref{thm:67} but with $\mathbb{U}$ replaced by $C$. We now show that $\mu($dom$(P))$ is computable. Let $\delta>0$ be arbitrary. Since $\mu($dom$(C))$ is computable, find a stage $t$ so that $\mu(\text{dom}(C))-\mu \text{(dom}( C_t )) < \delta.$
(The $t$ can be found as follows: Compute a $q'$ such that $|q-q'|< \delta/2$. So, $q'-\delta/2<q<q'+ \delta/2$. Since $\mu \text{(dom}( C_t ))\nearrow q$, as $t\rightarrow \infty$, we can compute a $t$ such that, $q'-\delta/2<\mu \text{(dom}( C_t ))<q'+ \delta/2$.). We may compute $S$, the set of those strings $\pi=\sigma\tau \in$ dom$(P)$ and $\sigma \in $ dom$(C_t)$. So, dom$(P)\backslash S$ consists of strings $\pi=\sigma\tau$ such that $\sigma\in$ dom$(C)\backslash$dom$(C_t)$. So, it is easy to see that $\mu($dom$(P))$-$\mu(S) < \mu($dom$(C))$-$\mu($dom$(C_t)) <\delta$. As $\delta>0$ was arbitrary, this shows that $\mu($dom$(P))$ is computable.
Now, let $\sigma \in 2^n$ be any classical bitstring such that $QK_C^{\epsilon}(|\sigma\big>\big<\sigma|)<\infty$. Let $\lambda$ and $F \subseteq \mathbb{C}^{2^{n}}_{alg}$ orthonormal such that $|\lambda|+ \text{log}(|F|)=QK_C^{\epsilon}(|\sigma\big>\big<\sigma|)$,  $\sum_{v \in F} \big<v |\sigma\big>\big<\sigma|v \big> > \epsilon$ and $C(\lambda)=F$. Let $O:= \sum_{v\in F} |v\big>\big<v|$. Note that since
$\epsilon<\sum_{v \in F} \big<v |\sigma\big>\big<\sigma|v \big> = \big<\sigma|O|\sigma\big>$, $\sigma \in S^{\epsilon}_{E,O}$ where $E$ is the standard basis. Let $\tau$ be a length $ \lceil \text{log}(\epsilon^{-1}|F|)\rceil$ string such that $g(\tau)=\sigma$. Then, we see that $P(\lambda \tau) = \sigma $. So, \[K_P(\sigma)\leq  |\lambda|+|\tau|\leq |\lambda|+\text{log}(\epsilon^{-1})+1+ \text{log}(|F|) = QK^{\epsilon}_C(|\sigma\big>\big<\sigma|)+1.\]

\end{proof}
\begin{remark}
\label{rem:qkc}
Theorem \ref{thm:677} establishes one direction of the coincidence of $K_C$ and $QK_C$ for classical qubitstrings. In the other direction, take some classical bitstring $\sigma$ with $K_C(\sigma)<\infty$ for some $C$. Let $C(\pi)=\sigma$ and $|\pi|=\text K_C(\sigma)$. Then, letting $F=\{\sigma\}$ in \ref{defn:4}, $QK_C^{\epsilon}(|\sigma\big>\big<\sigma|)\leq QK_C^{1}(|\sigma\big>\big<\sigma|)\leq |\pi|=K_C(\sigma)$, for any $\epsilon>0$. 
\end{remark}
\subsection{Quantum Schnorr randomness and $QK_C$}
Theorem \ref{thm:schnor} is a quantum analogue of the classical characterization of Schnorr randomness: $X$ is Schnorr random if and only if for any computable measure machine, $C$, there is a constant $d$ such that for all $n$, $K_{C}(X\upharpoonright n)>n-d$.
\begin{thm}
\label{thm:schnor}
A state $\rho$ is quantum Schnorr random if and only if for any computable measure machine, $C$ and any $\epsilon > 0$, there is a constant $d>0$ such that for all $n$, $QK^{\epsilon}_{C}(\rho_{n})>n-d$.
\end{thm}
\begin{proof}
($\Rightarrow$) We prove it by contraposition. I.e., show that $\rho$ is not quantum Schnorr random if there is a $C$ and an $\epsilon>0$ such that for all $d$, there is an $n=n_d$ such that $QK^{\epsilon}_{C}(\rho_{n}) \leq n-d$. Let $T_{s}$ be the set of all $\sigma$ such that $C_{s}(\sigma)\downarrow = F_{\sigma}$, an orthonormal set such that $|\sigma|+$ log $|F_{\sigma}|< n_{\sigma} $ and $F_{\sigma} \subseteq  \mathbb{C}^{2^{n_{\sigma}}}$ for some $n_{\sigma}$. Let $T=\bigcup_{s}T_{s}.$ For all strings $\sigma$, let $P_{\sigma}:= \sum_{v\in F_{\sigma}} |v \big>\big<v |.$
Let $Q_s$ be the sequence of those $P_{\sigma}$ for $\sigma \in T_s$, $Q$ the sequence of those $P_{\sigma}$ for $\sigma \in T$ and $D_s$ the sequence of those $P_{\sigma}$ for $\sigma \in T\backslash T_s$.
Next, we show that $\alpha:=\sum_{P\in Q} \tau(P)$ is computable by showing how to approximate it within $2^{-k}$ for an arbitrary $k$:  Computably find a $t$ (using the same method as in Theorem \ref{thm:677}) such that $
    \mu(\text{dom}(C )) -\mu \text{(dom}(C_t ))  < 2^{-k}.$
 We show that  $\sum_{P_{\sigma}\in Q_{t}}\tau(P_{\sigma})$ is within $2^{-k}$ of $\alpha$. Note that for all $\sigma \in T$, $2^{|\sigma|}|F_{\sigma}|< 2^{n_{\sigma}}$. So, $ \tau(P_{\sigma})= 2^{-n_{\sigma}}|F_{\sigma}|< 2^{-|\sigma| }$. 
\[\alpha -\sum_{P_{\sigma}\in Q_{t}} \tau(P_{\sigma})  = \sum_{P_{\sigma}\in D_t} \tau(P_{\sigma}) =
     \sum_{\sigma \in T/T_{t}}|F_{\sigma}|2^{-n_{\sigma}} 
     \leq\sum_{\sigma\in T/T_{t}}2^{-|\sigma|}\]\[
     \leq \sum_{\sigma\in\text{dom}( C)/\text{dom}( C_t ) }2^{-|\sigma|}
     \leq  (\mu(\text{dom}( C ))-\mu (\text{dom}( C_t )))< 2^{-k}\]

Note that $\sum_{P_{\sigma}\in Q_{t}}\tau(P_{\sigma})$ is a rational, uniformly computable in $t$ since dom$(C_t)\subseteq 2^t$ is uniformly computable in $t$. This shows that $Q$ is a quantum Schnorr test. By the assumption, we see that is a infinite sequence $d_1 < d_2 < \cdots$ and a list of distinct natural numbers $n_{d_{1}}, n_{d_{2}} \cdots$ so that for all $i$, there is a $P_i$ in $Q$ such that Tr$(P_i\rho_{n_{d_{i}}})>\epsilon$. So, $\rho$ fails $Q$ at $\epsilon$.

($\Leftarrow$)
We prove it by contraposition. Suppose that $\rho$ fails a quantum-Schnorr test, $(S^{r})_r$ at $\epsilon$. For all $j$, let $s_j$ be the least $t$ such that \[\sum_{i=0}^{t} \tau (S^{i}) > \alpha -2^{-j}.\]
We show how the sequence $(s_j)_j$ can be computed. First, let $\sum_{r}\tau(S^{r})=\alpha$, be a computable real which is not a dyadic rational. $s_j$ may be computed as follows: Note that as $\alpha$ is not a dyadic rational but $\tau(S^i)$ is a dyadic rational for all $i$, we have that
\[\sum_{i=0}^{s_{j}-1} \tau (S^{i}) < \alpha-2^{-j} < \sum_{i=0}^{s_{j}} \tau (S^{i}).\]
By Proposition 5.1.1 in \cite{misc1}, the left cut, $L(\alpha-2^{-j})$ of $\alpha-2^{-j}$ is computable. So, we may search for rationals $q \in L(\alpha-2^{-j}), q' \notin L(\alpha-2^{-j}) $ and for a $t$ such that,
\[\sum_{i=0}^{t-1} \tau (S^{i}) \leq q<q'\leq \sum_{i=0}^{t } \tau (S^{i}).\]
This $t$ is the needed $s_j$.
Now, let $\alpha$ be a dyadic rational. Then, $\alpha - 2^j$ has a finite binary representation and $s_j$ can be directly computed. So, in summary, the $(s_j)_j$ is a computable sequence, after (non-uniformly) knowing whether $\alpha$ is a dyadic rational or not. 
For all $r\geq 0$, define special projections
\[ G_{r}:= \sum_{i=s_{r}+1}^{s_{r+1}}S^{i}.\]
So,
\[ \tau(G_{r}) \leq \sum_{i=s_{r}+1}^{\infty}\tau(S^{i}) = \alpha -\sum_{i=0 }^{s_{r}}\tau(S^{i}) < 2^{-r} .\]

Notation: Let each $S^i$ be an operator on $ \mathbb{C}^{2^{n_{i}}}$. Let $n_{r} =$ max$\{n_{i}:s_{r}+1 \leq i \leq s_{r+1}\}$. By tensoring with the identity, we may assume that all
$S^i$, for $s_{r}+1 \leq i \leq s_{r+1}$, are operators on $\mathbb{C}^{2^{n_{r}}}$. Let $F_r \subseteq \mathbb{C}^{2^{n_{r}}}$
be an orthonormal set of complex algebraic vectors spanning the range of $G_r$.
Define a computable measure machine, $C$ as follows. 
On input $0^{r}10$, $C$ outputs $F_{2r}$ and on input $0^{r}11$, $C$ outputs $F_{2r+1}$. $C$ is clearly prefix-free and the measure of its domain is $\sum_{r}2^{-r+2}$, which is computable. Since each $G_r$ is a finite sum of the $S^i$s and as $\rho$ fails $(S^i)_i$ at $\epsilon$, there exist infinitely many $r$ such that Tr$(\rho_{n_{r}}G_{r})>\epsilon$. Since  we may let $n_r$ be strictly increasing in $r$, there are infinitely many such $n_r$. Fix such an $n_r$ and let $x= \lfloor r/2 \rfloor$ (I.e., $r=2x$ or $r=2x+1$). Then, QK$^{\epsilon}_{C}(\rho_{n_{r}}) \leq x+2+n_{r}+\text{log}\tau(G_{r}) \leq x+2+n_{r}- 2x.$ So, $\lfloor r/2 \rfloor-2 \leq n_{r}- $QK$^{\epsilon}_{C}(\rho_{n_{r}}) $. Letting $r$ go to infinity completes the proof.

\end{proof}

Theorem \ref{thm:Chaitin} is a Chaitin-type characterization of quantum-Schnorr randomness using QK$^{\epsilon}_{C}$. Together with Theorem \ref{thm:677} and lemma 3.9 in \cite{bhojraj2020quantum}, it implies that Schnorr randoms have a Chaitin type characterization in terms of $K_C$ (Theorem \ref{class}). To the best of our knowledge, this is the first, albeit simple, instance where results in quantum algorithmic randomness are used to prove a new result in the classical theory.
\begin{thm}
\label{thm:Chaitin}
$\rho$ is quantum Schnorr random if and only if for all computable measure machines $C$ and all $\epsilon$, $\forall d \forall^{\infty} n$ QK$^{\epsilon}_{C}(\rho_{n})>n+d$.
\end{thm}
\begin{proof}
($\Rightarrow$)
Suppose toward a contradiction that there is a $C$, an $\epsilon>0$ and $c>0$ such that there are infinitely many $n$ with $ \text{QK}^{\epsilon}_{C}(\rho_{n})\leq n+c.$ Define a quantum Schnorr test $Q$ as follows. Let $T_{s}$ be the set of all $\sigma$ such that $C_{s}(\sigma)\downarrow = F_{\sigma}$, an orthonormal set such that $|\sigma|+$ log $|F_{\sigma}|< n_{\sigma} + c$ and $F_{\sigma} \subseteq  \mathbb{C}^{2^{n_{\sigma}}}$ for some $n_{\sigma}$. Let $T=\bigcup_{s}T_{s}.$ For all strings $\sigma$, let $P_{\sigma}:= \sum_{v\in F_{\sigma}} |v \big>\big<v |.$
Let $Q_s$ be the sequence of those $P_{\sigma}$ for $\sigma \in T_s$ and $Q$ the sequence of those $P_{\sigma}$ for $\sigma \in T$. That $Q$ is a quantum Schnorr test is shown by replacing $2^{-k}$ by $2^{-k-c}$ in the $\Longrightarrow$ direction of the proof of Theorem \ref{thm:schnor}. For any $n$ such that $ \text{QK}^{\epsilon}_{C}(\rho_{n})< n+c$, there is a $\sigma \in T$ such that Tr$(P_{\sigma} \rho_{n})>\epsilon$. So, $\rho$ fails $Q$ at $\epsilon$.

($\Leftarrow$) If $\rho$ is not quantum Schnorr random then by Theorem \ref{thm:schnor}, there is a $C$ and an $\epsilon$ such that $\forall d \exists n$ such that QK$^{\epsilon}_{C}(\rho_{n})\leq n-d$.
\end{proof}
We now show the classical version of Theorem \ref{thm:Chaitin}.
\begin{thm}
\label{class}
An infinite bitstring $X$ is quantum Schnorr random if and only if for all computable measure machines $C$, $\forall d \forall^{\infty} n$ K$_{C}(X\upharpoonright n)>n+d$.
\end{thm}
\begin{proof}
$(\Longrightarrow):$ Suppose first that $X$ is Schnorr random. Then, $\rho:=\rho_X$, the state induced by $X$ is quantum Schnorr random by lemma 3.9 in \cite{bhojraj2020quantum}. Suppose for a contradiction that there is a $C$ and a $d$ such that $\exists^{\infty} n$ such that $K_C(X\upharpoonright n) \leq n+d$. By Remark \ref{rem:qkc}, $\exists^{\infty} n$ such that $QK^{0.5}_C(\rho_n) \leq n+d$, contradicting Theorem \ref{thm:Chaitin}.
($\Longleftarrow): $ Suppose that $X$ is not Schnorr random. Once again, by lemma 3.9 in \cite{bhojraj2020quantum}, we have that $\rho:=\rho_X$ is not quantum Schnorr random. By Theorem \ref{thm:Chaitin}, there is a $C$, an $\epsilon$ and a $d$ such that $\exists^{\infty} n$ such that  QK$^{\epsilon}_{C}(\rho_{n}) \leq n+d$. By Theorem \ref{thm:677}, there is a $P$ such that $\exists^{\infty} n$ such that  $K_{P}(\rho_{n}) \leq n+d+1$, a contradiction.
\end{proof}
\chapter{Generating classical randomness from a non-quantum random state}\label{3}
\section{Introduction}

 This chapter investigates the following question: Can a non-random, computable quantum source be used to generate a random sequence of bits? We answer it in the affirmative by constructing a computable, non-random sequence of qubits which yields a Martin-L{\"o}f random bitstring with probability one when `measured'.  Martin-L{\"o}f randomness, a strong notion of randomness, quantifies the `true' randomness of an infinite bitstring\cite{misc} (See \cite{DBLP:books/daglib/p/Calude17} for a nice discussion of the distinction between true randomness and pseudorandomness and an overview of various schemes to extract true randomness from quantum sources). In fact, the sequence of qubits we construct yields an arithmetically random bitstring with probability one. Arithmetic randomness is a notion of randomness strictly stronger than Martin-L{\"o}f randomness \cite{misc1}.
 
 We first formalize our main question in the language of quantum algorithmic randomness\cite{bhojraj2020quantum, unpublished}. While versions of this question have been studied in the past  \cite{Pironio_2010, Baumeler2017KolmogorovAF, Kessler_2020,DBLP:phd/hal/Abbott15,DBLP:conf/birthday/AbbottCS15, qrng2020}, this work is the first one to study it using notions from quantum algorithmic randomness.

 We let $2^{\omega}$ denote Cantor space (the collection of infinite sequences of bits), let $2^n$ denote the set of bit strings of length $n$, $2^{<\omega} := \bigcup_{n} 2^n$ and let $2^{\leq \omega}:=  2^{<\omega} \cup 2^{ \omega}$.  Martin-L{\"o}f randomness (MLR) and Quantum-Martin-L{\"o}f randomness (q-MLR) has been defined already in previous chapters. Our motivating question can now be framed as: Is there a computable, non q-MLR state which can be used to `generate' a MLR sequence of bits. To make the question fully precise, we define `generate'.

Measuring a finite dimensional quantum system is a pivotal concept in quantum theory \cite{bookA}. It hence seems natural to extend the notion of measurement from finite dimensional systems to states, which are coherent, increasing sequences of finite dimensional systems. We define (see Section \ref{measuring}) such a notion and explain how measuring a state yields an infinite bitstring. With this notion in hand, our main question assumes the precise form: Is there a computable, non q-MLR state which yields a MLR bitstring with probability one when measured?

We give an overview of the chapter.  Section \ref{measuring} formalizes how `measurement' of a state in a computable basis induces a probability measure on Cantor space. Section \ref{mr} introduces the key notion of measurement randomness for states. A state is defined to be `measurement random' (mR) if the measure induced by it, under any computable basis, assigns probability one to the set of Martin-L{\"o}f randoms. Equivalently, a state is mR if and only if measuring it in any computable basis yields a Martin-L{\"o}f random with probability one. 

We then show that quantum-Martin-L{\"o}f random states are mR. As an answer to our main question, we show in Section \ref{main} that the converse fails: there is a computable mR state, $\rho$ which is not quantum-Martin-L{\"o}f random. In fact, something stronger is true. Measuring $\rho$ in any computable basis yields an \emph{arithmetically} random sequence with probability one. Our result hence provides a scheme for generating randomness from a quantum source. To the best of our knowledge, none of the schemes proposed so far \cite{Pironio_2010, Baumeler2017KolmogorovAF, Kessler_2020,DBLP:phd/hal/Abbott15,DBLP:conf/birthday/AbbottCS15, qrng2020} generate  arithmetic randomness.

Section \ref{mrequalqmlr} shows that mR is equivalent to q-MLR for a certain special class of states.

 Let $A\in 2^{\omega}$. We define an $A$-computable function to be a total function that can be realized by a Turing machine with $A$ as an oracle. By `computable', we will refer to $\emptyset$-computable. The concept of an $A$-computable sequence of natural numbers will come up frequently in our discussion.
\begin{defn}
A sequence $(a_n)_{n\in \mathbb{N}}$ is said to be $A$-computable if there is a $A$-computable function $\phi$ such that $\phi(n)=a_n$
\end{defn}
\section{Measuring a state}
\label{measuring}
To fix notation, let $X(n)$ denote the $n$th bit of an $X \in 2^{\leq \omega}$ , let $p(E)$ stand for the probability of the event $E$.
\begin{defn}
\label{defn:78}
An A-computable measurement system $B= ((b^{n}_{0},b^{n}_{1}))_{n=1}^{\infty}$ (or just `measurement system' for short) is a sequence of orthonormal bases for $\mathbb{C}^{2}$ such that each $b^{n}_{i}$ is complex algebraic and the sequence $ ((b^{n}_{0},b^{n}_{1}))_{n=1}^{\infty}$ is A-computable.
\end{defn}

Let $\rho=(\rho_n)_{n=1}^{\infty}$ be a state and $B= ((b^{n}_{0},b^{n}_{1}))_{n=1}^{\infty}$ be a measurement system. We now work towards formalizing a notion of \emph{qubitwise} measurement of $\rho$ in the bases in $B$.
 A (probability) premeasure \cite{misc1},$p$ (also called a measure representation \cite{misc}), is a function from the set of all finite bit strings to $[0,1]$ satisfying  $\forall n$, $\forall \tau \in 2^n, p(\tau)=p(\tau 0)+p(\tau 1)$. $p$ induces a measure on $2^{\omega}$ which is seen to be unique by Carath\'eodory's extension theorem (See 6.12.1 in \cite{misc1}). Flipping a $0,1$ sided fair coin repeatedly induces a probability measure (which happens to be the uniform measure) on $2^{\omega}$ as follows. Let the random variable $Z(n)$ denote the outcome of the the $n$th coin flip. The sequence $(Z(n))_{n\in \mathbb{N}}$ induces a premeasure, $p$, on $ 2^{<\omega}$ which extends to the uniform measure on $2^{\omega}$. Here, $p(\sigma)=2^{-n}$ is the probability that $Z(i)=\sigma(i)$ for all $i\leq |\sigma|$. Similarly the act of measuring $\rho$ qubit by qubit in $B$ induces a premeasure on $2^{<\omega}$ which extends to a probability measure (denoted $\mu_{\rho}^{B}$) on $2^{\omega}$ as follows. Let the random variable $X(n)$ be the $0,1$ valued outcome of the measurement of the $n$th qubit of $\rho$. Let $p$ be the premeasure induced by the sequence $(X(n))_{n \in \mathbb{N}}$ on $ 2^{<\omega}$. $p$ extends to $\mu_{\rho}^{B}$ on $2^{\omega}$. For any $A \subseteq 2^{\omega}$, $\mu_{\rho}^{B}(A)$ is the probability that $X \in A$ where $X$ is the element of $2^{\omega}$ obtained in the limit by the qubit by qubit measurement of $\rho$ in $B$. The most conspicuous difference between the two situations is that while the $(Z(n))_{n\in \mathbb{N}}$ are independent, $(X(n))_{n\in \mathbb{N}}$ need not be independent as the elements of $\rho$ can be entangled. We now formalize the above. The following calculations follow from standard results mentioned, for example, in \cite{bookA}. 
 
 We now define $(X(n))_{n \in \mathbb{N}}$ and $p$, the induced  premeasure. Measure $\rho_1$ by the measurement operators $\{|b^{1}_0\big>\big<b^{1}_0|,|b^{1}_1\big>\big<b^{1}_1|\}$ and define $X(1) := i$ where $i\in \{0,1\}$ is such that $b^{1}_{i}$ was obtained by the above measurement. Let $\hat{\rho_{2}}$ be the density matrix corresponding to the post-measurement state of $\rho_2$ given that $\rho_2$ yields $|b^{1}_{X(1)}\big> \big<b^{1}_{X(1)}| \otimes I$ if measured in the system 
 \[ (|b^{1}_{i}\big> \big<b^{1}_{i}|   \otimes I )_{i\in\{0,1\}}.\]
  
 I.e,
\[\hat{\rho_{2}} = \dfrac{(|b^{1}_{X(1)}\big> \big<b^{1}_{X(1)}| \otimes I) \rho_2 (|b^{1}_{X(1)}\big> \big<b^{1}_{X(1)}| \otimes I)}{tr((|b^{1}_{X(1)}\big> \big<b^{1}_{X(1)}| \otimes I) \rho_2 \big)}. \]

 To define $X(2)$, measure $\hat{\rho_{2}}$ by the measurement operators 
\[ (I \otimes |b^{2}_{i}\big> \big<b^{2}_{i}|)_{i\in\{0,1\}},\]
and set $X(2):= i$ where $i \in \{0,1\}$ is such that $I \otimes |b^{2}_{i}\big> \big<b^{2}_{i}|$ is obtained after the measurement. We use $\hat{\rho_{2}}$ instead of $\rho_2$ to define $X(2)$ to account for the previous measurement of the first qubit. $X(n)$ is defined similarly.
By the above, 
 \[p(ij):=p(X(1)=i,X(2)=j) = p(X(1)=i)p(X(2)=j|X(1)=i)=\]\[p(X(1)=i) tr \big[I \otimes |b^{2}_{j}\big> \big<b^{2}_{j}| (\dfrac{(|b^{1}_{i}\big> \big<b^{1}_{i}| \otimes I) \rho_2 (|b^{1}_{i}\big> \big<b^{1}_{i}| \otimes I)}{tr((|b^{1}_{i}\big> \big<b^{1}_{i}| \otimes I) \rho_2)})\big].\]
 
 Since $PT_{\mathbb{C}^{2}}(\rho_2)=\rho_{1}$, $p(X(1)=i)=tr((|b^{1}_{i}\big> \big<b^{1}_{i}| \otimes I) \rho_2\big)$. So, \[p(ij)=tr \big[ \rho_2 (|b^{1}_{i}b^{2}_{j}\big> \big<b^{1}_{i}b^{2}_{j}|)\big].\]
 
 Given $\tau \in 2^n$, similar calculations show that 
 \begin{align}
 \label{eq:21}
    p(\tau) := p(X(1)=\tau(1), \dots,X(n)=\tau(n)) = tr \big[ \rho_n (|\bigotimes_{i=1}^{n} b^{i}_{\tau(i)}\big> \big<\bigotimes_{i=1}^{n} b^{i}_{\tau(i)}\big| \big]. 
 \end{align}

 This defines $p$. The following lemma shows that $p(.)$ is a premeasure. Define $\mu^{B}_{\rho}$ to be the unique probability measure induced by it.
 \begin{lem}
 $\forall n$, $\forall \tau \in 2^n, p(\tau)=p(\tau 0)+p(\tau 1)$
 \end{lem}
 \begin{proof}
Noting that for $j\in \{0,1\}$,
  \[
    \rho_{n+1} (|\bigotimes_{i=1}^{n} b^{i}_{\tau(i)} \otimes b^{n+1}_{j}\big> \big<\bigotimes_{i=1}^{n} b^{i}_{\tau(i)} \otimes b^{n+1}_{j}|) =\rho_{n+1} (|\bigotimes_{i=1}^{n} b^{i}_{\tau(i)}\big> \big<\bigotimes_{i=1}^{n} b^{i}_{\tau(i)}|\otimes |b^{n+1}_{j}\big>\big<b^{n+1}_{j}|),\]
    
    and letting $A:=|\bigotimes_{i=1}^{n} b^{i}_{\tau(i)}\big> \big<\bigotimes_{i=1}^{n} b^{i}_{\tau(i)}| $, the right hand side is
     \[=
    tr\big[(A\otimes |b^{n+1}_{0}\big>\big<b^{n+1}_{0}|)\rho_{n+1} + (A\otimes |b^{n+1}_{1}\big>\big<b^{n+1}_{1}|)\rho_{n+1} ]  \]
     \[=
    tr\big[(A\otimes (|b^{n+1}_{0}\big>\big<b^{n+1}_{0}|  + |b^{n+1}_{1}\big>\big<b^{n+1}_{1}|))\rho_{n+1} ]=
    tr\big[(A\otimes I) \rho_{n+1} ] = tr[A \rho_n] = p(\tau)\]
\end{proof}
\begin{remark}
\label{rem:1}
 If $B$ is $S$-computable and $\rho$ is $T$-computable, then the sequence $\{\mu^{B}_{\rho}(\sigma)\}_{\sigma \in \mathbb{N}}$ is $S\oplus T$-computable.
\end{remark} 
Here, $S\oplus T$ is obtained by putting $S$ on the even bits and $T$ on the odd bits \cite{misc}.

\section{Measurement Randomness}
\label{mr}
Let $MLR\subset 2^\omega$ be the set of MLR bitstrings. If $\rho$ is a state and $B$ a measurement system, $\mu^{B}_{\rho} (MLR)$ is the probability of getting a MLR bitstring by a qubit-wise measurement of $\rho$ as described in the previous section. 
\begin{defn}
$\rho$ is measurement random (mR) if for any computable measurement system, B, $\mu^{B}_{\rho} (MLR)=1$ 
\end{defn}
\begin{thm}
\label{thm:mlrthenmr}
All q-MLR states are also mR states.
\end{thm}
\begin{proof}

 Let $\rho=(\rho_n)_{n\in \mathbb{N}}$ be q-MLR. Suppose towards a contradiction that there is a $\delta \in (0,1)$ and a computable $B= ((b^{n}_{0},b^{n}_{1}))_{n=1}^{\infty}$ such that $\mu^{B}_{\rho} (2^{\omega}/MLR)>\delta$. Let $(S^{m})_m$ be the universal MLT \cite{misc} and let for all $m$,
\begin{align}
\label{eq:20}
   S^{m}=\bigcup_{m\leq i} \llbracket A^{m}_{i} \rrbracket, 
\end{align}

where the $A^{m}_{i}$s satisfy the conditions of Definition \ref{def:10}. By the definition of a MLT, for all $m$ and all $i\geq m$, we can write $A^{m}_{i}= \{\tau^{m,i}_{1},\dots,\tau^{m,i}_{k^{m,i}}\} \subset 2^i$ for some $0 \leq k^{m,i} \leq 2^{i-m}$. Now define a q-MLT as follows. For all $m$ and $i\geq m$, let $\tau_{a} = \tau_{a}^{m,i}$ for convenience and define the special projection:
\begin{align}
\label{eq:23}
   p^{m}_{i}= \sum_{a\leq k^{m,i}} (|\bigotimes_{q=1}^{i} b^{q}_{\tau_{a}(q)}\big> \big<\bigotimes_{q=1}^{i} b^{q}_{\tau_{a}(q)}\big|\big).
\end{align}
Letting $P^{m}:=(p^{m}_{i})_{m\leq i}$, we see that $(P^{m})_{m \in \mathbb{N}}$ is a q-MLT (For each $m$, the sequence $(p^{m}_{i})_{m\leq i}$ is computable since $B$ and $(A^{m}_{i})_{m\leq i}$ are computable. Condition 3 in Definition \ref{def:10} implies that for all $i$, range$(p^{m}_{i})\subseteq$range$(p^{m}_{i+1})$. So, $P^{m}$ is a q-$\Sigma_{0}^{1}$ class for all $m$. $k^{m,i} \leq 2^{i-m}$ for all $m,i$ implies that $\tau(P^m)\leq 2^{-m}$ for all $m$. Since $(S^{m})_{m\in \mathbb{N}}$ is a MLT, $(P^{m})_{m\in \mathbb{N}}$ is a computable sequence.)
For all m, $ (2^{\omega}/MLR) \subseteq S^{m}$ holds by the definition of a universal MLT. Hence, since \ref{eq:20} is an increasing union and as $\mu^{B}_{\rho}(2^{\omega}/MLR)>\delta$, for all $m$ there exists an $i(m)>m$ such that
\begin{align}
\label{eq:22}
\mu^{B}_{\rho}(\llbracket A^{m}_{i(m)} \rrbracket)> \delta.
\end{align} Fix such an $m$ and corresponding $i=i(m)$ and let $A^{m}_{i}= \{\tau_{1},\dots,\tau_{k^{m,i}}\}$ for some $k^{m,i} \leq 2^{i-m}$ as in \ref{eq:23}. By \ref{eq:21} and \ref{eq:22}, we have that
\begin{align}
\label{eq:24}
\delta<\sum_{a\leq k}   p(\tau_{a}) = \sum_{a\leq k^{m,i}} tr \big[ \rho_i (|\bigotimes_{q=1}^{i} b^{q}_{\tau_{a}(q)}\big> \big<\bigotimes_{q=1}^{i} b^{q}_{\tau_{a}(q)}\big|\big)\big]=tr \big[ \rho_i\sum_{a\leq k^{m,i}} (|\bigotimes_{q=1}^{i} b^{q}_{\tau_{a}(q)}\big> \big<\bigotimes_{q=1}^{i} b^{q}_{\tau_{a}(q)}\big|\big)\big]
\end{align}
So, by \ref{eq:23} and \ref{eq:24}, we see that for all $m$ there is an $i$ such that,
\[\delta < tr[\rho_{i}p^{m}_{i}] \leq \rho(P^m).\]
So, inf$_{m}(\rho (P^{m}))>\delta$, contradicting that $\rho$ is q-MLR.
\end{proof}

\begin{defn}
$\rho=(\rho_{n})_{n \in \mathbb{N}}$ is computable if the sequence $(\rho_{n})_{n \in \mathbb{N}}$ is computable.
\end{defn}
\section{A measurement random, non q-MLR state}
\label{main}
We show that Theorem \ref{thm:mlrthenmr} does not reverse:
\begin{thm}
\label{thm:0000}
There is a computable state which is not q-MLR but is mR.
\end{thm}
\begin{proof}
All matrices in this proof are in the standard basis.
Let $\rho= \bigotimes_{n=5}^{\infty} d_{n}$ and for $N>5$, $S_{N}:= \bigotimes_{n=5}^{N} d_{n}$.
where $d_n$ is a $2^n$ by $2^n$ matrix with $2^{-n}$ along the diagonal and $r_{n}:=\lfloor 2^n/n \rfloor $ many $2^{-n}$s on the extreme ends of the anti-diagonal. Formally, define $d_n$ to be the symmetric matrix such that:
For $i \leq r_n$, $d_{n}(i,j)= 2^{-n}$ if $j=i$ or $j=2^{n}-i+1$ and $d_{n}(i,j)=0$ otherwise.
For $ r_n < i < 2^{n}-r_{n}$, $d_{n}(i,j)= 2^{-n}$ if $j=i$ and $d_{n}(i,j)=0$ otherwise. For example, $r_{3}=2$ and so, 
\[   d_3 = 
\begin{bmatrix}
2^{-3} & 0 & 0 & 0 & 0 & 0 & 0 & 2^{-3}\\
0 & 2^{-3} & 0 & 0 & 0 & 0 & 2^{-3}  & 0\\
0 & 0 & 2^{-3} & 0 & 0 & 0 & 0 & 0\\
0 & 0 & 0 & 2^{-3}  & 0 & 0 & 0 & 0\\
0 & 0 & 0 &  0 & 2^{-3} & 0 & 0 & 0\\
0 & 0 & 0 & 0 & 0 & 2^{-3} & 0 & 0\\
0 & 2^{-3}  & 0 & 0 & 0 & 0 & 2^{-3}  & 0\\
2^{-3} & 0 & 0 & 0 & 0 & 0 & 0 & 2^{-3}\\
\end{bmatrix}
\]
Clearly, $d_n$ is a density matrix.
The theorem will be proved via the following lemmas.
\begin{lem}
\label{lem:20}
$\rho$ is not q-MLR.
\end{lem}
\begin{proof} It is easy to see that zero has multiplicity $r_n$ as an eigenvalue of $d_n$. Hence, letting $q_{n}= 2^{n}-r_{n}$, the eigenpairs of $d_n$ can be listed as $\{\alpha^{n}_{i}, v^{n}_i\}_{i=1}^{2^{n}}$ where $\alpha^{n}_{i}=0$ if $ q_{n}+1 \leq i \leq 2^{n} $ and $(v^{n}_i)_{i=1}^{2^{n}}$ is a orthonormal basis of $\mathbb{C}^{2^{n}}$.

Fix a $N>5$. By properties of the Kronecker product, $S_{N}$ has a orthonormal basis of eigenvectors:  \[\{ \bigotimes_{n=5}^{N} v^{n}_{l(n)}: (l(n))_{n=5}^{N} \text{ is a sequence such that for all } n, l(n) \leq 2^n  \},\]
and $\bigotimes_{n=5}^{N} v^{n}_{l(n)}$ has eigenvalue $ \prod_{n=5}^{N} \alpha^{n}_{l(n)} $.
Letting $M_N$ be those elements of the above eigenbasis having non-zero eigenvalues, we have that
\begin{align}
\label{eq:eig}
M_{N}=\{ \bigotimes_{n=5}^{N} v^{n}_{l(n)}: (l(n))_{n=5}^{N} \text{ is a sequence such that for all } n,  l(n) \leq q_{n}  \}.  
\end{align}

(See Remark \ref{comment} for an intuitive explanation of the reason for choosing such an $M_N$.) By the definition of $q_n$, \[|M_N|= \prod_{n=5}^{N} 2^{n}-\lfloor 2^n/n \rfloor \leq \prod_{n=5}^{N} 2^{n}- (2^n/n) + 1 = \prod_{n=5}^{N} 2^{n}(1 - n^{-1} + 2^{-n}) = \prod_{n=5}^{N} 2^{n} \prod_{n=5}^{N}  (1 - n^{-1} + 2^{-n}).\]
Noting that $\prod_{n=5}^{\infty}  (1 - n^{-1} + 2^{-n}) =0 $,  define a q-MLT $(T_{m})_{m\in \mathbb{N}}$ as follows. Given $m$, we describe the construction of $T_{m}$. Find $N=N(m)$ such that
$\prod_{n=5}^{N}  (1 - n^{-1} + 2^{-n})< 2^{-m}$. Let $\gamma(N):= \sum_{n=5}^{N} n$  and let \[p_{\gamma(N)}=   \sum_{v \in M_{N}}   |v \big>\big<v |.\]
$p_{\gamma(N)}$ is a special projection on $\mathbb{C}^{2^{{\gamma(N)}}}$ having rank equal to $|M_{N}|$. Let $p_{k}=\emptyset$ for $k<\gamma(N)$ and \[p_{k}:= p_{\gamma(N)} \otimes \bigotimes_{i=1} ^{k-\gamma(N)} I   \]
for $k>\gamma(N)$. Using that $\rho$ is computable, it is easy to see that $(p_{k})_{k\in \mathbb{N}}$ is a q-$\Sigma^{0}_{1}$ class. Let $T_{m}:= (p_{k})_{k\in \mathbb{N}}$. $(T_{m})_{m\in \mathbb{N}}$ is a q-MLT since the choice of $N(m)$ implies that $\tau(T_{m}) < 2^{-m}$ and as $N(m)$ can be computed from $m$. $(T_{m})_{m\in \mathbb{N}}$ demonstrates that $\rho$ is not q-MLR as follows. Fix $m$ arbitrarily and let $N(m)$ be as above. Recalling that $M_{N}$ is the set consisting of all eigenvectors of $S_{N}$ with non-zero eigenvalue, we have that, \[\rho(T_{m}) \geq tr(\rho_{\gamma(N)}p_{\gamma(N)}) =  tr(S_{N}p_{\gamma(N)})=tr(S_{N})=1.\] Since $m$ was arbitrary, $inf_{m \in \mathbb{N}}(\rho(T_{m}))=1$.
\end{proof}
The following technical lemma, although seems unmotivated at this juncture, is crucial at a later point in the proof.  
\begin{lem}
\label{lem:21}
Let $\{[a_{i}, b_{i}]^{T}\}_{i=1}^{n} $ be a set of unit column vectors in $\mathbb{C}^2$. Let $V=\bigotimes_{i=1}^{n}[a_{i}, b_{i}]^{T}$ be their Kronecker product. If  $V=[v_{1},v_{2},\dots,v_{2^n}]^{T}$, then for all $k \leq 2^{n-1}$, we have that \[|v_{k}||v_{2^{n}-k+1}| = \prod_{i=1}^{n} |a_{i}||b_{i}|.\]
\end{lem}
\begin{proof}

For natural numbers $u$ and $q$, let $[u]_{q}$ denote the remainder obtained by dividing $u$ by $q$. We use the following convention for the Kronecker product \cite{Regalia:1989:KPU:76594.76599}:

\[\begin{bmatrix}
a_1\\
b_1\\
\end{bmatrix}
\otimes
\begin{bmatrix}
a_2\\
b_2\\
\end{bmatrix}
=
\begin{bmatrix}
a_1 a_2\\
b_1 a_2\\
a_1 b_2\\
b_1 b_2
\end{bmatrix}
.\]
So, $v_{1}= \prod_{i=1}^{n} a_{i}$ and $v_{2^{n}}= \prod_{i=1}^{n} b_{i}$. For any $k\leq 2^{n-1}$, $v_k$ has the form $v_k= \prod_{i=1}^{n} c^{k}_{i}$, for some  $c^{k}_{i} \in \{a_{i}, b_{i}\}$ and $v_{2^{n}-k+1}$ has the form $v_{2^{n}-k+1} = \prod_{i=1}^{n} e^{k}_{i}$, for some  $e^{k}_{i} \in \{a_{i}, b_{i}\}$. Note that $c^{k}_{1}=a_{1}$ if and only if $k$ is odd if and only if   $e^{k}_{1}=b_{1}$. Similarly, we have the following. $c^{k}_{2}=a_{2}$ if and only if $[k]_{2^{2}} \in \{1,2\}$ if and only if $e^{k}_{2}=b_{2}$.
$c^{k}_{3}=a_{3}$ if and only if $[k]_{2^{3}} \in \{1,\dots,2^{2}\}$ if and only if $e^{k}_{3}=b_{3}$.
In general, for $i\leq n$, for all $k\leq 2^{n-1}$, \[c^{k}_{i}=a_{i} \iff  [k]_{2^{i}} \in \{1,\dots,2^{i-1}\} \iff e^{k}_{i}=b_{i}.\]
This proves the lemma. Intuitively, this happens for the following reason. Imagine moving from $v_1$ to $v_{2^{n-1}}$ (by incrementing $k$) and keeping track of the values of $c^{k}_i$ as you move along the $v_k$s. Also, imagine moving from $v_{2^n}$ to $v_{2^{n-1}}$ and keeping track of the values of $e^{k}_i$ as you move along the $v_{2^n - k +1}$s. Both motions are in opposite directions since as $k$ is incremented, the first motion is from lower to higher indices and the second is from higher to lower indices. Consider the behavior of $c^{k}_1,e^{k}_1$ as $k$ is incremented. At the `start' point, $c^{1}_1=a_1$, $e^{1}_1=b_1$. Now, as you move (i.e as you increment $k$), $c^{k}_1$ alternates between $a_1$ and $b_1$ equalling it's starting value, $a_1$ at odd $k$s and $e^{k}_1$ alternates between $b_1$ and $a_1$ equalling it's starting value $b_1$  for odd $k$s. Now, take any $i\leq n$.
$c^{k}_i$ alternates between $a_i$ and $b_i$ in blocks of length $2^{i-1}$. $c^{k}_i=a_i$ when $k$ is in the first block, $\{1,2,\dots, 2^{i-1}\}$ (i.e, when $[k]_{2^{i}} \in \{1,2,\dots, 2^{i-1}\}$) and $c^{k}_i=b_i$ when $k$ is in the second block, $\{2^{i-1}+1,\dots, 2^{i}\}$(i.e, when $[k]_{2^{i}} \in \{2^{i-1}+1,\dots, 0\}$) and so on. Similarly, $e^{k}_i$ alternates between $b_i$ and $a_i$ in blocks of length $2^{i-1}$.
\end{proof}

\begin{lem}
\label{lem:22}
Let $n\in \mathbb{N}$ and let $\{[a_{i}, b_{i}]^{T}\}_{i=1}^{n} $ be such that for all i, $[a_{i}, b_{i}]^{T}$ is unit column vector in $\mathbb{C}^2$ and let $W=\bigotimes_{i=1}^{n}[a_{i}, b_{i}]^{T}$. Then, $|\big<W|d_{n}|W\big>| \in [2^{-n}(1-2n^{-1}),2^{-n}(1+2n^{-1})]$
\end{lem}
\begin{proof}

Fix $n$ and $V$ as in the statement and write $d_n$ as a block matrix with each block of size $2^{n-1}$ by $2^{n-1}$.
\[d_n = \begin{bmatrix}
A & B\\
B^{T} & A\\
\end{bmatrix}.
\]
Letting $V=\bigotimes_{i=1}^{n-1}[a_{i}, b_{i}]^{T}$, in block form, $W=[a_{n}V^{T}, b_{n}V^{T}]^{T} $. Let $V=[v_{1},v_{2},\dots,v_{2^{n-1}}]^{T}$. It is easily checked that 
\[\big<W|d_{n}|W\big>= 2^{-n} +  a_{n}^{*} b_{n} V^{\dagger}BV+ a_{n}b_{n}^{*} V^{\dagger}B^{T}V.\]
By the form of B we get, \[V^{\dagger}BV = 2^{-n} [v^{*}_{1},v^{*}_{2},\dots,v^{*}_{2^{n-1}}][v_{2^{n-1}},v_{2^{n-1}-1},\dots,v_{2^{n-1}-r_{n}+1},0,\dots ,0]^{T}. \]

\[= 2^{-n}\sum_{k=1}^{r_{n}}v^{*}_{k}v_{2^{n-1}-k+1}.\]
By the previous lemma, \[|V^{\dagger}BV| \leq 2^{-n}\sum_{k=1}^{r_{n}}|v_{k}||v_{2^{n-1}-k+1}| = 2^{-n} r_{n} \prod_{i=1}^{n-1} |a_{i}||b_{i}| = 2^{-n} r_{n} \prod_{i=1}^{n-1} |a_{i}|\sqrt{1-|a_{i}|^{2}}. \]
Since $x\sqrt{1-x^{2}}$ has a maximum value of $1/2$ and recalling definition of $r_n$,
\[|V^{\dagger}BV| \leq 2^{-n}\dfrac{1}{2^{n-1}}\dfrac{2^{n}}{n} = \dfrac{2^{1-n}}{n}. \]
Similarly, 
$|V^{\dagger}B^{T}V| \leq  \dfrac{2^{1-n}}{n}.$ Noting that $|a_{n}^{*} b_{n}|, |a_{n} b_{n}^{*}| \leq 1/2$, \[ |\big<W|d_{n}|W\big>| \leq 2^{-n} +  |a_{n}^{*} b_{n} V^{\dagger}BV| + |a_{n}b_{n}^{*} V^{\dagger}B^{T}V|\leq 2^{-n} +  \dfrac{2^{1-n}}{n},\]
and 
\[ |\big<W|d_{n}|W\big>| \geq 2^{-n} - |a_{n}^{*} b_{n} V^{\dagger}BV| - |a_{n}b_{n}^{*} V^{\dagger}B^{T}V|\geq 2^{-n} -  \dfrac{2^{1-n}}{n}.\]
\end{proof}
\begin{lem}
\label{lem:24}
$\rho$ is mR.
\end{lem}
If $p$ is any measure on $2^{\omega}$, we can define Martin-L{\"o}f randomness with respect to $p$ exactly as we defined it for the uniform measure. Denote by $MLR(p)$, the set of bitstrings Martin-L{\"o}f random with respect to $p$ \cite{bookE}.\\
\begin{proof}
We use ideas similar to Theorem 196(a) in \cite{bookE}. For convenience, for all $i>5$, define \[\beta_{i} := \sum_{q=5}^{i-1}q.\]
Let $B$ be any computable measurement system. We show that $ MLR(\mu^{B}_{\rho}) \subseteq MLR$. Since $\mu^{B}_{\rho}[MLR(\mu^{B}_{\rho})]=1$, this implies that $\mu^{B}_{\rho}(MLR)=1$. Denote $\mu^{B}_{\rho}$ by $\mu$ for convenience. Let $\lambda$ denote the uniform measure.  We will abuse notation by writing $\mu(\tau)$ instead of the more cumbersome $\mu(\llbracket\tau\rrbracket)$ for $\tau \in 2^{<\omega}$.
Let $X \in MLR(\mu)$. Write $X$ as a concatenation of finite bitstrings : $X=\sigma_{5}\sigma_{6}\dots \sigma_{n}\dots$ where $\sigma_{n} \in 2^n$ for all $n\in \mathbb{N}$. Let $S_n := \sigma_{5}\sigma_{6}\dots \sigma_{n}$ be the concatenation upto $n$. Let $\mu_{i}$ be such that for all $\tau \in 2^i$, \[\mu_{i}(\tau):= tr\big[d_{i}(|\bigotimes_{q=1}^{i} b^{q+\beta_{i}}_{\tau(q)}\big> \big<\bigotimes_{q=1}^{i} b^{q+\beta_{i}}_{\tau(q)}|)\big].\]
By \ref{eq:21} and by the form of $\rho$ we see that,
\[\mu(S_{n})= \prod_{i=5}^{n} \mu_{i}(\sigma_{i}).\]
Note that $\mu$ is computable \cite{bookE} since $\rho$ and $B$ are. Since $X \in MLR(\mu)$, by the Levin-Schnorr theorem (Theorem 90, section 5.6 in \cite{bookE}) there is a $C_{1}$ such that
\begin{align*}
    \forall n ,  -\log(\mu(S_{n})) - C_{1} \leq KM(S_{n}). 
\end{align*}
By Theorem 89, section 5.6 in \cite{bookE} fix a $C_2$ such that
\begin{align*}
    \forall n ,   KM(S_{n}) \leq -\log(\lambda(S_{n})) + C_{2}.
\end{align*} 
By these inequalities and taking exponents, we see that there is a constant $\alpha>0$ such that \[\forall n,  \mu(S_{n}) \geq \alpha \lambda(S_{n}).\]
Letting $r_{i}:= \mu_{i}(\sigma_{i})$ and $\delta_{i}:= \lambda(\sigma_{i})-r_{i}$ in the above,
\begin{align}
\label{eq:25}
    \forall n,  \prod_{i=5}^{n} r_{i}  \geq \alpha \prod_{i=5}^{n} r_{i}+ \delta_{i}.
\end{align}
Let $\mu'$ be a probability measure on $2^{\omega}$ such that for all $\sigma\in 2^{<\omega}, \mu'(\sigma):=2\mu(\sigma)-\lambda(\sigma)$. In particular, this implies that \[\forall n,\mu'(S_{n}) = \prod_{i=5}^{n} r_{i}-\delta_{i}.\]
Note that $\mu'$ is computable since $\mu$ and $\lambda$ are. Applying the same argument which resulted in \ref{eq:25}, we get that there is an $\epsilon>0$ such that,
\begin{align}
\label{eq:26}
    \forall n,  \prod_{i=5}^{n} r_{i}  \geq \epsilon \prod_{i=5}^{n} r_{i}- \delta_{i}.
\end{align}
By Lemma \ref{lem:22}, for all $i, r_{i}\in [2^{-i}(1-2i^{-1}),2^{-i}(1+2i^{-1})]. $
So, $|\delta_{i}|=|r_{i}-2^{-i}| \in [0,2^{-i+1}i^{-1}]$.
Hence, \\$r_{i}+\delta_{i} \geq 2^{-i}-2^{-i+1}i^{-1} - 2^{-i+1}i^{-1} = 2^{-i}[1-4i^{-1}]>0$, since $i\geq 5$. Similarly, $r_{i}-\delta_{i}\geq 0$. By this, multiplying \ref{eq:25} and \ref{eq:26} gives,
\begin{align}
\label{eq:27}
    \forall n,  \prod_{i=5}^{n} r^{2}_{i}  \geq \alpha \epsilon \prod_{i=5}^{n} r^{2}_{i}- \delta^{2}_{i}=\alpha \epsilon\prod_{i=5}^{n} r^{2}_{i} \prod_{i=5}^{n} \big(1-\dfrac{\delta^{2}_{i}}{r^{2}_{i}}\big).
\end{align}
By the above, \[\dfrac{|\delta_{i}|}{r_{i}} \leq \dfrac{2^{-i+1}i^{-1}}{2^{-i}(1-2i^{-1})} = 2(i-2)^{-1}.\] 
Letting $F>0$ be the constant,
\begin{align*}
    \forall n,  \prod_{i=5}^{n} \big(1-\dfrac{\delta^{2}_{i}}{r^{2}_{i}}\big) \geq \prod_{i=5}^{\infty} \big(1-\dfrac{\delta^{2}_{i}}{r^{2}_{i}}\big) \geq \prod_{i=5}^{\infty} \big(1-4(i-2)^{-2}\big)=F,
\end{align*}
\ref{eq:27} gives,
\begin{align}
\label{eq:28}
    \forall n,  (\alpha \epsilon)^{-1}\prod_{i=5}^{n} r^{2}_{i}  \geq  \prod_{i=5}^{n} r^{2}_{i}- \delta^{2}_{i} \geq \prod_{i=5}^{n} r^{2}_{i} F.
\end{align}
From \ref{eq:25}, \ref{eq:26} and \ref{eq:28}, it is easy to see that there is a $G>0$ such that for all $n$
\[\prod_{i=5}^{n} r_{i}+ \delta_{i} \geq G \prod_{i=5}^{n} r_{i}. \]
Recalling the definitions of $r_i$ and $\delta_i$, 
\[\forall n, \lambda(S_n) \geq G \mu(S_n).  \]
Letting $D= C_{1} -\log(G)$ and recalling the definition of $C_1$,
\[\forall n, -\log(\lambda(S_n)) \leq  -\log(\mu(S_n)) -\log(G) \leq  KM(S_n) + D. \]
By Theorem 85 in \cite{bookE}, $KM(.) \leq K(.) + O(1)$ and so there is a $E>0$ such that
\[\forall n, -\log(\lambda(S_n))  \leq  K(S_n) + E. \]
Noting that $-\log(\lambda(S_n))=|S_n|= \beta_{n}+n$, 3.2.14 from \cite{misc} implies that $X$ is MLR.
\end{proof}
The theorem is proved. 
\end{proof}

Intuitively, the non-equivalence of mR and q-MLR should not be surprising given that entanglement in composite systems cannot be detected by independent measurements of the subsystems. Let us elaborate on this remark.
\begin{remark}
\label{comment}
$\rho$ in Theorem \ref{thm:0000} is built up from $d_n$s where each $d_n$ has $r_n$ many entangled eigenvectors with non-zero eigenvalue and $r_n$ many entangled eigenvectors with zero eigenvalue. This inhomogeneity in the distribution of eigenvalues is solely due to these entangled eigenvectors (all the $2^{n}- 2r_n$ many non entangled eigenvectors of $d_n$ have the same  non-zero eigenvalues). A crucial part in showing that $\rho$ is non q-MLR was to use the inhomogeneous eigenvalue distribution to bound the size  $M_N$ (see \ref{eq:eig} in the proof of Lemma \ref{lem:20}). Heuristically speaking, the the non-quantum randomness of $\rho$ is a reflection of the non-uniform eigenvalue distribution of $d_n$ which in turn is due to the presence of entangled eigenvectors of $d_n$. It is hence reasonable to expect that the quantum non-randomness of $\rho$, which stems from entanglement, cannot be captured by measurements in the sense of Definition \ref{defn:78} using pure tensors (i.e. measuring each 2-dimensional subsystem independently).
\end{remark}

\section{Generalizations}
We sketch some ways in which the Section \ref{main}'s results generalize. Given $S\in 2^{\omega}$, we may relativize the notion of Martin-L{\"o}f randomness to define the set $MLR^{S} \subset 2^{\omega}$ of infinite bitstrings which are Martin-L{\"o}f random with respect to $S$. The halting problem, denoted by $\emptyset^{\prime} \subset \mathbb{N}$ is an incomputable set important in computability theory. Letting $\emptyset^{(n)}$ be the $n-1$th iterate of the halting problem, an element of Cantor space is said to be \emph{arithmetically random} if it is in $MLR^{\emptyset^{(n)}}$ for every $n$ (see 6.8.4 in \cite{misc1}). Given 
$ S \in 2^{\omega} $, relativizing the proof of Lemma \ref{lem:24} shows that $MLR^{S}(\mu^{B}_{\rho}) \subseteq MLR^{S}$ as follows. Take an $X\in MLR^{S}(\mu^{B}_{\rho})$. Relativizing Theorems 85 and 90 from \cite{bookE} and 3.2.14 from \cite{misc} to $S$ and noting that $KM^{S}(.) \leq KM(.)$ and following the proof of Lemma \ref{lem:24}  shows that 
$X \in MLR^{S}$. This shows that $\mu^{B}_{\rho}(MLR^{S})=1$ holds for any $S \in 2^{\omega}$ and any computable measurement system $B$. In particular, if $B$ is any computable measurement system, 
$\mu^{B}_{\rho}(MLR^{\emptyset^{(n)}})=1$ for all $n$. So, \[\mu^{B}_{\rho}\big[\bigcap_{n\in \mathbb{N}}(MLR^{\emptyset^{(n)}})\big]=1.\]

So, measuring $\rho$, the state constructed in Theorem \ref{thm:0000} in any computable measurement system yields an arithmetically random  infinite sequence of bits, with probability one. The above note naturally suggests a definition:

\begin{defn}
$\rho$ is said to be strong measurement random (strong mR), if $\mu^{B}_{\rho}(MLR^{S})=1$ holds for any $S \in 2^{\omega}$ and any computable measurement system $B$. 

\end{defn}
By Remark \ref{rem:1} and by the above discussion on relativizations, we can also consider measurement of a state in non-computable measurement systems by using an appropriate oracle. We do not explore this here.

One may ask if we can build other computable examples of $\rho$s which are not q-MLR and are mR. We note that a straightforward modification of the proof of Theorem \ref{thm:0000} yields a family of such $\rho$s. We do not provide all the details here for lack of space. 
 Let $h: \mathbb{N} \longrightarrow \mathbb{N}$ and $g: \mathbb{N} \longrightarrow (0,1)$ be computable, satisfying the following for some constants $\delta \in (0,1)$ and $F>0$: \[\prod_{n=5}^{\infty} (1- h(n)2^{-n})=0 , \prod_{n=5}^{\infty} (1- h(n)[2^{-n}-g(n)])= \delta,\]\[ \forall n, g(n) \leq 2^{-n} \text{ and } \prod_{n=5}^{\infty} \big[1- \dfrac{4g^{2}(n)h^{2}(n)}{(1-2g(n)h(n))^{2}}\big] =F .\] 
Let $\rho$ be defined as in the proof of the main Theorem but with $r_n$ replaced by $h(n)$ and with the $h(n)$ many entries on the extreme ends of the anti-diagonal of $d_n$ being equal to $g(n)$ instead of $2^{-n}$. Then, this $\rho$ is computable and mR (in fact, it is strong mR) and fails a q-MLT at order $\delta$.

\section{Measurement randomness and q-MLR for product states}
\label{mrequalqmlr}
Although Theorem \ref{thm:0000} shows that mR and q-MLR are not equivalent in general, it is interesting to investigate if these notions are indeed equivalent for certain special states.
\begin{defn}
A state $\rho = (\rho_s)_s$ is defined to be a product state if there is a $2^m$ by $2^m$ computable density matrix $d$ such that for all $n$, $\rho_{nm} = \otimes_{s=1}^{n} d$.
\end{defn}
 
\begin{thm}
\label{prod}
Measurement randomness is equivalent to q-MLR for product states.
\end{thm}
We first prove some purely linear algebraic lemmas. We will use the block matrix and block vector notation; capital letters will indicate that the block is a matrix and not a scalar. For $n\geq 1,$ unit vector $v \in \mathbb{C}^{2^n}$ will be called \emph{atomic} if it is of the form $v= \otimes_{s=1}^n v_s$ for some complex algebraic unit vectors $v_s \in \mathbb{C}^{2 }$. I.e., $v \in \mathbb{C}^{2^n}$ is atomic if it is a product tensor of $n$ many complex algebraic unit vectors, $v_s \in \mathbb{C}^{2 }$.
\begin{lem}
\label{lem:1}
If $E$ is $2^n$ by $2^n$ and $v^{\dagger} E v =0$ for all atomic $v\in \mathbb{C}^{2^n}$, then $E$ is the zero matrix.
\end{lem}
\begin{proof}
The proof is by induction. Suppose $E$ is $2^{n+1}$ by $2^{n+1}$ and satisfies the hypotheses of the lemma. Let 
\[E = \begin{bmatrix}
A & B\\
C & D\\
\end{bmatrix},
\]
where each block is $2^n$ by $2^n$. Let $X$ be an arbitrary $2^n$ by $1$, atomic column vector. Then, if $\overline{0}$ represents the $2^n$ by $1$ zero column vector,
\[ \begin{bmatrix}
1\\
0\\
\end{bmatrix} \otimes X = \begin{bmatrix}
X\\
\overline{0}\\
\end{bmatrix}
\]
is atomic too. So,\[\begin{bmatrix}
X^{\dagger} & \overline{0}
\end{bmatrix} \begin{bmatrix}
A & B\\
C & D\\
\end{bmatrix}\begin{bmatrix}
X\\
\overline{0}\\
\end{bmatrix} = X^{\dagger}AX = 0.
\]
As $X$ was an arbitrary atomic vector, $A$ is the zero matrix by the induction hypothesis. Similarly, $D$ is the zero matrix. Note that
\[ \begin{bmatrix}
\dfrac{1}{\sqrt{2}}\\
\dfrac{1}{\sqrt{2}}\\
\end{bmatrix} \otimes X = \dfrac{1}{\sqrt{2}}\begin{bmatrix}
X\\
X\\
\end{bmatrix}
\]
is atomic. So,\[0=\dfrac{1}{\sqrt{2}}\begin{bmatrix}
X^{\dagger} & X^{\dagger} 
\end{bmatrix} E\dfrac{1}{\sqrt{2}}\begin{bmatrix}
X\\
X\\
\end{bmatrix}=\begin{bmatrix}
X^{\dagger} & X^{\dagger}
\end{bmatrix} \begin{bmatrix}
0 & B/2\\
C/2 & 0\\
\end{bmatrix}\begin{bmatrix}
X\\
X\\
\end{bmatrix} = \dfrac{(X^{\dagger}BX+X^{\dagger}CX)}{2}.
\]
So, for all atomic $X$, $X^{\dagger}(B+C)X = 0$. By the induction hypothesis, $B=-C$.
\[ \begin{bmatrix}
\dfrac{i}{2}\\
\dfrac{\sqrt{3} }{2}\\
\end{bmatrix} \otimes X = \begin{bmatrix}
\dfrac{i}{2} X\\
\dfrac{\sqrt{3} }{2} X\\
\end{bmatrix}
\]
is atomic. So,\[0= \begin{bmatrix}
\dfrac{-i}{2}X^{\dagger} & \dfrac{\sqrt{3} }{2} X^{\dagger} 
\end{bmatrix} E\begin{bmatrix}
\dfrac{i}{2}X\\
\dfrac{\sqrt{3} }{2}X\\
\end{bmatrix}=\begin{bmatrix}
\dfrac{-i}{2}X^{\dagger} & \dfrac{\sqrt{3} }{2} X^{\dagger} 
\end{bmatrix}\begin{bmatrix}
0 & -B \\
B & 0\\
\end{bmatrix}\begin{bmatrix}
\dfrac{i}{2}X\\
\dfrac{\sqrt{3} }{2}X\\
\end{bmatrix} = \dfrac{i\sqrt{3}}{2} X^{\dagger}BX.
\]
So, for all atomic $X$, $X^{\dagger}BX = 0$. By the induction hypothesis, $B=0$. This proves the induction step. We omit the details of the base case (i.e., when $n=1$ and $E$ is two by two) as it can be proved similarly to the induction step.
\end{proof}
Let $I_n$ denote the $2^n$ by $2^n$ identity matrix.
\begin{lem}
\label{lem}
If $E$ is a $2^n$ by $2^n$ Hermitian matrix such that $v^{\dagger} E v = 2^{-n}$ for all atomic $v\in \mathbb{C}^{2^n}$, then $E= 2^{-n} I_n$.
\end{lem}
\begin{proof}
The proof is by induction. Suppose $E$ is  $2^{n+1}$ by $2^{n+1}$ and satisfies the hypotheses of the lemma. Note that because the standard (computational) basis vectors are atomic, $E$ has $2^{-n-1}$ along the diagonal. Let 
\[E = \begin{bmatrix}
A & B\\
B^{\dagger} & C\\
\end{bmatrix},
\]
where each block is $2^n$ by $2^n$. Let $X$ be an arbitrary $2^n$ by $1$, atomic column vector. Then, if $\overline{0}$ represents the $2^n$ by $1$ zero column vector,
\[ \begin{bmatrix}
1\\
0\\
\end{bmatrix} \otimes X = \begin{bmatrix}
X\\
\overline{0}\\
\end{bmatrix}
\]
is atomic too. So,\[\begin{bmatrix}
X^{\dagger} & \overline{0}
\end{bmatrix} \begin{bmatrix}
A & B\\
B^{\dagger} & C\\
\end{bmatrix}\begin{bmatrix}
X\\
\overline{0}\\
\end{bmatrix} = X^{\dagger}AX = 2^{-n-1}.
\]
So, for an arbitrary atomic vector  $X$ , \[X^{\dagger}(2A)X= 2^{-n}.\] Note that $2A=2A^{\dagger}$ (as $E=E^{\dagger}$) and that $2A$ has $2^{-n}$ along the diagonal. So, by the induction hypothesis, $2A =2^{-n} I_n$. Similarly, $2C=2^{-n} I_n$.
\[ \begin{bmatrix}
\dfrac{1}{\sqrt{2}}\\
\dfrac{1}{\sqrt{2}}\\
\end{bmatrix} \otimes X = \dfrac{1}{\sqrt{2}}\begin{bmatrix}
X\\
X\\
\end{bmatrix}
\]
is atomic. So,\[2^{-n-1}=\dfrac{1}{\sqrt{2}}\begin{bmatrix}
X^{\dagger} & X^{\dagger} 
\end{bmatrix} E\dfrac{1}{\sqrt{2}}\begin{bmatrix}
X\\
X\\
\end{bmatrix}=\begin{bmatrix}
X^{\dagger} & X^{\dagger}
\end{bmatrix} \begin{bmatrix}
2^{-n-2}I_n & B/2\\
B^{\dagger}/2 & 2^{-n-2}I_n\\
\end{bmatrix}\begin{bmatrix}
X\\
X\\
\end{bmatrix} \]\[= \dfrac{(X^{\dagger}BX+X^{\dagger}B^{\dagger}X)}{2}+ 2^{-n-1}.
\]
So, for all atomic $X$, Re$(X^{\dagger} B X )= 0$. Similarly, we show that Im$(X^{\dagger} B X )= 0$ as follows:

\[ \begin{bmatrix}
\dfrac{1}{\sqrt{2}}\\
\dfrac{i}{\sqrt{2}}\\
\end{bmatrix} \otimes X = \dfrac{1}{\sqrt{2}}\begin{bmatrix}
X\\
iX\\
\end{bmatrix}
\]
is atomic. So,\[2^{-n-1}=\dfrac{1}{\sqrt{2}}\begin{bmatrix}
X^{\dagger} & -iX^{\dagger} 
\end{bmatrix} E\dfrac{1}{\sqrt{2}}\begin{bmatrix}
X\\
iX\\
\end{bmatrix}\]\[=\begin{bmatrix}
X^{\dagger} & -iX^{\dagger}
\end{bmatrix} \begin{bmatrix}
2^{-n-2}I_n & B/2\\
B^{\dagger}/2 & 2^{-n-2}I_n\\
\end{bmatrix}\begin{bmatrix}
X\\
iX\\
\end{bmatrix} = \dfrac{i(X^{\dagger}BX-X^{\dagger}B^{\dagger}X)}{2}+ 2^{-n-1}.
\]
So, for all atomic $X$, $ X^{\dagger} B X  = 0$. By Lemma \ref{lem:1}, $B $ is the zero matrix. This proves the induction step. We omit the details of the base case (i.e., when $n=1$ and $E$ is Hermitian two by two) as it can be proved similarly to the induction step.
\end{proof}
We now prove Theorem \ref{prod}.
\begin{proof}
By Theorem \ref{thm:mlrthenmr}, it suffices to show that if a product state $\rho$ is mR, then it is q-MLR. Let state $\rho = (\rho_s)_s$ be a mR product state. So, there is a $2^m$ by $2^m$ computable density matrix $d$ such that for all $n$, $\rho_{nm} = \otimes_{s=1}^{n} d$. We show, using Lemma \ref{lem}, that $d$ must be the $2^{-n}I_n$, and hence that $\rho$ is q-MLR. Suppose that there is an atomic $v \in C^{2^n}$ and a $p$ such that $v^{\dagger}dv = p \neq 2^{-n}$ (Note that $p \in [0,1]$ as $d$ is a density matrix). So, $v= \otimes_{s=1}^n v_s$ for some complex algebraic unit vectors $v_s \in \mathbb{C}^{2 }$. For each $s$, let $w_s$ be the unique complex algebraic unit vector such that $v_s$ and $w_s$ form a orthonormal basis of $\mathbb{C}^2$. Define a measurement system by $B= ((b^{t}_{0},b^{t}_{1}))_{t=1}^{\infty}$ where $b^{t}_{0}:= v_s$ and $b^{t}_{1}:=w_s$ for $t=s$ (mod $n$). I.e., informally speaking, $B$ consists of copies of $v$. As $v$ has `length' equal to $n$, $B$ repeats with period $n$. Consider dividing an infinite bitstring, $X$ into blocks of length $n$ (I.e., the first block is $X(1)X(2)\cdots X(n)$, the second block is $X(n+1)X(n+2) \cdots X(2n)$ and so on). Given an $X$, let $f_X (s) =$ the number of blocks which are equal to $0^n$ in the first $sn$ many bits of $X$. By the strong law of large numbers, for $\mu^{\rho}_{B}$- almost every $X$, \[\lim_{s \rightarrow \infty}\dfrac{f_X (s)}{s}=p.\]  However, it is known that if $X$ is MLR, then \[\lim_{s \rightarrow \infty}\dfrac{f_X (s)}{s}= 2^{-n} \neq p.\]
So, $\mu^{\rho}_{B} (MLR)=0$ and so $\rho$ is not mR. So, such $v$ and $p$ cannot exist and by Lemma \ref{lem}, $d=2^{-n}I_n$.
\end{proof}

\section{Conclusion}
We constructed a computable, non-random qubitstring which almost surely yields a arithmetic random bitstring when measured. Formally, we construct a computable, non q-MLR state which yields an arithmetically random bitstring with probability one when `measured'. Arithmetic randomness is a strong form of classical randomness, strictly stronger than 
Martin-L{\"o}f randomness (See 6.8.4 in \cite{misc1} for details on arithmetic randomness). Our result hence provides further evidence for the philosophically and practically important claim that `true' randomness (as against pseudorandomness)\cite{DBLP:books/daglib/p/Calude17} can be extracted from certain quantum systems.
While several schemes exist for generating a random bitstring from a quantum source \cite{ Pironio_2010, Baumeler2017KolmogorovAF, Kessler_2020,DBLP:phd/hal/Abbott15,DBLP:books/daglib/p/Calude17, DBLP:conf/birthday/AbbottCS15, qrng2020}, to the best of our knowledge, none of these produce arithmetic randomness. It hence seems plausible that our results may prove to be relevant to the construction of quantum random number generators \cite{RevModPhys.89.015004,DBLP:books/daglib/p/Calude17 }.

Abbott, Calude and Svozil have also studied bitstrings resulting from measuring a quantum system\cite{DBLP:phd/hal/Abbott15,DBLP:conf/birthday/AbbottCS15}. However, their notion of measurement is significantly different from ours. In contrast to our work which considers measurement of an infinite sequence of qubits, they studied the randomness of a sequence of bits generated by \emph{repeatedly measuring a finite dimensional} quantum system. They go on to apply this to quantum random number generators and their certification \cite{DBLP:phd/hal/Abbott15,DBLP:journals/mscs/AbbottCS14,DBLP:conf/birthday/AbbottCS15,DBLP:books/daglib/p/Calude17}.

\section{Acknowledgements}

I thank James Hanson for many helpful discussions pertaining to Section \ref{mrequalqmlr}. Joe Miller and Peter Cholak (independently) asked if there is a notion of `measuring a state'. These questions were one of the factors which led me to explore this area. Andr\'e Nies, whom I thank for introducing me to quantum algorithmic randomness, independently suggested that one might get a measure on Cantor space by `measuring' a state.  

I thank Joe Miller, my thesis advisor, for his encouragement and guidance.

\chapter{Entropy and computable states}\label{4}

The von Neumann entropy of a density matrix is the Shannon entropy of the distribution given by its eigenvalues \cite{Nielsen:2011:QCQ:1972505}. In this chapter, we are concerned with the von Neumann entropies of the initial segments of states. Recall from Definition \ref{def:trace} that each finite initial segment, $\tau_n$ of the tracial state, $\tau$ is the maximally mixed state with a (maximum possible) von-Neumann entropy equal to $n$. This suggests that computable states whose initial segments' total eigenmass of one is `evenly spread out' over all the eigenvalues are quantum Martin-L{\"o}f random (q-MLR).  Heuristically speaking, since the eigenvectors of the initial segments of computable states are computable and hence easy to describe, the randomness of computable states cannot stem from the randomness of the individual eigenvectors of the initial segments. Rather, their randomness should result from the eigenvalue mass of their initial segments being `uniformly spread out' and hence difficult to `capture' using a quantum Martin-L{\"o}f test. The uniform spreading of the eigenvalues should be reflected in an asymptotically high von-Neumann entropy of the state (States having initial segments whose eigenvalues are evenly spread out have a `high von-Neumann entropy', where we use the quotes as the von Neumann entropy is defined for density matrices and not for states. We describe below a way to make sense of von Neumann entropy for states). 

Motivated by this, we explore the asymptotic behavior of the von-Neumann entropy of the initial segments of computable states. Let $H(\rho_n)$ be the von-Neumann entropy of $\rho_n$. We show  Theorems \ref{thm:9} and \ref{thm:99} which can be summarized as:
For any computable $\rho$, \[\exists c>0 \exists^{\infty} n H(\rho_n) > n-c \Rightarrow \rho \text { is q-MLR} \Rightarrow  H(\rho):=\text{lim}_n \dfrac{H(\rho_n )}{n} = 1.\]

Further, we provide an example to show that the first implication doesn't reverse. It is easy to see that the second doesn't reverse too. So, these implications are strict.

Recall that weak Solovay randomness is equivalent to q-MLR for computable states and hence all results here also hold for weak Solovay random states.
\begin{defn}
The von-Neumann entropy of a density matrix $\rho_{n}$ on $n$ qubits, $H(\rho_{n})$ is defined as: Let $\rho_n$ have a orthonormal eigenbasis $(|\psi^{i}\big> )_{i\leq 2^n}$ and corresponding eigenvalues $(\alpha_{i})_{i\leq 2^n}$.  So, \[ \rho_{n} = \sum_{i\leq 2^n} \alpha_{i}|\psi^{i}\big> \big< \psi^{i}| \] 
Since $\rho$ is a density matrix, $Tr(\rho)= \sum_{i}\alpha_{i} =1$. So, we can define: \[H(\rho_{n}) = -\sum_{i\leq 2^n}\alpha_{i}\text{log}_{2}(\alpha_{i})\]
\end{defn}

\section{q-MLR implies maximum entropy per qubit.}

Theorem \ref{thm:99} shows that, asymptotically speaking, computable q-MLR states have maximum von-Neumann entropy `per qubit'.

\begin{thm}
\label{thm:99}
If $\rho=(\rho_{n})_{n}$ is a computable  quantum Martin-L{\"o}f random state, then \[H(\rho) :=  \text{liminf}_{n}\dfrac{H(\rho_{n})}{n}  = 1\] 
\end{thm}
In fact, noting that for all $n$ $\dfrac{H(\rho_{n})}{n}\leq 1$ gives that limsup$_{n}\dfrac{H(\rho_{n})}{n}\leq 1$. So, the theorem implies that lim$_{n}\dfrac{H(\rho_{n})}{n} $ exists and is 1.
\\
\begin{proof}
Proof sketch: Suppose towards a contradiction that $H(\rho) < \epsilon<1$. This implies that (see Lemma \ref{lemma}) there is a $\delta >0$ such that for infinitely many $n$, there are $2^{n \epsilon }$ many eigenvalues  $\alpha_1 , ....\alpha_{2^{n \epsilon }}$ of $\rho_n$ with  $\sum_i \alpha_i > \delta$.
To prove this: one argues as follows: If no such $\delta$ existed, then one can bound  the entropies of the $\rho_n$s from below and show that $H(\rho) \geq \epsilon$. Then, the computability of $\rho$ allows us to build a test which $\rho$ fails at $\delta$.
 \\

Details:
Towards a contradiction, let $\rho=(\rho_{n})_{n}$ be q-MLR with $H(\rho)<\epsilon<1$. Let  \[\rho_{n} = \sum_{i\leq 2^n} \alpha^{n}_{i}|\psi^{n}_{i}\big> \big< \psi^{n}_{i}|
\]
where, $\alpha^{n}_{1} \geq \alpha^{n}_{2}\dots \geq\alpha_{2^{n}}^{n}$.

We begin with a lemma. 
\begin{lem}
\label{lemma}
For the above $\epsilon$, there is a $\delta>0$ such that $\exists^{\infty} n $ \[\sum_{i \leq \lceil2^{n\epsilon}\rceil}\alpha^{n}_{i} > \delta. \]
\end{lem}
The lemma says that a constant ($\delta$) amount of eigenvalue `mass' concentrates at the first $\lceil 2^{n \epsilon} \rceil$ many largest eigenvalues, infinitely often.
\begin{proof}

Suppose towards a contradiction that $\forall \delta \exists N_{\delta}$ such that,  \[ n>N_{\delta} \Rightarrow \sum_{i \leq \lceil2^{n\epsilon}\rceil}\alpha^{n}_{i} \leq \delta.\]
Fix a $\delta< 0.5$ and a $n>N_{\delta}$. For this $n$, define the sequence, $(r^{n}_{i})_{i\leq 2^{n}}$ as follows: 
For $i<\lceil2^{n\epsilon}\rceil$, let $r^{n}_{i} := \alpha^{n}_{i}$. For $\lceil2^{n\epsilon}\rceil \leq i < I$,($I \in \omega$ will be defined shortly), let $r^{n}_{i} := \alpha^{n}_{\lceil2^{n\epsilon}\rceil}$. For $I \leq i \leq 2^{n}$, let $r^{n}_{i} := 0$. Here, $I$ is picked to ensure that
\begin{align}
\label{eq:6}
     1-\alpha^{n}_{\lceil2^{n\epsilon}\rceil} \leq \sum_{i\leq I} r^{n}_{i} \leq 1 
\end{align}

So, $r_i=\alpha_{i}$ from $1 \leq i \leq \lceil2^{n\epsilon}\rceil$ and $r_i$ is the constant $\alpha^{n}_{\lceil2^{n\epsilon}\rceil}$ from $ \lceil2^{n\epsilon}\rceil \leq i <I$ where $I$ is the largest number so that the $r_{i}$s sum to less than 1.
Why does such an $I$ exist?  $\sum_{i<t} r^{n}_{i}$ increases by $\alpha^{n}_{\lceil2^{n\epsilon}\rceil}$ when $t$ increases by $1$, for $t>\lceil2^{n\epsilon}\rceil $. We can keep increasing $t$ and stop the first time $\sum_{i<t} r^{n}_{i}>1$. I.e., we find a $I$ such that $\sum_{i\leq I} r^{n}_{i}<1 \leq \sum_{i\leq I+1} r^{n}_{i}$. Since, $  \sum_{i\leq I+1} r^{n}_{i} - \sum_{i\leq I} r^{n}_{i}= \alpha^{n}_{\lceil2^{n\epsilon}\rceil}$, the inequality \ref{eq:6} holds. Let $S_{n} = \sum_{i} r^{n}_{i\leq 2^{n}} $ and we drop the subscript of $S$. 
\\
Let $p_{i}^{n}:= S^{-1}r_{i}^{n}$. So, $p=(p_{i}^{n})_{i\leq 2^{n}}$ is a probability distribution on $2^n$ and dominates $(r_{i})_i$. Let $H(p)$ be it's Shannon entropy, which we now bound from below.
 \[H(p) = -\sum_{i< \lceil2^{n\epsilon}\rceil}S^{-1}\alpha^{n}_{i}\text{log}(S^{-1}\alpha^{n}_{i}) - \sum_{ \lceil2^{n\epsilon}\rceil\leq i \leq I} S^{-1} \alpha^{n}_{\lceil2^{n\epsilon}\rceil}\text{log}(S^{-1}\alpha^{n}_{\lceil2^{n\epsilon}\rceil}).\]
 \\
$n>N_{\delta}$ implies that $\alpha^{n}_{1}<\delta<0.5$ and hence that for all $i$,
\begin{align}
\label{eq:7}
    \alpha^{n}_{i} \leq \alpha^{n}_{1}<\delta<0.5.
\end{align}
Note that $0.5<1-\delta$, since $\delta<0.5$. So, 
\begin{align}
\label{eq:8}
    0.5< 1-\delta < 1-\alpha^{n}_{\lceil2^{n\epsilon}\rceil} \leq S \leq 1
\end{align}
Putting \ref{eq:6}, \ref{eq:7} and \ref{eq:8} together, for all $i$,
\begin{align}
\label{eq:9}
    \alpha^{n}_{i} \leq \alpha^{n}_{1}<\delta<1-\delta < 1-\alpha^{n}_{\lceil2^{n\epsilon}\rceil} \leq S \leq 1
\end{align}

So, log$(S^{-1}\alpha^{n}_{i})\leq 0$ for all $i$. So, the  first sum in the expression for $H(p)$ is non-negative. This gives, 
 \[H(p) >  - \sum_{ \lceil2^{n\epsilon}\rceil\leq i \leq I} S^{-1} \alpha^{n}_{\lceil2^{n\epsilon}\rceil}\text{log}(S^{-1}\alpha^{n}_{\lceil2^{n\epsilon}\rceil}) = -(I-\lceil2^{n\epsilon}\rceil)S^{-1} \alpha^{n}_{\lceil2^{n\epsilon}\rceil}\text{log}(S^{-1}\alpha^{n}_{\lceil2^{n\epsilon}\rceil}) \]
 
 Again, since log$(S^{-1}\alpha^{n}_{\lceil2^{n\epsilon}\rceil})<0$, this sum is non-negative. So, we can ignore $S^{-1}\geq 1$ to get:
 
 \begin{align}
 \label{eq:11}
     H(p) >  -(I-\lceil2^{n\epsilon}\rceil) \alpha^{n}_{\lceil2^{n\epsilon}\rceil}\text{log}(S^{-1}\alpha^{n}_{\lceil2^{n\epsilon}\rceil}) 
 \end{align}
 
 Now, by choice of $n>N_{\delta}$, \[S= \sum_{i} r^{n}_{i} = \sum_{i<\lceil2^{n\epsilon}\rceil}\alpha^{n}_{i} + (I-\lceil2^{n\epsilon}\rceil)\alpha^{n}_{\lceil2^{n\epsilon}\rceil} \leq \delta + (I-\lceil2^{n\epsilon}\rceil)\alpha^{n}_{\lceil2^{n\epsilon}\rceil}\]
  \\
  So, 
  \[S-\delta \leq (I-\lceil2^{n\epsilon}\rceil)\alpha^{n}_{\lceil2^{n\epsilon}\rceil}\]
  
  By \ref{eq:9},
  \[1-\delta \leq 1- \alpha^{n}_{\lceil2^{n\epsilon}\rceil}  \leq S\]
This gives,
\[1-2\delta \leq (I-\lceil2^{n\epsilon}\rceil)\alpha^{n}_{\lceil2^{n\epsilon}\rceil}.\]
Since, $-\text{log}(S^{-1}\alpha^{n}_{\lceil2^{n\epsilon}\rceil})>0$, we can put  $1-2\delta$ in place of $ (I-\lceil2^{n\epsilon}\rceil)\alpha^{n}_{\lceil2^{n\epsilon}\rceil}$ in \ref{eq:11} to get:
\begin{align}
\label{eq:10}
    H(p) >  -(1-2\delta)\text{log}(S^{-1}\alpha^{n}_{\lceil2^{n\epsilon}\rceil}) 
\end{align}

Further, note that $\alpha^{n}_{\lceil2^{n\epsilon}\rceil} \leq \delta 2^{-n\epsilon}$. (If not, then, for all  $i \leq \lceil2^{n\epsilon}\rceil$, $\alpha^{n}_{i} \geq \alpha^{n}_{\lceil2^{n\epsilon}\rceil} > \delta 2^{-n\epsilon}$. This would give 
\[ \sum_{i \leq \lceil2^{n\epsilon}\rceil}\alpha^{n}_{i}>\lceil2^{-n\epsilon}\rceil2^{-n\epsilon} \delta \geq 2^{n\epsilon}2^{-n\epsilon} \delta=\delta\] contradicting the choice of $n>N_{\delta}$). So, taking log on both sides:
\[\text{log} (\alpha^{n}_{\lceil2^{n\epsilon}\rceil}) \leq \text{log} (\delta) + \text{log} (2^{-n\epsilon})\]
So, 
\begin{align}
\label{eq:12}
    -\text{log} (\alpha^{n}_{\lceil2^{n\epsilon}\rceil}) \geq -\text{log} (\delta) +  n\epsilon.
\end{align}

 Using

\[-\text{log} (S^{-1}\alpha^{n}_{\lceil2^{n\epsilon}\rceil}) = \text{log}(S)-\text{log} (\alpha^{n}_{\lceil2^{n\epsilon}\rceil}),\]

we can write \ref{eq:10} as

\begin{align}
\label{eq:13}
    H(p) >   (1-2\delta)[\text{log}(S)-\text{log} (\alpha^{n}_{\lceil2^{n\epsilon}\rceil})]
\end{align}
$\alpha=(\alpha^{n}_{i})_{i}$ and $p=(p^{n}_{i})_{i}$ are both distributions on $2^n$ with the latter dominating the former on supp($p$) (Since $S\leq 1$, we have that  $p_{i}>0 \Rightarrow p_{i} = S^{-1}r_{i}\geq r_{i} \geq \alpha_{i}$). By the definition of being a distribution, there is $i$ such that
$p_{i}=0$ and $\alpha_{i}>0$. So, $H(\rho_n) = H(\alpha) \geq H(p)$. So, by \ref{eq:13}, we have,
\begin{align}
    H(\rho_{n}) >   (1-2\delta)[\text{log}(S)-\text{log}(\alpha^{n}_{\lceil2^{n\epsilon}\rceil})]
\end{align}
Recalling that $1-\delta<S\leq 1$ we have that log$(S)>$log$(1-\delta)$. By this and \ref{eq:12}, 
\begin{align}
    H(\rho_{n}) >   (1-2\delta)[\text{log}(1-\delta)-\text{log}(\alpha^{n}_{\lceil2^{n\epsilon}\rceil})]\geq (1-2\delta)[\text{log}(1-\delta)-\text{log}(\delta) +  n\epsilon]
\end{align}
So, 
\[\dfrac{H(\rho_{n})}{n} > (1-2\delta)[\dfrac{ \text{log}(1-\delta)-\text{log}(\delta) +  n\epsilon}{n}]\]
But this holds for any $n>N_{\delta}$ and so we have this inequality holding for each $n$ in a sequence.
So, using that liminf($x_n + y_n) \geq $ liminf($x_n$) + liminf ($y_n)$  we get:
\[ H(\rho) = \text{liminf}_{n}\dfrac{H(\rho_{n})}{n} > (1-2\delta)[\text{liminf}_{n}\dfrac{\text{log}(1-\delta)-\text{log} (\delta)}{n} + \epsilon]=(1-2\delta)\epsilon\]
Recall that $\delta$ was arbitrary and so we have that for all $\delta$, $H(\rho)>(1-2\delta)\epsilon$. By assumption, $H(\rho)<\epsilon$ and so we can find a $\delta_0$ such that $H(\rho)<(1-2\delta_{0})\epsilon < \epsilon$. Contradiction. 
\end{proof}

Now, to get a contradiction, we build a q-MLT capturing $\rho$.
Let $\delta$ be as in the lemma and let $\delta$ be rational, without loss of generality.
Fix a $m$; we describe the construction of $G^{m}= (G^{m}_{n})_{n}$.
\\
Find an $n$ such that both of the following hold:
\begin{itemize}
    \item $\sum_{i \leq \lceil2^{n\epsilon}\rceil}\alpha^{n}_{i} > \delta.$
    Infinitely many such $n$ exist by the lemma
    
    \item $\dfrac{2^{n\epsilon}+1}{2^n}<2^{-m}$. This holds for almost every $n$ since $\epsilon<1$
    
\end{itemize}

Recall that $\rho$ is computable and so the $n$ can be found computably, uniformly in $m$.
Set
\[ G^{m}_{n} := \sum_{i\leq 2^{n\epsilon}} |\psi^{i}_{n}\big> \big< \psi^{i}_{n}|. \] 
For $k<n$, $G^{m}_{k} = \emptyset $ and for $k>n$, $G^{m}_{k}=G^{m}_{n} \otimes I^{k-n}$ where $I$ is the 2-by-2 identity. $2^{-n}Tr(G^{m}_{n})<2^{-m}$ by the second condition on $n$. Hence, $(G^{m})_m$ is a q-MLT and $Tr(\rho_{n}G^{m}_{n})>\delta$ by the first condition on $n$. So, $\rho$ fails this test at $\delta$.
\end{proof}

\section{ A Levin-Schnorr type condition on entropy implies q-MLR.}
Theorem \ref{thm:9} shows that a Levin-Schnorr type condition on the von-Neumann entropy implies q-MLR.
\begin{thm}
\label{thm:9}
Let $\rho$ be any (not necessarily computable) state such that for some $c \in \omega$, $H(\rho_{n})>n-c$ for almost every $n$. Then, $\rho$ is q-MLR. In fact, a weaker condition suffices and we get a stronger form of the result: Let $\rho$ be any (not necessarily computable) state such that for some $c \in \omega$, $H(\rho_{n})>n-c$ for infinitely many $n$. Then, $\rho$ is q-MLR.
\end{thm}

Before proving this, recall the following consequence of the singular value decomposition (SVD).
\begin{thm}
\label{thm:10}
Let $A$ be a $n$ by $d$ matrix with singular vectors $v_1, v_2, \dots v_r$ and corresponding singular values $\sigma_1 \geq \sigma_2, \dots \geq \sigma_r$. Let $k \leq r$ and $w_1,w_2, \dots, w_k$ be any orthonormal set. Then, 
\[ \sum_{i \leq k} |Av_{i}|^{2} \geq \sum_{i \leq k} |Aw_{i}|^{2}
\]
\end{thm}

For $\rho=\rho_{n}$ as above with eigenvalues (equal to the singular values since $\rho$ is symmetric and positive semidefinite) $\alpha_{1} \geq \alpha_{2}\dots \geq\alpha_{2^{n}}$, let \[ \sqrt \rho  = \sum_{i\leq 2^n} \sqrt \alpha_{i} |\psi^{i}\big> \big< \psi^{i}| \] 
Note that $\sqrt \rho$ has eigenpairs ($\psi_{i},\sqrt \alpha_{i}$) and is PSD and symmetric. So, it's eigenpairs are the same as it's singular vector-singular value pairs.
If $k \leq 2^n$ and $w_1,w_2, \dots, w_k$ is any orthonormal set, then since the first k singular vectors of $\sqrt \rho$ are $\psi_1, \dots, \psi_k$, theorem \ref{thm:10} gives:

\begin{align}
\label{eq:14}
    \sum_{i \leq k} |\sqrt \rho \psi_{i}|^{2} \geq \sum_{i \leq k} |\sqrt \rho w_{i}|^{2}
\end{align}
 
Let $P$ be the orthogonal projection onto the subspace spanned by $w_1,w_2, \dots, w_k$. Since $\sqrt \rho$ is self-adjoint, $\big<x|\sqrt \rho y \big> = \big< (\sqrt \rho)^{*} x| y \big> = \big< \sqrt \rho  x| y \big>$,

\[ Tr( \rho P)= \sum_{i \leq k}   Tr(  \rho |w_{i}\big>\big<w_{i}|) = \sum_{i \leq k}   \big<w_{i}|  \rho |w_{i}\big> = \sum_{i \leq k}   \big<\sqrt \rho w_{i}|  \sqrt \rho w_{i}\big> = \sum_{i \leq k} |\sqrt \rho w_{i}|^{2}
\]
Similarly, if $G$ is the orthogonal projection onto the subspace spanned by $\psi_1, \dots, \psi_k$, then \[Tr(\rho G) = \sum_{i \leq k} |\sqrt \rho \psi_{i}|^{2}
\] So, by \ref{eq:14}, for any orthonormal projection, $G$, of rank=k, if $P$ is the orthogonal projection onto the subspace spanned by the first k singular vectors of $\rho$, $\psi_1, \dots, \psi_k$, then
\begin{align}
\label{eq:15}
     Tr( \rho G) \leq Tr( \rho P) = \sum_{i \leq k}   Tr(  \rho |\psi_{i}\big>\big<\psi_{i}|)= \sum_{i \leq k}   \alpha_{i}
\end{align}
Now we can prove Theorem \ref{thm:9}.
\begin{proof}

Proof sketch: Suppose for a contradiction that $\rho$ fails a q-MLT , $(G_m )_m$ at $\delta$. Fix an $m$. As $\rho$ fails $(G_m )_m$ at $\delta$ we have that for a.e. $n$, the sum, $S_n$ of the first $2^{n-m}$ many largest eigenvalues of $\rho_n$ exceeds $\delta$. This implies that for a.e $n$, $H(\rho_n ) \leq  1- mS_n +n < 1- m\delta +n$ (to get this bound, we consider a distribution more `uniform' than that induced by $\rho_n$ and use its entropy to bound $H(\rho_n ))$. Noting that $m$ was arbitrary, we get a contradiction.
\\
Proof details:
For a contradiction, let $\rho$ satisfy the condition but not be q-MLR. Fix a q-MLT, $(G^m)_m$ and a $\delta>0$ such that $\rho(G^m)>\delta$ for all $m$. I.e, $\forall m $ for almost every $n$, we have that Tr$(\rho_{n}G^{m}_{n})>\delta$. Note that since rank$(G^{m}_{n})\leq 2^{n-m}$, $G^{m}_{n}$ is a orthogonal projection onto a subspace spanned by atmost $2^{n-m}$ orthonormal vectors. So, by \ref{eq:15}, we have that for all $m$ for a.e $n$,
\begin{align}
\label{eq:16}
    \delta<Tr(\rho_{n} G^{m}_{n})\leq \sum_{i \leq 2^{n-m}}   \alpha^{n}_{i}
\end{align} 
For a fixed $m$ take a $N_{m}$ such that for all $n>N_{m}$, \ref{eq:16} holds. For an $n>N_{m}$, let \[S_{m,n} = \sum_{i \leq 2^{n-m}}   \alpha^{n}_{i}.\] Let $(r_{n}^{m}(i))_{i\leq 2^n}$ be the distribution on $1,2,....2^n$:\\
 $r(i)= S_{m,n} 2^{m-n}$ if $i\leq 2^{n-m}$\\
 $r(i)= (1-S_{m,n})/2^{n}(1-2^{m})$ if $ 2^{n-m}<i\leq 2^n$.
 So, $(r_{n}^{m}(i))_{i\leq 2^n}$ is uniform on the 2 blocks $\{1,...,2^{n-m}\}$ and $\{ 2^{n-m}+1,....,2^{n}\}$. It's mass on the first block is $S_{m,n}$ and that on the second is $1-S_{m,n}$. The distribution of the total mass between the 2 blocks is same for both $(r_{n}^{m}(i))_{i\leq 2^n}$ and $(\alpha_{n}(i))_{i\leq 2^n}$. Within each block, $(r_{n}^{m}(i))_{i\leq 2^n}$ is more uniform than $(\alpha_{n}(i))_{i\leq 2^n}$. I.e, $(r_{n}^{m}(i))_{i\leq 2^n}$ is obtained by first considering $(\alpha_{n}(i))_{i\leq 2^n}$ on each block and then by uniformly distributing the mass within each block. So, by exercise 2.18 on page 50 of the book by Thomas and Cover \cite{Cover2006}, we have that-
 \begin{align}
 \label{eq:17}
     H(\rho_n)\leq H((r_{n}^{m}(i))_{i\leq 2^n})
 \end{align}
 Recall that this holds for $\forall m$ and $\forall n >N_{m}$. Let, $(r_{n}^{m}(i))_{i\leq 2^{n}} $ be denoted by $r^{m}_{n}$.
 We now bound $H(r^{m}_{n})$ from above to get a contradiction.
 
 \[H(r^{m}_{n})=-[ 2^{n-m}(S_{m,n} 2^{m-n}\text{log}(S_{m,n} 2^{m-n}))+(2^{n}-2^{n-m})(1-S_{m,n})/2^{n}(1-2^{m})\text{log}((1-S_{m,n})/2^{n}(1-2^{m}))]\]
 
 This simplifies to: 
 \[H(r^{m}_{n})=h(S_{m,n}) - mS_{m,n} + n +(1-S_{m,n})\text{log} (1-2^{-m})\]
 where $ h(S_{m,n}) = -S_{m,n} \text{log}(S_{m,n})-(1-S_{m,n})\text{log} (1-S_{m,n})$ is a positive real number between 0 and 1 and so we can replace it by 1 to get an upper bound.
$\text{log}(1-2^{-m})<0$ is negative, $(1-S_{m,n})>0$ and so $(1-S_{m,n})\text{log} (1-2^{-m})<0$. So, we can drop it to get an upper bound, 
\begin{align}\label{eq:18}
    H(r^{m}_{n}) \leq 1 - m S_{m,n} + n.
\end{align}
By \ref{eq:17} and \ref{eq:18}, we get that for all $m$ for all $n>N_{m}$, 
\begin{align}
\label{eq:19}
    H(\rho_{n}) \leq H(r^{m}_{n}) \leq 1 - mS_{m,n} + n
\end{align}
By assumption fix a $c$ such that $\forall n$,  $H(\rho_n) \geq n- c $. In fact, all we need is that

\begin{align}
\label{eq:220}
    \exists^{\infty}n ,   H(\rho_n) \geq n- c
\end{align}. 

Now, for all $m$, \ref{eq:19} holds for almost every $n$ and \ref{eq:220} holds for infinitely many $n$. So, for all $m$, \ref{eq:19} and \ref{eq:220} hold for infinitely many $n$. I.e, \[\forall m \exists^{\infty} n,\text{ such that } n-c \leq  H(\rho_{n})\leq 1 - mS_{m,n} + n\]
So, 
\[\forall m \exists^{\infty} n,\text{ such that }  -c \leq    1 - mS_{m,n}.  \]
So,
\[\forall m \exists^{\infty} n,\text{ such that }   c \geq    -1 + mS_{m,n}.  \]
Noting that $S_{m,n}>\delta$ for $n>N_{m}$, we get that
\[\forall m \exists^{\infty} n,\text{ such that }   c \geq    -1 + mS_{m,n} \geq -1 + m \delta .  \]
So,
\[\forall m,    c +1 \geq   + m \delta . \]
  Contradiction.
\end{proof}

\begin{remark}
Any state differing from the tracial state, $\rho$ ($\rho= (\rho_n )_n $ where $\rho_n = 2^{-n} I_n$ and $I_n$ is the $2^n$ by $2^n$ identity), at only finitely many qubits clearly satisfies the hypothesis of \ref{thm:9}. We construct another one: Let $f$ be any function on $(0,1)$ satisfying

\[\int_{0}^{1}f(s)ds =1 \text{ and } -\infty < -\int_{0}^{1}f(s)\text{log}(f(s))ds < \infty.\]

For example, let\[f(x)= \dfrac{2}{x(1-ln x)^{3}},\] on $(0,1)$ where $ln$ stands for the natural logarithm.
Define a diagonal state $\rho = (\rho_{n})_{n}$ as follows. Fix $n$. For all $\sigma \in 2^n$, let
\[\alpha_{\sigma} = \int_{[\sigma]} f(s)ds\]
Here, $[\sigma] \subset (0,1]$ is the open interval defined by the string $\sigma$. \[\rho_{n}= \sum_{\sigma \in 2^{n}} \alpha_{\sigma} |\sigma\big>\big<\sigma|.\]
$\rho$ is a coherent state since $\alpha_{\sigma}= \alpha_{\sigma 1} + \alpha_{\sigma 0}$ by definition and since $\int_{0}^{1}f(s)ds =1 $.  Now we show that $\rho$ satisfies the hypothesis of \ref{thm:9}. For any $n$, by definition of the $\alpha$s, we have, \[H(\rho_{n}) = -\sum_{\sigma \in 2^{n}}\int_{[\sigma]} f(s)ds \text{log}(\int_{[\sigma]} f(s)ds)\]
By the mean-value theorem and continuity of $f$, for all $\sigma$ there is a $x_{\sigma}\in [\sigma]$ such that
\[\int_{[\sigma]} f(s)ds = 2^{-n}f(x_{\sigma})\]
So,
\[H(\rho_{n}) = -\sum_{\sigma \in 2^{n}}2^{-n}f(x_{\sigma}) \text{log}(2^{-n}f(x_{\sigma})) \]
\[=-\sum_{\sigma \in 2^{n}}2^{-n}f(x_{\sigma}) (-n+\text{log}(f(x_{\sigma}))) \]

\[=-\sum_{\sigma \in 2^{n}}2^{-n}f(x_{\sigma}) \text{log}(f(x_{\sigma})) + n \sum_{\sigma \in 2^{n}}2^{-n}f(x_{\sigma}) \]

\[=-\sum_{\sigma \in 2^{n}}2^{-n}f(x_{\sigma}) \text{log}(f(x_{\sigma})) + n \sum_{\sigma \in 2^{n}}\int_{[\sigma]} f(s)ds \]

\[=\text{Riemann Sum}[-f(.)\text{log}(f(.)), \text{Mesh Size}= 2^{-n}] + n \int_{0}^{1} f(s)ds \]
So, as the last integral is equal to 1,
\[H(\rho_n)-n = \text{Riemann Sum}[-f(.)\text{log}(f(.)), \text{Mesh Size}= 2^{-n}]\]

But, \[-\int_{0}^{1}f(s)\text{log}(f(s))ds = c,\]

for some constant $c$. Since the Reimann Sum converges to $c$, we have that $H(\rho_n)-n = O(1)$ as required.

\end{remark}
 
We now give an example to show that Theorem \ref{thm:9} cannot be reversed. The construction will be along the same lines as the preceding remark.
Let $f$ be any function on $(0,1)$ satisfying:
\[\int_{0}^{1}f(s)ds =1 \text{ and } -\int_{0}^{1}f(s)\text{log}(f(s))ds = - \infty.\]
For example, let \[f(x)= \dfrac{1}{x(1-\text{ln} x)^{2}}\] on $(0,1)$.
Define a diagonal state $\rho = (\rho_{n})_{n}$ as follows. Note that \[\int_{0}^{1}f(s)ds =1\]. Fix $n$. For all $\sigma \in 2^n$, let
\[\alpha_{\sigma} = \int_{[\sigma]} f(s)ds\]
Here, $[\sigma] \subset (0,1]$ is the open interval defined by the string $\sigma$. \[\rho_{n}= \sum_{\sigma \in 2^{n}} \alpha_{\sigma} |\sigma\big>\big<\sigma|.\]
$\rho$ is coherent since $\alpha_{\sigma}= \alpha_{\sigma 1} + \alpha_{\sigma 0}$ by definition. \begin{lem}
$\rho$ is q-MLR.
\end{lem}
\begin{proof}
Let $(G^m)_{m}$ be the universal q-MLT. Given a $\delta>0$ we find an $m$ such that $\rho(G^m) <\delta$. Since $f$ is in $L_1[0,1]$, by absolute continuity,  find a $m$ such that 
\[\int_{I} f(s)ds < \delta\] for any interval $I\subset (0,1)$, with $|I| \leq 2^{-m}$
We claim that this $m$ works by showing that for all $n$, $Tr(\rho_{n}G^{m}_{n})<\delta$. Fix an $n$. $G^{m}_{n}$ is an orthogonal projection with rank atmost $2^{n-m}$. By using the consequence of the SVD,
\[Tr(\rho_{n}G^{m}_{n})< \sum_{\sigma \in L} \alpha_{\sigma} = \sum_{\sigma \in L} \int_{[\sigma]} f(s)ds =\int_{E} f(s)ds < \delta \] where $L$ is the set of $\sigma$ s corresponding to the $2^{n-m}$ largest singular values of $\rho_{n}$, which are the $\alpha_{\sigma}$s. I.e., L = $\{\sigma_{1},.....\sigma_{2^{n-m}}\}$ if $\{\alpha_{\sigma_{1}},.....\alpha_{\sigma_{2^{n-m}}}\}$ are the first $2^{n-m}$ largest $\alpha$s. $E = \bigcup_{L}[\sigma]$ is a interval. Since $|L|= 2^{n-m}$ and $|[\sigma]|= 2^{-n}$, we have that $|E|= 2^{-n}2^{n-m}= 2^{-m}$.
\end{proof} 
\begin{lem}
For all $c \in \omega$, for almost every $n$, $H(\rho_n)-n <-c$. So, both the conditions of theorem \ref{thm:9} (the stronger and even the weaker one with `for infinitely many $n$') do not hold for $\rho$.
\end{lem}
\begin{proof}
For any $n$, by definition of the $\alpha$s, we have, \[H(\rho_{n}) = -\sum_{\sigma \in 2^{n}}\int_{[\sigma]} f(s)ds \text{log}(\int_{[\sigma]} f(s)ds)\]
By the mean-value theorem and continuity of $f$, for all $\sigma$ there is a $x_{\sigma}\in [\sigma]$ such that
\[\int_{[\sigma]} f(s)ds = 2^{-n}f(x_{\sigma})\]
So,
\[H(\rho_{n}) = -\sum_{\sigma \in 2^{n}}2^{-n}f(x_{\sigma}) \text{log}(2^{-n}f(x_{\sigma}))\]
\[=-\sum_{\sigma \in 2^{n}}2^{-n}f(x_{\sigma}) (-n+\text{log}(f(x_{\sigma})))\]

\[=-\sum_{\sigma \in 2^{n}}2^{-n}f(x_{\sigma}) \text{log}(f(x_{\sigma})) + n \sum_{\sigma \in 2^{n}}2^{-n}f(x_{\sigma})\]
By definition,
\[=-\sum_{\sigma \in 2^{n}}2^{-n}f(x_{\sigma}) \text{log}(f(x_{\sigma})) + n \sum_{\sigma \in 2^{n}}\int_{[\sigma]} f(s)ds \]

\[=\text{Riemann Sum}[-f(.)\text{log}(f(.), \text{Mesh Size}= 2^{-n}] + n \int_{0}^{1} f(s)ds \]
So, as the last integral is equal to 1,
\[H(\rho_n)-n = \text{Riemann Sum}[-f(.)\text{log}(f(.), \text{Mesh Size}= 2^{-n}].\]

But, \[-\int_{0}^{1}f(s)\text{log}(f(s))ds = -\infty.\]

So, 

\[\lim_{n\to\infty}H(\rho_n)-n =\lim_{n\to\infty} \text{Riemann Sum}[-f(.)\text{log}(f(.), \text{Mesh Size}= 2^{-n}]\]
\[= -\int_{0}^{1}f(s)\text{log}(f(s)ds = -\infty.\]

So, for all $c\in \omega$ there is an $N$ such that $n>N$ implies that $H(\rho_n)-n < -c $.

\end{proof}

So, $\rho$ is the required computable q-MLR.
\chapter{Open questions}\label{5}

An important open question is whether weak Solovay random states have a Levin\textendash Schnorr characterization in terms of $QK$. Techniques similar to those used in subsection \ref{subsect:weak} may prove to be useful in answering this.

It still remains to find a complexity based characterization of q-MLR states. In this direction, Nies and Scholz found a partial Miller\textendash Yu theorem concerning the quantum descriptive complexity of q-MLR and weak Solovay random states \cite{unpublished}. One can ask whether a Miller \textendash Yu type result holds for q-MLR and/or weak Solovay random states when using $QK$ as a complexity measure. 

An interesting question is whether weak Solovay randomness is equivalent to q-MLR, a positive answer to which will yield a $QK$ based characterizations (namely, those in Theorems \ref{thm:7} and \ref{thm:39}) of q-MLR.

Another interesting question is to find a von-Neumann entropy based characterization of q-MLR for computable states. Chapter \ref{4} contains partial results towards answering this question.
%
%
\bibliographystyle{siam}
\bibliography{thesis}

\end{document}